\documentclass{elsarticle}
\usepackage[T1]{fontenc}
\usepackage{amsmath}
\usepackage{amssymb}
\usepackage{amsthm}
\usepackage{mathtools}
\usepackage{subfigure}
\usepackage{csquotes}
\usepackage{hyperref}
\usepackage[basic]{complexity}
\usepackage{todonotes}
\usepackage{scrextend}
\usepackage{needspace}
\usepackage{verbatim}
\usepackage{float}
\usepackage{tabularx}
\usepackage{xcolor,colortbl}
\usepackage{booktabs}
\usepackage{thmtools}
\usepackage{thm-restate}
\usepackage{array, makecell}
\usepackage{amsfonts,bm}
\usepackage{comment}
\usepackage{url}
\usepackage[linesnumbered,ruled,vlined]{algorithm2e}
\usepackage{setspace}
\usepackage{tikz}
\usetikzlibrary{calc}
\usetikzlibrary{shapes.geometric}
\usepackage[sort&compress,nameinlink,noabbrev,capitalize]{cleveref}

\newcommand{\dfnc}[3]{#1:#2\rightarrow #3}
\newcommand{\DEF}[1]{{\em #1\/}}

\newcommand{\dset}[2]{\left\{#1 \:|\: #2\right\}}
\newcommand{\lset}[2]{\left\{#1, \ldots, #2\right\}}

\newcommand{\personaltodo}[4][noinline]{\todo[#1,color=#2]{#3: #4}}
\newcommand{\alex}[2][noinline]{\personaltodo[#1]{teal!30}{Al}{#2}}

\title{Graph drawing applications\\ in combinatorial theory of maturity models}

\author[fnm,imfm]{Špela Kajzer}
\author[vienna]{Alexander Dobler}
\author[fnm,kranj]{Janja Jerebic}
\author[vienna]{Martin Nöllenburg}
\author[graz]{Joachim~Orthaber}
\author[fnm,imfm]{Drago Bokal}

\affiliation[fnm]{organization={Faculty of Natural Sciences and Mathematics, University of Maribor}}
\affiliation[imfm]{organization={Institute of Mathematics, Physics, and Mechanics, Ljubljana, Slovenia}}
\affiliation[vienna]{organization={Algorithms and Complexity Group, TU Wien}}
\affiliation[kranj]{organization={Faculty of Organizational Sciences, University of Maribor}}
\affiliation[graz]{organization={Institute of Software Technology, Graz University of Technology, Austria}}

\newenvironment{wlist}
{\vspace{-10pt}
\begin{list}{}
{\setlength{\labelwidth}{15mm}
\setlength{\partopsep}{0pt}
\setlength{\parskip}{0pt}
\setlength{\topsep}{0pt}
\setlength{\itemsep}{0pt}
\setlength{\parsep}{0pt}
\setlength{\labelsep}{10pt}
\setlength{\leftmargin}{15mm}}
\item[]}
{\end{list}
\vspace{2pt}
\smallskip}

\newcommand{\probb}[6]{%
  \needspace{3\baselineskip}
  \begin{quote}
    \begin{labeling}{#6}%
    \item[#1]
    \item[\emph{#2}]#3
    \item[\emph{#4}]#5
    \end{labeling}%
  \end{quote}%
}

\newcommand{\probdef}[3]{\probb{#1}{Instance:}{#2}{Question:}{#3}{as}}

\newcommand{\panelcrossmin}{\ensuremath{\textsc{Panel Crossing Minimization}}}

\newcommand{\cC}{\ensuremath{\mathcal{C}}}

\newenvironment{proof*}[1]
  {\renewcommand{\proofname}{#1}%
   \begin{proof}}
  {\end{proof}}

\theoremstyle{plain}
\newtheorem{theorem}{Theorem}[section]
\newtheorem{lemma}[theorem]{Lemma}
\newtheorem{corollary}[theorem]{Corollary}

\newtheorem{proposition}[theorem]{Proposition}
\newtheorem{open}[theorem]{Open problem}

\theoremstyle{definition}
\newtheorem{definition}[theorem]{Definition}

\begin{document}

\begin{abstract}
In this paper, we introduce tiled graphs as models 
of learning and maturing processes. 
We show how tiled graphs can combine graphs of 
learning spaces or antimatroids (partial hypercubes) and maturity models 
(total orders) to yield models of learning processes. 
For the visualization of these processes it is a natural approach to aim for certain optimal drawings.
We show for most of the more detailed models that the drawing problems resulting from them are \mbox{\NP-}complete.
The terse model of a maturing process that ignores 
the details of learning, however, results in 
a polynomially solvable graph drawing problem. 
In addition, this model provides insight into 
the process by ordering the subjects at each 
test of their maturity. We investigate 
extremal and random instances of this problem, 
and provide exact results and bounds on 
their optimal crossing number. 

Graph-theoretic models offer two approaches to the design of optimal maturity models 
given observed data: (1) minimizing intra-subject inconsistencies, which manifest as regressions of subjects, 
is modeled as the well-known feedback arc set problem.
We study the alternative of (2) finding a maturity model 
by minimizing the inter-subject inconsistencies, which manifest as crossings in the respective drawing. 
We show this to be \NP-complete.
\end{abstract}

\begin{keyword}Maturity models\sep Learning space\sep Crossing minimization\sep Tile crossing number
\end{keyword}

\maketitle

\section{Introduction}


Maturity models have been used for decades to track progress over time made by some entities, called subjects, with respect to some linearly ordered set of stages or, hereafter, categories.
The subjects can, for example, be technologies \cite{H2020, NASA}, products \cite{SINNWELL2019557}, organizations \cite{crawford2021project, SPICE}, or people~\cite{nonaka1994dynamic}.
For instance, NASA defined the technology readiness levels (TRLs)~\cite{NASA}, which were later adapted by the European Commission for Horizon 2020 and Horizon Europe projects~\cite{H2020}.
This 9-level scale shows the progression in the development of new technologies, starting at the first level with the observation of basic principles, and ending at the ninth level with the technology being successfully applied in its operational environment.
The TRL scale thus presents a maturity model for knowledge that any researcher in applied sciences will at least implicitly come in contact with. 

The above are just 
some
examples, but in general maturity models have been applied to a variety of scientific and practical contexts~\cite{Gottschalk, CMM,  SPICE, SINNWELL2019557}. 
They are also called stages-of-growth models, stage models, or stage theories~\cite{Prananto}.
The models are based on the assumption that the growth patterns of the observed subjects are predictable, that is, the models describe an assumed, desired, or expected stage-by-stage growth of those subjects~\cite{Gottschalk,poppelbuss2011makes}. 
With the growing interest in maturity models, they attracted both scientific attention and criticism~\cite{Bach,Mettler} asking (i) for an improvement in the understanding of maturity models in general and (ii) for the creation of an optimal maturity model when given specific observed data.
We address both these challenges in our contribution.
Our key observation is that longitudinal studies of the mentioned growth patterns produce ordinal panel data \cite{JerebicKajzerBokalOpda}.
This allows us to address challenge (i) by proposing a more detailed model that formalizes maturity models as rankings of knowledge states in a learning space (cf.\ \cref{sec:graphmodels}).
Then we address challenge (ii) by optimizing these rankings, as presented in \cref{sec:optimize}.

Our formalization
of maturity models and the learning process behind them builds on 
learning spaces~\cite{falmagne2010learning}, 
also known as antimatroids~\cite{doignon2014learning, eppstein2007media}.  
These are exhaustively studied combinatorial structures, 
whose graphs are partial hypercubes~\cite{eppstein2007media}.
We propose several models of subject progress in
learning. The detailed models result in $k$-tiled graphs 
\cite{vegi2021counting, vegi2023counting} known 
from understanding crossing critical graphs
\cite{bokal2015degree,bokal2021properties,pinontoan2003crossing}.

In later sections, we show that for several relevant models it is in general \NP-complete to obtain their optimal (tile) drawings. 
This motivates the investigation of simplest relevant instances. These reduce the observability of learning progress to the stages of a maturity model at prescribed timestamps or tests and ignore the detailed learning progress between the two tests.
We show that for a given maturity model, the optimal ordinal panel data drawing of its longitudinal study can be obtained in polynomial time.
This motivates introducing the notion of the \emph{panel crossing number} as the minimum number of crossings required to consistently represent trajectories of subjects through a series of tests that assign each of them an ordinal variable (category).
While the ordering of subjects is partially prescribed by the ordinal variable, the intra-category ordering of subjects (i.e., the ordering of the subjects that are assigned the same category in a given test) is not prescribed and defines a degree of freedom allowing for crossing minimization. 
The obtained minimum number of crossings is the panel crossing number of an ordinal panel data instance.
A drawing realizing this crossing number exhibits most consistent, least turbulent progress of subjects through the sequence of tests.
This observation demonstrates the potential for applications of graph drawing techniques in (discrete) data analysis.

The crossing minimization problem that corresponds to optimizing 
the inter-subject consistency in maturity models is closely related to 
variants of crossing minimization in layered graph drawings~\cite{hn-hda-13,stt-mvuhss-81}. However, 
unlike general layered or hierarchical graphs, the simplest representations
of maturity models that ignore the inter-test learning process (i.e., the details of a learning process of each subject in a given test) feature 
a set of subjects that form $x$-monotone, potentially crossing 
paths over time. So the resulting graphs are basically collections of 
intertwined paths. Without the ordered categories, this is very similar 
to storyline layouts~\cite{GronemannJLM16,KostitsynaNP0S15,tm-dcosv-12} 
or metro line crossing minimization problems~\cite{ArgyriouBKS10,bnuw-miecw-07,
fp-mcmhatc-13,n-iamlc-10}. Yet, the fact that each subject belongs 
to exactly one category at each time point and categories are 
linearly ordered puts much more constraints on the feasible permutations of subjects in each step. Hence, we investigate the specific crossing 
properties of such constrained maturity model instances of ordered 
panel data and optimize the corresponding maturity model 
visualizations.

In addition to solving the panel crossing minimization problem in polynomial time, we analyze extremal instances where the panel crossing number is maximal.
We continue with the panel crossing number of random instances, and we solve the question of finding an optimal maturity model, i.e., a model that for given data allows for a representation with the smallest panel crossing number.
We show that this problem is \NP-complete, but integer linear programming (ILP) models such as the one presented in \cref{sec:optimize} can solve practical instances to optimality. 
In fact, it may even be \NP-complete just asking  whether the optimal maturity model results in a planar instance with no crossings among subjects.

The remainder of the paper is structured as follows.
We start by describing the graph models of the learning and testing processes and how they yield the tile graphs and we define mathematical prerequisites along the way (\cref{sec:graphmodels}).
This already contributes to the first aforementioned challenge of improving the understanding of the maturity models.
We proceed by exhibiting \NP-completeness of crossing minimization problems 
of drawings respecting maturity models related to all but the simplest of the defined models (\cref{sc:NPcompleteness}).
Then we give a polynomial time algorithm that produces optimal drawings of said graph model of the learning and testing processes (\cref{sc:polynomial}).
Next, we study extremal (\cref{sec:extremal}) and random instances (\cref{sec:expected}).
As a final result of the paper, we establish \NP-completeness 
for the problem of producing an optimal maturity model minimizing inconsistencies between maturing subjects for some given data 
(\cref{sec:optimize}).
Thus, we contribute to the second aforementioned challenge of producing optimal maturity models for the observed data.

\section{Learning graphs and their applications}\label{sec:graphmodels}

In this section, we address challenge (i) of getting a better understanding of maturity models. In particular, using tiles introduced by Pinontoan and Richter~\cite{pinontoan2003crossing}, we link maturity models to another well-studied combinatorial structure, so-called learning spaces~\cite{falmagne2011learning}. In the following we reproduce the core definitions.

A \textit{knowledge structure} is a pair $(Q,{\cal K})$, in which $Q$ is a nonempty set and $\cal K$ is a family of subsets of $Q$ containing at least $Q$ and the empty set $\emptyset$. The set $Q$ is called the \textit{domain} of the knowledge structure, its elements are called \textit{knowledge items}, and the elements of $\cal K$ are called \textit{knowledge states}. Since $Q$ is always the largest set in $\cal K$, it can also be omitted when discussing a knowledge structure.

A knowledge structure $(Q,{\cal K})$ is a \DEF{learning space} if it satisfies two axioms, the axiom of \textit{learning smoothness}, intuitively stating that if a state $K$ is a subset of a state~$L$, then a learner can reach $L$ from $K$ by learning one item at a time, and the axiom of \textit{learning consistency}, intuitively stating that knowing more does not prevent the learner from learning something new. 

Formally, learning smoothness stipulates that for every pair of states $K,L\in{\cal K}$ with $K\subsetneq L$, there exists a finite chain of states
\begin{equation*}
K=K_0\subsetneq K_1\subsetneq\ldots\subsetneq K_p=L,
\end{equation*}
such that for all $1 \le i \le p$ we have $|K_i\setminus K_{i-1}|=1$, implying that $|L\setminus K|=p$.
Learning consistency stipulates that if $K\subseteq L$ are two states and $q$ is a knowledge item such that $K\cup\{q\} \in {\cal K}$ holds, then also $L\cup \{q\}\in {\cal K}$ holds.

\cref{fig:ex-learningSpaceGraph} shows an example of a graph drawing of a learning space. Vertices of the graph represent states, while edges represent knowledge items.
\begin{figure}[htb]
    \centering
    \scalebox{.7}{\begin{tikzpicture}
\begin{scope}
\node [circle, fill=black, inner sep=1pt] (v2) at (0,0) {};
\node [circle, fill=black, inner sep=1pt] (v3) at (2,0) {};
\node [circle, fill=black, inner sep=1pt] (v5) at (0,2) {};
\node [circle, fill=black, inner sep=1pt](v4) at (2,2) {};
\node [circle, fill=black, inner sep=1pt](v1) at (0,-2) {};
\node [circle, fill=black, inner sep=1pt] (v7) at (1,1) {};
\node [circle, fill=black, inner sep=1pt] (v6) at (3,1) {};
\node [circle, fill=black, inner sep=1pt](v9) at (3,3) {};
\node [circle, fill=black, inner sep=1pt] (v8) at (1,3) {};
\draw  (v1) -- (v2) node [pos=0.50, sloped, above]{ Set theory};
\draw  (v2) -- (v3)node [pos=0.50, sloped, below]{Knowledge}node [pos=0.50, yshift=-8.5, sloped, below]{structure};
\draw  (v3)-- (v4);
\draw  (v5) edge (v4);
\draw  (v5) -- (v2)node [pos=0.50, rotate=180, sloped, above]{Medium};
\draw  (v3) edge (v6);
\draw  (v2) -- (v7)node [pos=0.6, sloped, above]{Partial}node[pos=0.50, yshift=4.5, xshift=4.5, sloped, below]{cubes};
\draw  (v7) edge (v6);l
\draw  (v8) edge (v7);
\draw  (v5) edge (v8);
\draw  (v8) edge (v9);
\draw  (v9) edge (v6);
\draw  (v4) edge (v9);
\node [circle, fill=black, inner sep=1pt] (v12) at (4,0) {};
\node [circle, fill=black, inner sep=1pt] (v13) at (5,1) {};
\node [circle, fill=black, inner sep=1pt] (v11) at (5,3) {};
\node [circle, fill=black, inner sep=1pt] (v10) at (4,2) {};
\draw  (v4) edge (v10);
\draw  (v10) edge (v11);
\draw  (v9) edge (v11);
\draw  (v3) -- (v12)node [pos=0.50, sloped, below]{Knowledge}node [pos=0.50, yshift=-8.5, sloped, below]{space};
\draw  (v12) edge (v10);
\draw  (v12) edge (v13);
\draw  (v11) edge (v13);
\draw  (v6) edge (v13);
\node [circle, fill=black, inner sep=1pt] (v14) at (6,0) {};
\node [circle, fill=black, inner sep=1pt] (v15) at (7,1) {};
\node [circle, fill=black, inner sep=1pt] (v16) at (6,2) {};
\node [circle, fill=black, inner sep=1pt] (v17) at (7,3) {};
\draw  (v12) -- (v14)node [pos=0.50, sloped, below]{Learning}node [pos=0.50, yshift=-8.5, sloped, below]{space};
\draw  (v14) edge (v15);
\draw  (v13) edge (v15);
\draw  (v16) edge (v17);
\draw  (v11) edge (v17);
\draw  (v10) edge (v16);
\draw  (v16) edge (v14);
\draw  (v17) edge (v15);
\node [circle, fill=black, inner sep=1pt] (v19) at (7,5) {};
\node [circle, fill=black, inner sep=1pt] (v20) at (9,5) {};
\node [circle, fill=black, inner sep=1pt] (v18) at (9,3) {};
\node [circle, fill=black, inner sep=1pt] (v22) at (11,5) {};
\node [circle, fill=black, inner sep=1pt] (v21) at (11,3) {};
\node [circle, fill=black, inner sep=1pt] (v23) at (13,3) {};
\node [circle, fill=black, inner sep=1pt] (v24) at (13,5) {};
\draw  (v17) -- (v18)node [pos=0.50, sloped, below]{\small Graph rep.}node [pos=0.50, yshift=-8.5, sloped, below]{\small of a medium};
\draw  (v17) -- (v19)node [pos=0.50, yshift=8, sloped, above]{\small Learning sp.}node [pos=0.50, sloped, above]{\small of a medium};
\draw  (v19) edge (v20);
\draw  (v20) edge (v18);
\draw  (v18) -- (v21)node [pos=0.50, sloped, below]{\small Graph of}node [pos=0.50, yshift=-8.5, sloped, below]{\small a medium};
\draw  (v21) edge (v22);
\draw  (v22) edge (v20);
\draw  (v21) -- (v23)node [pos=0.50, sloped, below]{\footnotesize Connection btw.}node [pos=0.50, yshift=-8.5, sloped, below]{\footnotesize graphs of med.}node [pos=0.50, yshift=-17, sloped, below]{\footnotesize and partial cubes};
\draw  (v23) edge (v24);
\draw  (v24) edge (v22);
\end{scope}
\end{tikzpicture}}
    \caption{Example of a graph drawing of a learning space.}
    \label{fig:ex-learningSpaceGraph}
\end{figure}
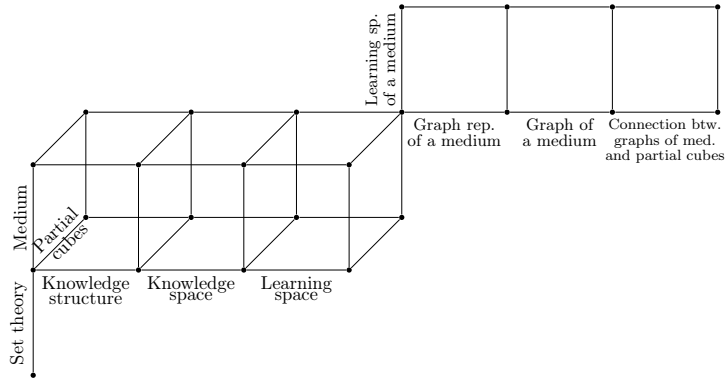

We can observe that learning smoothness, by applying it to $K=\emptyset$ and \mbox{$L=Q$}, implies finiteness of the learning space $(Q,{\cal K})$. 
It has also been mentioned in the literature that the definition of a learning space is equivalent to the definition of an 
antimatroid~\cite{korte2012greedoids}. 

For a learning space $(Q,{\cal K})$, there is a natural definition of the \textit{graph of the learning space} $G(Q,{\cal K})$: Its vertices are all the knowledge states in $\cal K$ and two knowledge states $K \subsetneq L$ are adjacent if and only if there exists a knowledge item $q \in Q$ such that $K\cup\{q\} = L$. It is clear that this graph is a \DEF{partial hypercube}, that is, a graph whose vertices can be labeled by $\{0,1\}$-bitstrings such that adjacent vertices differ only in a single bit.

Having reproduced the mathematical model of learning a set of knowledge items, we observe that this model allows for a significant level of detail in representing the current knowledge state of a learner. For practical purposes, the level of detail is often significantly simplified, such as in educational systems or in career rankings. There, a knowledge state is usually simplified into an ordinal variable such as a grade, school level/year, or career badge. For this reason, we continue by describing an \textit{ordinal panel data instance} as the underlying model of the data collected in such a (simplified) model of observing learning progress.

A panel, or longitudinal, data set follows a given group of subjects over time. It thus provides, for each subject in the group, multiple observations of the same set of variables at different timestamps \cite{hsiao2014analysis}. In the case of maturity models we obtain, for a given group of subjects, observations of one ordinal variable: the maturity level. Combined with in-depth assessments that result in this maturity level, we periodically observe the competencies, skills, or other knowledge items of subjects -- their knowledge state -- and use this information to identify the maturity level.

To define panel data, we make use of a \emph{permutation} or \emph{ordering} $\pi$ of a set~$X$, which is a (total) linear order on $X$. We write $x\prec_{\pi}x'$ if $x,x'\in X$ and $x$ comes before $x'$ according to $\pi$. Let further $\Pi(X)$ be the set of all permutations of the set $X$.
Further, for $Y\subseteq X$ we define the induced permutation $\pi[Y]$ of the set $X$ such that for $y,y'\in Y$, $y\prec_{\pi}y'$ if and only if $y\prec_{\pi[Y]}y'$.

\begin{definition}
A \emph{panel data (PD) instance} is a triple 
$(S,\mathcal{C},T)$, 
where $S$, $\mathcal{C}$, and $T$, respectively, 
are a set of subjects, categories, and 
linearly ordered time\-stamps, respectively. 
The timestamps $T=\{t_0,t_1,\dots,t_m\}$, $m > 0$, are ordered
increasingly by indices.
For each $i=1,\dots,m$, let $t_i:S\rightarrow \cal C$ 
be a test function that assigns a category 
to each subject at timestamp $t_i$.

An \DEF{ordinal panel data (OPD) instance} 
$(S,\mathcal{C},T,\sigma)$ is a PD instance 
$(S,\mathcal{C},T)$ with an 
additional linear ordering $\sigma$ of the 
categories $\mathcal{C}$.

A \DEF{combinatorial layout} of an OPD instance 
$(S,\mathcal{C},T,\sigma)$ is a sequence of 
permutations $\pi_1,\pi_2,\dots,\pi_m\in \Pi(S)$ such that 
    \[\forall t_i\in T,\:\forall s,s'\in S:
    t_i(s)\prec_{\sigma}t_i(s') \implies s\prec_{\pi_{i}}s'.\]
    Informally, $\pi_i$ is an ordering of the subjects 
    such that if the category $t_i(s)$ is before 
    $t_i(s')$ in $\sigma$, then $s$ is 
    before $s'$ in $\pi_i$. 
\end{definition}

We slightly abuse notation and use the same notation 
$t\in T$ for both the timestamp and for the function
assigning the categories to the subjects. To emphasize 
the difference in the discussion, we refer to the function
as a test.

Unless stated otherwise, 
$n:=|S|$ denotes the number of subjects,
$m:=|T|-1$ the number of intervals between timestamps, 
and $k:=|\mathcal{C}|$ the number of categories. 
Note that we are predominantly interested in the 
behavior of subjects between timestamps.
For this reason, we label the first timestamp $t_0$,
as then the number of intervals is equal to the last index 
of a timestamp in the sequence, and the interval between
timestamps gets assigned the index of the later timestamp.

The context is now established to link the maturity models
with learning spaces. The levels of maturity, defined in 
specific maturity models, are defined by competencies
required by the subjects at a certain maturity. Those
competencies may depend on additional knowledge items, 
and we can define the set $Q$ to contain 
all the competencies
and all the other skills required for a subject to be in any
of the maturity levels of the maturity model. 
As the learning
of the competencies can progress one 
knowledge item at a time,
and assuming the knowledge items are consistent, there is a
learning space $(Q,{\cal K})$ that models the learning
process of the maturity model in greater detail than the 
maturity model itself. Let
$(S,\mathcal{C},T,\sigma)$ 
be an ordinal panel data instance,
such that the subjects in $S$ learn the knowledge items
$Q$ in the learning space $(Q,\mathcal{K})$, whose maturity
levels are the categories in $\mathcal{C}$. Each test in $T$
then checks knowledge of the items in $Q$ and assigns a state from $\mathcal{K}$ to a
subject in $S$. This information
is more detailed than assigning a category, and there is 
a \DEF{ranking function} 
$\dfnc{\alpha}{\mathcal{K}}{\mathcal{C}}$
that assigns to each knowledge state in $\mathcal{K}$
a maturity level, i.e., a category in $\mathcal{C}$.
To simplify notation, we will not distinguish between
tests that assign subjects the categories or tests that
assign knowledge states. If needed, we will implicitly 
assume existence of a ranking function, linking the two 
interpretations of tests.

We are now ready for defining a visual representation of the 
data collected in the model. The central concept
in this representation are tiles, a concept introduced
by Pinontoan and Richter while studying 
crossing-critical graphs~\cite{pinontoan2003crossing}.

\begin{definition}[\cite{pinontoan2003crossing}]
Let $G$ be a graph and let $L = (\lambda_0, \lambda_1,\ldots, 
    \lambda_l)$ (called \emph{left wall}) and $R = (\rho_0, \rho_1, \ldots, \rho_r)$ 
    (called \emph{right wall}) be two sequences of distinct vertices of $G$, 
    such that no vertex of $G$ appears in both. The triple $(G, L, R)$ 
    is called a \DEF{tile}. A vertex of $G$ that belongs to neither wall
    is called an \DEF{internal vertex}.
    
 A \DEF{tile drawing} of a tile $T=(G, L, R)$ is a drawing of $G$ 
    in the unit square $[0, 1] \times [0, 1]$ that meets the boundary 
    of the square precisely in the vertices of $L \cup R$ so that 
    the vertices of $L$ have $x$-coordinate $0$, with the 
    $y$-coordinates of $\lambda_0,\lambda_1, \ldots, \lambda_l$ 
    strictly decreasing, and the vertices of $R$ have $x$-coordinate $1$, 
    with the $y$-coordinates of $\rho_0, \rho_1, \ldots, \rho_r$ strictly decreasing.
    \begin{figure}[H]
        \centering
        \begin{tikzpicture}

\node [circle, inner sep=0.5pt, fill=black, label=below:{\tiny $(0,0)$}] (v1) at (0,0) {};
\node [circle, inner sep=0.5pt, fill=black, label=below:{\tiny $(1,0)$}] (v2) at (4,0) {};
\node [circle, inner sep=0.5pt, fill=black, label=above:{\tiny $(0,1)$}] (v4) at (0,4) {};
\node [circle, inner sep=0.5pt, fill=black, label=above:{\tiny $(1,1)$}] (v3) at (4,4) {};
\draw [color=black] (v1) edge (v2);
\draw [color=black] (v3) edge (v2);
\draw [color=black] (v4) edge (v3);
\draw [color=black] (v4) edge (v1);
\node [circle, inner sep=1pt, fill=black] (v9) at (0,0.4) {};
\node [circle, inner sep=1pt, fill=black] (v11) at (0,1) {};
\node [circle, inner sep=1pt, fill=black] (v7) at (0,3) {};
\node [circle, inner sep=1pt, fill=black] (v5) at (0,3.6) {};
\node [circle, inner sep=1pt, fill=black] (v14) at (4,0.4) {};
\node [circle, inner sep=1pt, fill=black] (v23) at (4,1.2) {};
\node [circle, inner sep=1pt, fill=black] (v13) at (3.8,0.8) {};
\node [circle, inner sep=1pt, fill=black] (v12) at (3.6,0.6) {};
\node [circle, inner sep=1pt, fill=black] (v15) at (3.4,1.3) {};
\node [circle, inner sep=1pt, fill=black] (v21) at (4,2.5) {};
\node [circle, inner sep=1pt, fill=black] (v18) at (4,3.3) {};
\node [circle, inner sep=1pt, fill=black] (v19) at (3.4,3.1) {};
\node [circle, inner sep=1pt, fill=black] (v20) at (3.7,2.8) {};
\node [circle, inner sep=1pt, fill=black] (v8) at (0.6,3.1) {};
\node [circle, inner sep=1pt, fill=black] (v6) at (0.6,3.4) {};
\node [circle, inner sep=1pt, fill=black] (v25) at (0.9,2.6) {};
\node [circle, inner sep=1pt, fill=black] (v10) at (0.6,0.6) {};
\draw [thick,color=black] (v5) edge (v6);
\draw [thick,color=black] (v6) edge (v7);
\draw [thick,color=black] (v7) edge (v8);
\draw [thick,color=black] (v9) edge (v10);
\draw [thick,color=black] (v11) edge (v10);
\draw [thick,color=black] (v12) edge (v13);
\draw [thick,color=black] (v12) edge (v14);
\node [circle, inner sep=1pt, fill=black] (v26) at (2.3,2.3) {};
\draw [color=black] plot[smooth, tension=.7] coordinates {(v14)};
\node [circle, inner sep=1pt, fill=black] (v16) at (3.2,2) {};
\node [circle, inner sep=1pt, fill=black] (v17) at (2.7,2.8) {};
\node [circle, inner sep=1pt, fill=black] (v22) at (1.6,3.1) {};

\draw [thick,color=black] (v18) edge (v19);
\draw [thick,color=black] (v19) edge (v20);
\draw [thick,color=black] (v20) edge (v21);
\draw [thick,color=black] (v18) edge (v20);
\draw [thick,color=black] (v20) edge (v17);

\draw [color=black] plot[smooth, tension=.7] coordinates {(v8) (v22) (v17) (v16) (v15) (v13)};

\draw [thick,color=black] (v23) edge (v15);
\node [circle, inner sep=1pt, fill=black] (v24) at (2.1,1.5) {};
\draw [thick,color=black] plot[smooth, tension=.7] coordinates {(v15) (v24) (v25) (v7)};
\draw [thick, color=black] (v11) edge (v26);
\draw [thick,color=black] (v17) edge (v26);
\draw [thick,color=black] plot[smooth, tension=.7] coordinates {(v10) (2.1,1.5) (v26)};
\node [circle, inner sep=1pt, fill=black] at (2.1818,3) {};
\end{tikzpicture}
        \caption{Example of a tile drawing.}
        \label{fig:eg-tile}
    \end{figure}
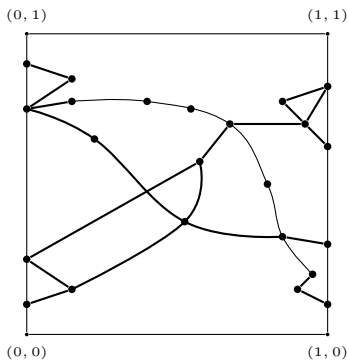
    
The \DEF{tile crossing number} $tcr(T)$ of a tile $T$ is 
    the minimum number of edge crossings over all tile drawings of $T$.
    
    A tile $(G,L,R)$ is \DEF{compatible} with a tile $(G',L',R')$ 
    if $|R|=|L'|$. A sequence of tiles $(T_0,T_1,\ldots,T_m)$ 
    is \emph{compatible} if $T_i$ is compatible with $T_{i+1}$ 
    for $i=0,1,\ldots,m-1$.
    
    The \DEF{join} of two compatible tiles $(G,L,R)$ and $(G',L',R')$ is defined as $(G,L,R)\otimes (G',L',R')=(G\otimes G',L,R')$, where $G\otimes G'$ is the graph obtained from the disjoint union of $G$ and $G'$ by identifying $\rho_i$ with $\lambda'_i$ for $i=0,1,\ldots,|R|$.
Since this operation is associative, we can define the join of a compatible sequence of tiles $(T_0,T_1,\ldots,T_m)$ as $\otimes(T_0,T_1,\ldots,T_m) =T_0\otimes T_1\otimes\cdots\otimes T_m$ which is a tile $(G_0\otimes G_1\otimes\cdots\otimes G_n,L_0,R_m)$. 
\begin{figure}[H]
    \centering
    \input{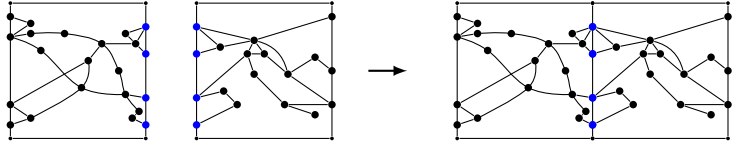}
    \caption{Two non-trivial compatible tiles and their join.}
    \label{fig:join}
\end{figure}

Note that in a join of a compatible sequence of tiles, 
the intermediate walls may be of relevance once the
tile is joined. In such a case, we keep track of those 
walls by emphasizing $(G,L,R)\otimes (G',L',R')=
(G\otimes G',L,L',R')$ and $\otimes(T_0,T_1,\ldots,T_m) = 
(G_0\otimes G_1\otimes\cdots\otimes G_n,L_0,L_1,\ldots,L_m,R_m)$.       
\end{definition} This is an extension of a commonly used notation.
As the two walls are identified in the join,
an equivalent notation in the previous definition could be
$(G\otimes G',L,L',R')$ and $(G_0\otimes G_1\otimes
\cdots\otimes G_n,L_0,R_1,\ldots,R_{m-1},R_m)$ or any combination of the two.

We continue with an illustrative example of the introduced concepts.
First, Figure \ref{fig:pctA} presents a graph of a learning
space of the learning space theory, as reproduced from 
\cite{BokalJerebicIntro}. The edges correspond to 
mastering new concepts introduced in learning space theory,
and the vertices correspond to knowlege states. A ranking function of the maturity model
assigned to this learning space is represented by colors and shapes of the 
graph's vertices, with maturity increasing in the sequence of
green squares, blue circles, yellow diamonds, red dots.

\begin{figure}[H]
    \centering
\usetikzlibrary{matrix}
\usetikzlibrary{shadows}
\usetikzlibrary{shapes.geometric}
\usetikzlibrary{calc}
\begin{tikzpicture}
\begin{scope}[scale=0.25]
\coordinate(a0) at (0,-4) {}{};
\coordinate (a1) at (0,0) {} {};
\coordinate (a2) at (4,0) {}{};
\coordinate (a3) at (4,4) {}{};
\coordinate (a4) at (0,4) {}{};
\coordinate (a5) at ($({sqrt(2)}, {sqrt(2)})$) {}{};
\coordinate (a6) at ($({4+sqrt(2)}, {sqrt(2)})$) {}{};
\coordinate (a7) at ($({sqrt(2)+4}, {sqrt(2)+4})$) {}{};
\coordinate (a8) at ($({sqrt(2)}, {sqrt(2)+4})$) {}{};
\coordinate (a9) at (8,0) {} {};
\coordinate (a10) at (12,0) {}{};
\coordinate (a11) at (8,4) {}{};
\coordinate (a12) at (12,4) {}{};
\coordinate (a13) at ($({sqrt(2)+8}, {sqrt(2)})$) {}{};
\coordinate (a14) at ($({12+sqrt(2)}, {sqrt(2)})$) {}{};
\coordinate (a15) at ($({sqrt(2)+8}, {sqrt(2)+4})$) {}{};
\coordinate (a16) at ($({sqrt(2)+12}, {sqrt(2)+4})$) {}{};
\draw (a0)-- (a1);
\draw (a1) -- (a2) -- (a3) -- (a4)--cycle;
\draw (a9) -- (a10) -- (a12) -- (a11)--cycle;
\draw (a13) -- (a14) -- (a16) -- (a15)--cycle;
\draw (a5)--(a8);
\draw (a7)--(a6);
\draw (a6)--(a5);
\draw (a1) -- (a5);
\draw (a2) -- (a6);
\draw (a3) -- (a7);
\draw (a4) -- (a8);

\draw(a3)--(a11);
\draw (a2)--(a9);
\draw (a6)--(a13);
\draw (a13)--(a9);
\draw (a14)--(a10);
\draw (a8)--(a15);
\draw (a15)--(a11);
\draw (a16)--(a12);
\coordinate (b1) at ($({sqrt(2)+12}, {sqrt(2)+8})$) {}{};
\coordinate (b2) at ($({sqrt(2)+12}, {sqrt(2)+4})$) {}{};
\coordinate (b3) at ($({sqrt(2)+16}, {sqrt(2)+8})$) {}{};
\coordinate (b4) at ($({sqrt(2)+16}, {sqrt(2)+4})$) {}{};
\coordinate (b5) at ($({sqrt(2)+20}, {sqrt(2)+8})$) {}{};
\coordinate (b6) at ($({sqrt(2)+20}, {sqrt(2)+4})$) {}{};
\coordinate (b7) at ($({sqrt(2)+24}, {sqrt(2)+8})$) {}{};
\coordinate (b8) at ($({sqrt(2)+24}, {sqrt(2)+4})$) {}{};

\draw (b1)--(b2)--(b4)--(b3)--cycle;
\draw (b5)--(b6)--(b8)--(b7)--cycle;
\draw(b3)--(b5);
\draw(b4)--(b6);
\node [rectangle, fill=black!40!green,  inner sep=2pt] at (a0) {};
\node [rectangle, fill=black!40!green, inner sep=2pt] at (a1) {};
\node [circle,  draw, very thick, color=blue, inner sep=1.5pt] at (a2) {};
\node [circle,  draw, very thick, color=blue, inner sep=1.5pt] at (a3) {};
\node [circle,  draw, very thick, color=blue, inner sep=1.5pt] at (a4) {};
\node [circle,  draw, very thick, color=blue, inner sep=1.5pt] at (a5) {};
\node [circle,  draw, very thick, color=blue, inner sep=1.5pt] at (a6) {};
\node [circle,  draw, very thick, color=blue, inner sep=1.5pt] at (a7) {};
\node [circle,  draw, very thick, color=blue, inner sep=1.5pt] at (a8) {};
\node [circle,  draw, very thick, color=blue, inner sep=1.5pt] at (a9) {};
\node [diamond, fill=black!20!yellow, inner sep=2pt] at (a10) {};
\node [circle,  draw, very thick, color=blue, inner sep=1.5pt] at (a11) {};
\node [diamond, fill=black!20!yellow, inner sep=2pt] at (a12) {};
\node [circle,  draw, very thick, color=blue, inner sep=1.5pt] at (a13) {};
\node [diamond, fill=black!20!yellow, inner sep=2pt] at (a14) {};
\node [circle,  draw, very thick, color=blue, inner sep=1.5pt] at (a15) {};
\node [diamond, fill=black!20!yellow, inner sep=2pt] at (a16) {};
\node [circle,  fill=black!20!red, inner sep=1.5pt] at (b1) {};

\node [circle, fill=black!20!red, inner sep=1.5pt] at (b3) {};
\node [circle, fill=black!20!red, inner sep=1.5pt] at (b4) {};
\node [circle, fill=black!20!red, inner sep=1.5pt] at (b5) {};
\node [circle, fill=black!20!red, inner sep=1.5pt] at (b6) {};
\node [circle, fill=black!20!red, inner sep=1.5pt] at (b7) {};
\node [circle, fill=black!20!red, inner sep=1.5pt] at (b8) {};
\end{scope}

\end{tikzpicture}
    \caption{Graph drawing of a learning space together with a ranking function assigning four maturity levels to knowledge states.}
    \label{fig:pctA}
\end{figure}
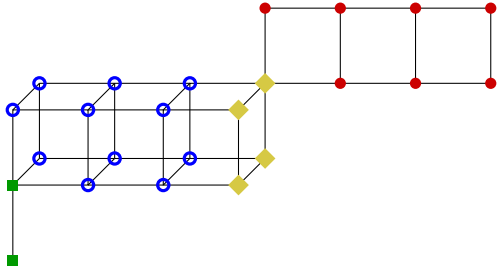

In a detailed model of a learning process, a test $t\in T$ 
assigns a subject $s\in S$ a knowledge state $k\in\mathcal{K}$ such that $t(s)=k$
and then the ranking function $\dfnc{\alpha}{\mathcal{K}}{\mathcal{C}}$
assigns a category $C\in\mathcal{C}$ to the subject $s$ at test $t$, such that $\alpha(t(s))=\alpha(k)=C$.
After the test, the subject proceeds its walk in the learning space,
presumably starting from the same knowledge state the test has revealed,
and reaching the state the next test will reveal. 
In this context, the progress of the subjects
in learning is represented by a walk in the graph of the learning space.
The first model of this repeated walk 
between two tests we propose
is a tile whose wall vertices represent categories and 
whose internal vertices represent the whole learning space. 
The path of the subject starts at a left wall vertex -- 
representing the category assigned by 
the test at the start of the interval, 
follows the vertices of the learning space, 
and concludes in a right wall vertex -- 
representing the category assigned by the test following the interval.
A wall vertex is adjacent to an internal vertex if
the ranking function 
maps that internal vertex
to the category corresponding to the wall vertex.
A join of $m$ such tiles constitutes the \emph{total learning model}: 
it models all possible paths the subjects with certain 
test scores could take to traverse the learning space
as progressing according to categories of the maturity model.
The graph has the property that each tile is the same,
as it allows all subjects to reside 
in any category at any test, and the number of wall vertices is equal 
the number of categories. 
The following formal definition of a 
\DEF{total learning tile} is illustrated in Figure \ref{fig:pctB}.
\begin{definition}
    Let $(S, \mathcal{C}, T, \sigma)$ be an OPD instance, 
    $(Q, \cal{K})$ a given learning space and let 
    $G(Q, \cal{K})$ be the graph of $(Q, \cal{K})$. 
    Let $\alpha : \cal{K} \rightarrow \mathcal{C}$ be a ranking 
    function that assigns a category to each knowledge state 
    $\cal{K}$. For each pair of consecutive tests $t', t \in T$, 
    we define a graph $G_L(t', t)$ with vertices 
    
    $$V(G_L(t',t)):=\dset{(t',C)}{C \in \mathcal{C}} \cup 
    \dset{(t, C)}{C \in \mathcal{C}} \cup \dset{(t,v)}{v \in \cal{K}}$$ 
    
    and edges 
    \begin{eqnarray*}
    E(G_L(t',t))&:=&\dset{(t',C)(t,v)}{v \in \mathcal{K}, C\in \mathcal{C}
    \wedge \alpha(v)=C}\\ 
    &\cup& \dset{(t,v)(t,C)}{v \in \mathcal{K}, C\in \mathcal{C}
    \wedge \alpha(v)=C}\\ 
    &\cup&\dset{(t, u)(t,v)}{uv \in E(G(Q,\cal{K}))}.
    \end{eqnarray*} 
    We define a tile $T_L(t', t)=(G_L(t',t), L, R)$, 
    where $L:=\dset{(t',C}{C \in \mathcal{C}}$  and 
    $R:=\dset{(t, C)}{C \in \mathcal{C}}$, 
    with ordering induced by the ordering of 
    $\mathcal{C}$, respectively. The \emph{total learning tile} 
    $T_L(t_0,\ldots,t_{m})$ of the instance 
    $(S, \mathcal{C}, T, \sigma)$ and 
    learning space $(Q, \mathcal{K})$ is obtained as 
    the join of the compatible sequence of tiles 
    $\otimes(T_L(t_0,t_1), T_L(t_1, t_2), \ldots , T_L(t_{m-1},t_{m}))$. 
\end{definition}

As such, the total learning tile is very rich. 
However, in the beginning of the maturing process, 
the subjects tend to reside in lower categories, and towards
the end of the maturing process, the subjects tend
to reside in higher categories. In order to simplify the 
graph, it is reasonable to reduce it to the data actually 
observed by the tests. There are two models for this.
The first simplified model is the \DEF{possibilistic model}, 
in which the wall vertices of the tiles
reflect the actual subjects in those categories, 
and which reduces the learning space graph in each tile 
to the subgraph that is
spanned by the (union of) 
shortest paths between the entry and exit
vertices assigned to subjects by the tests defining a tile. 
The following formal definition is illustrated in Figure \ref{fig:Possibilistic}.
\begin{figure}[H]
    \centering
    \usetikzlibrary{matrix}
\usetikzlibrary{shadows}
\usetikzlibrary{shapes.geometric}
\usetikzlibrary{calc}
\begin{tikzpicture}
\begin{scope}[scale=0.28]

\begin{scope}[rotate=14, shift={(1.4,3.5)}]
\coordinate(a0) at (0,-4) {}{};
\coordinate (a1) at (0,0) {} {};
\coordinate (a2) at (4,0) {}{};
\coordinate (a3) at (4,4) {}{};
\coordinate (a4) at (0,4) {}{};
\coordinate (a5) at ($({sqrt(2)}, {sqrt(2)})$) {}{};
\coordinate (a6) at ($({4+sqrt(2)}, {sqrt(2)})$) {}{};
\coordinate (a7) at ($({sqrt(2)+4}, {sqrt(2)+4})$) {}{};
\coordinate (a8) at ($({sqrt(2)}, {sqrt(2)+4})$) {}{};
\coordinate (a9) at (8,0) {} {};
\coordinate (a10) at (12,0) {}{};
\coordinate (a11) at (8,4) {}{};
\coordinate (a12) at (12,4) {}{};
\coordinate (a13) at ($({sqrt(2)+8}, {sqrt(2)})$) {}{};
\coordinate (a14) at ($({12+sqrt(2)}, {sqrt(2)})$) {}{};
\coordinate (a15) at ($({sqrt(2)+8}, {sqrt(2)+4})$) {}{};
\coordinate (a16) at ($({sqrt(2)+12}, {sqrt(2)+4})$) {}{};
\draw[thick] (a0)-- (a1);
\draw [thick](a1) -- (a2);
\draw [thick] (a2)--(a3);
\draw [thick] (a1)--(a4);
\draw[thick](a3) -- (a4);
\draw[thick] (a5) -- (a6);
\draw [thick](a5)--(a8);
\draw [thick] (a6)--(a7);
\draw [thick] (a7) -- (a8);
\draw [thick](a9) -- (a10);
\draw[thick](a12) -- (a11);
\draw[thick](a10)--(a12);
\draw[thick](a9)--(a11);
\draw [thick](a13) -- (a14);
\draw[thick] (a16) -- (a15);
\draw[thick](a14)--(a16);
\draw [thick] (a15)--(a13);

\draw [thick](a1) -- (a5);
\draw [thick](a2) -- (a6);
\draw [thick](a3) -- (a7);
\draw[thick] (a4) -- (a8);

\draw [thick](a3)--(a11);
\draw [thick](a2)--(a9);
\draw [thick](a6)--(a13);
\draw [thick](a13)--(a9);
\draw [thick](a14)--(a10);
\draw [thick](a7)--(a15);
\draw [thick](a15)--(a11);
\draw [thick](a16)--(a12);

\coordinate (b1) at ($({sqrt(2)+12}, {sqrt(2)+8})$) {}{};
\coordinate (b2) at ($({sqrt(2)+12}, {sqrt(2)+4})$) {}{};
\coordinate (b3) at ($({sqrt(2)+16}, {sqrt(2)+8})$) {}{};
\coordinate (b4) at ($({sqrt(2)+16}, {sqrt(2)+4})$) {}{};
\coordinate (b5) at ($({sqrt(2)+20}, {sqrt(2)+8})$) {}{};
\coordinate (b6) at ($({sqrt(2)+20}, {sqrt(2)+4})$) {}{};
\coordinate (b7) at ($({sqrt(2)+24}, {sqrt(2)+8})$) {}{};
\coordinate (b8) at ($({sqrt(2)+24}, {sqrt(2)+4})$) {}{};


\draw [thick](b1)--(b2);
\draw [thick](b4)--(b3);
\draw [thick](b2)--(b4);
\draw [thick](b1)--(b3);
\draw [thick](b5)--(b6)--(b8)--(b7)--cycle;
\draw [thick](b3)--(b5);
\draw [thick](b4)--(b6);
\node [rectangle, fill=black!40!green,  inner sep=2pt] at (a0) {};
\node [rectangle, fill=black!40!green, inner sep=2pt] at (a1) {};
\node [circle,  draw, very thick, color=blue, inner sep=1.5pt] at (a2) {};
\node [circle,  draw, very thick, color=blue, inner sep=1.5pt] at (a3) {};
\node [circle,  draw, very thick, color=blue, inner sep=1.5pt] at (a4) {};
\node [circle,  draw, very thick, color=blue, inner sep=1.5pt] at (a5) {};
\node [circle,  draw, very thick, color=blue, inner sep=1.5pt] at (a6) {};
\node [circle,  draw, very thick, color=blue, inner sep=1.5pt] at (a7) {};
\node [circle,  draw, very thick, color=blue, inner sep=1.5pt] at (a8) {};
\node [circle,  draw, very thick, color=blue, inner sep=1.5pt] at (a9) {};
\node [diamond, fill=black!20!yellow, inner sep=2pt] at (a10) {};
\node [circle,  draw, very thick, color=blue, inner sep=1.5pt] at (a11) {};
\node [diamond, fill=black!20!yellow, inner sep=2pt] at (a12) {};
\node [circle,  draw, very thick, color=blue, inner sep=1.5pt] at (a13) {};
\node [diamond, fill=black!20!yellow, inner sep=2pt] at (a14) {};
\node [circle,  draw, very thick, color=blue, inner sep=1.5pt] at (a15) {};
\node [diamond, fill=black!20!yellow, inner sep=2pt] at (a16) {};
\node [circle, fill=black!20!red, inner sep=1.5pt] at (b1) {};
\node [circle, fill=black!20!red, inner sep=1.5pt] at (b3) {};
\node [circle, fill=black!20!red, inner sep=1.5pt] at (b4) {};
\node [circle, fill=black!20!red, inner sep=1.5pt] at (b5) {};
\node [circle, fill=black!20!red, inner sep=1.5pt] at (b6) {};
\node [circle, fill=black!20!red, inner sep=1.5pt] at (b7) {};
\node [circle, fill=black!20!red, inner sep=1.5pt] at (b8) {};
\end{scope}
\draw[thick]  (25,20) rectangle (26,-4);
\coordinate (A) at (25, 2);
\coordinate (B) at (26,2);
\draw[thick] (A)--(B);
\coordinate (A1) at (25, 8);
\coordinate (B1) at (26,8);
\draw[thick] (A1)--(B1);
\coordinate (A2) at (25, 14);
\coordinate (B2) at (26,14);
\draw[thick] (A2)--(B2);
\draw[thick]  (-5,20) rectangle (-6,-4);
\coordinate (C) at (-5, 2);
\coordinate (D) at (-6, 2);
\draw[thick] (C)--(D);
\coordinate (C1) at (-5, 8);
\coordinate (D1) at (-6, 8);
\draw[thick] (C1)--(D1);
\coordinate (C2) at (-5, 14);
\coordinate (D2) at (-6, 14);
\draw[thick] (C2)--(D2);

\node(1) [circle, fill=black, inner sep=1pt] at (-5.5,-1.5) {};
\node(11) [circle, fill=black, inner sep=1pt] at (25.5,-1.5) {};
\draw[thick, color=black!40!green](1)--(a0)--(11);
\draw[thick, color=black!40!green](1)--(a1)--(11);

\node (2)[circle, fill=black, inner sep=1pt] at (-5.5,4.5) {};
\node (22)[circle, fill=black, inner sep=1pt] at (25.5,4.5) {};
\draw[thick, color=blue](2)--(a3)--(22);
\draw[thick, color=blue](2)--(a2)--(22);
\draw[thick, color=blue](2)--(a4)--(22);
\draw[thick, color=blue](2)--(a5)--(22);
\draw[thick, color=blue](2)--(a6)--(22);
\draw[thick, color=blue](2)--(a7)--(22);
\draw[thick, color=blue](2)--(a8)--(22);
\draw[thick, color=blue](2)--(a9)--(22);
\draw[thick, color=blue](2)--(a11)--(22);
\draw[thick, color=blue](2)--(a13)--(22);
\draw[thick, color=blue](2)--(a15)--(22);

\node(3) [circle, fill=black, inner sep=1pt] at (-5.5,10.5) {};
\node(33) [circle, fill=black, inner sep=1pt] at (25.5,10.5) {};
\draw[thick, color=black!20!yellow](3)--(a10)--(33);
\draw[thick, color=black!20!yellow](3)--(a12)--(33);
\draw[thick, color=black!20!yellow](3)--(a14)--(33);
\draw[thick, color=black!20!yellow](3)--(a16)--(33);

\node (4)[circle, fill=black, inner sep=1pt] at (-5.5,16.5) {};
\node (44)[circle, fill=black, inner sep=1pt] at (25.5,16.5) {};
\draw[thick, color=black!20!red](4)--(b1)--(44);
\draw[thick, color=black!20!red](4)--(b3)--(44);
\draw[thick, color=black!20!red](4)--(b4)--(44);
\draw[thick, color=black!20!red](4)--(b5)--(44);
\draw[thick, color=black!20!red](4)--(b6)--(44);
\draw[thick, color=black!20!red](4)--(b7)--(44);
\draw[thick, color=black!20!red](4)--(b8)--(44);
\end{scope}
\end{tikzpicture}
    \caption{Total learning tile}
    \label{fig:pctB}
\end{figure}
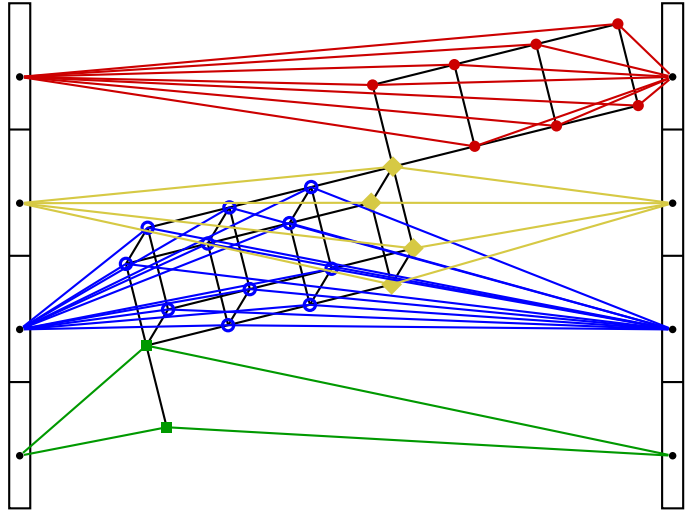

\begin{definition}
    Let $(S, \mathcal{C}, T, \sigma)$ be an OPD instance and 
    $(Q, \cal{K})$ a given learning space with a 
    graph $G(Q, \cal{K})$. For each $s \in S$ and each pair of
    consecutive tests $t',t \in T$, let $\mathcal{P}_{s, t}$ 
    denote 
    the set of all shortest paths of a subject $s$ between 
    knowledge states $t'(s)$ and $t(s)$ in $G(Q, \cal{K})$. 
    Let $$V_{t}:=\bigcup_{s\in S}\bigcup_{P \in 
    \mathcal{P}_{s, t}}V(P)$$ and let  
    $$E_{t}:=\bigcup_{s\in S}\bigcup_{P \in 
    \mathcal{P}_{s, t}}E(P).$$
    For each pair of consecutive tests $t', t \in T$, we define a graph 
    $G_P(t', t)$ with 
    vertices $$V(G_P(t',t))=\dset{(t',s)}{s \in S} \cup 
    \dset{(t,s)}{s \in S} \cup \dset{(t,v)}{v \in V_{t}}$$ 
    and edges 
    \begin{eqnarray*}
        E(G_P(t',t))&=&\dset{(t',s)(t,v)}{s \in S \wedge v
        \hbox{ is the first vertex of }P, P \in \mathcal{P}_{s, t}}\\ 
        &\cup&\dset{(t,v)(t,s)}{s \in S \wedge v
        \hbox{ is the last vertex of }P, P \in \mathcal{P}_{s, t}}\\
        &\cup&\dset{(t,u)(t,v)}{uv \in E_{t}}.
    \end{eqnarray*} 
    Let $\pi_0,\pi_1,\dots,\pi_m\in \Pi(S)$ be a 
    combinatorial layout of the OPD instance 
    $(S,\mathcal{C},T,\sigma)$. We define a tile 
    $T_P(\pi_{i-1},\pi_{i})=(G_P(t_{i-1},t_{i}),L,R)$, 
    where $L:=\dset{(t_{i-1},s)}{s\in S}$ is ordered by the permutation $\pi_{i-1}$ 
    and $R:=\dset{(t_{i},s)}{s\in S}$ is ordered by the permutation 
    $\pi_{i}$ for $i=0,1,\ldots,m$. 
    
    The \emph{possibilistic learning tile} 
    $T_P(\pi_0,\pi_1,\ldots,\pi_m)$ of the OPD instance
    $(S,\mathcal{C},T,\sigma)$ with the combinatorial layout
    $\pi_0,\pi_1,\dots,\pi_m$ is obtained by joining the 
    compatible sequence of tiles 
    $T_P(\pi_0,\pi_1),T_P(\pi_1,\pi_2),\ldots,T_P(\pi_{m-1},\pi_m)$ as 
    \[\otimes \left( T_P(\pi_0,\pi_1),\ldots, T_P(\pi_{m-1},\pi_m) \right).\] 
\end{definition} 
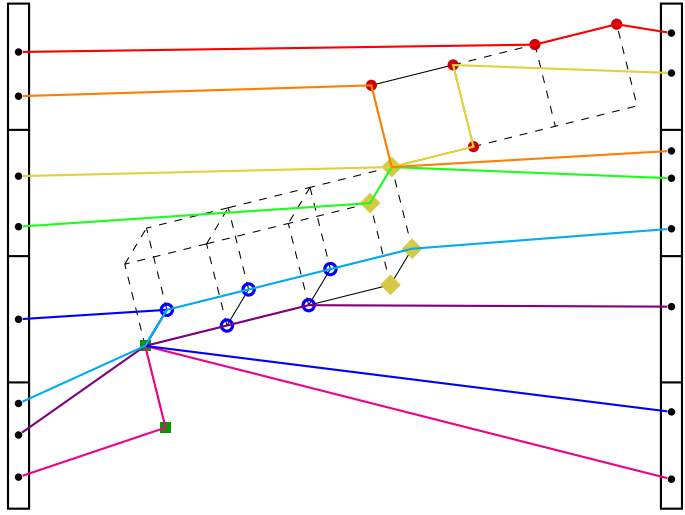
\begin{figure}[H]
    \centering
    \usetikzlibrary{matrix}
\usetikzlibrary{shadows}
\usetikzlibrary{calc}
\usetikzlibrary{shapes.geometric}
\begin{tikzpicture}
\begin{scope}[scale=0.28]
\begin{scope}[rotate=14, shift={(1.4,3.5)}]
\coordinate(a0) at (0,-4) {}{};
\coordinate (a1) at (0,0) {} {};
\coordinate (a2) at (4,0) {}{};
\coordinate (a3) at (4,4) {}{};
\coordinate (a4) at (0,4) {}{};
\coordinate (a5) at ($({sqrt(2)}, {sqrt(2)})$) {}{};
\coordinate (a6) at ($({4+sqrt(2)}, {sqrt(2)})$) {}{};
\coordinate (a7) at ($({sqrt(2)+4}, {sqrt(2)+4})$) {}{};
\coordinate (a8) at ($({sqrt(2)}, {sqrt(2)+4})$) {}{};
\coordinate (a9) at (8,0) {} {};
\coordinate (a10) at (12,0) {}{};
\coordinate (a11) at (8,4) {}{};
\coordinate (a12) at (12,4) {}{};
\coordinate (a13) at ($({sqrt(2)+8}, {sqrt(2)})$) {}{};
\coordinate (a14) at ($({12+sqrt(2)}, {sqrt(2)})$) {}{};
\coordinate (a15) at ($({sqrt(2)+8}, {sqrt(2)+4})$) {}{};
\coordinate (a16) at ($({sqrt(2)+12}, {sqrt(2)+4})$) {}{};
\draw (a0)-- (a1);
\draw (a1) -- (a2);
\draw [dashed] (a2)--(a3);
\draw [dashed] (a1)--(a4);
\draw[dashed](a3) -- (a4);
\draw[dashed] (a5) -- (a6);
\draw [dashed](a5)--(a8);
\draw [dashed] (a6)--(a7);
\draw [dashed] (a7) -- (a8);
\draw (a9) -- (a10);
\draw[dashed](a12) -- (a11);
\draw[dashed](a10)--(a12);
\draw[dashed](a9)--(a11);
\draw (a13) -- (a14);
\draw[dashed] (a16) -- (a15);
\draw[dashed](a14)--(a16);
\draw [dashed] (a15)--(a13);

\draw [dashed](a1) -- (a5);
\draw (a2) -- (a6);
\draw [dashed](a3) -- (a7);
\draw[dashed] (a4) -- (a8);

\draw [dashed](a3)--(a11);
\draw [dashed](a2)--(a9);
\draw [dashed](a6)--(a13);
\draw (a13)--(a9);
\draw (a14)--(a10);
\draw [dashed](a7)--(a15);
\draw [dashed](a15)--(a11);

\coordinate (b1) at ($({sqrt(2)+12}, {sqrt(2)+8})$) {}{};
\coordinate (b2) at ($({sqrt(2)+12}, {sqrt(2)+4})$) {}{};
\coordinate (b3) at ($({sqrt(2)+16}, {sqrt(2)+8})$) {}{};
\coordinate (b4) at ($({sqrt(2)+16}, {sqrt(2)+4})$) {}{};
\coordinate (b5) at ($({sqrt(2)+20}, {sqrt(2)+8})$) {}{};
\coordinate (b6) at ($({sqrt(2)+20}, {sqrt(2)+4})$) {}{};
\coordinate (b7) at ($({sqrt(2)+24}, {sqrt(2)+8})$) {}{};
\coordinate (b8) at ($({sqrt(2)+24}, {sqrt(2)+4})$) {}{};

\draw [dashed](b1)--(b2);
\draw (b4)--(b3);
\draw (b2)--(b4);
\draw (b1)--(b3);
\draw [dashed](b5)--(b6)--(b8)--(b7)--cycle;
\draw [dashed](b3)--(b5);
\draw [dashed](b4)--(b6);
\node [rectangle, fill=black!40!green,  inner sep=2pt] at (a0) {};
\node [rectangle, fill=black!40!green, inner sep=2pt] at (a1) {};
\node [circle,  draw, very thick, color=blue, inner sep=1.5pt] at (a2) {};
\node [circle,  draw, very thick, color=blue, inner sep=1.5pt] at (a5) {};
\node [circle,  draw, very thick, color=blue, inner sep=1.5pt] at (a6) {};
\node [circle,  draw, very thick, color=blue, inner sep=1.5pt] at (a9) {};
\node [diamond, fill=black!20!yellow, inner sep=2pt] at (a10) {};
\node [diamond, fill=black!20!yellow, inner sep=2pt] at (a12) {};
\node [circle,  draw, very thick, color=blue, inner sep=1.5pt] at (a13) {};
\node [diamond, fill=black!20!yellow, inner sep=2pt] at (a14) {};
\node [diamond, fill=black!20!yellow, inner sep=2pt] at (a16) {};
\node [circle, fill=black!20!red, inner sep=1.5pt] at (b1) {};
\node [circle, fill=black!20!red, inner sep=1.5pt] at (b3) {};
\node [circle, fill=black!20!red, inner sep=1.5pt] at (b4) {};
\node [circle, fill=black!20!red, inner sep=1.5pt] at (b5) {};
\node [circle, fill=black!20!red, inner sep=1.5pt] at (b7) {};
\end{scope}
\draw[thick]  (25,20) rectangle (26,-4);
\coordinate (A) at (25, 2);
\coordinate (B) at (26,2);
\draw[thick] (A)--(B);
\coordinate (A1) at (25, 8);
\coordinate (B1) at (26,8);
\draw[thick] (A1)--(B1);
\coordinate (A2) at (25, 14);
\coordinate (B2) at (26,14);
\draw[thick] (A2)--(B2);
\draw[thick]  (-5,20) rectangle (-6,-4);
\coordinate (C) at (-5, 2);
\coordinate (D) at (-6, 2);
\draw[thick] (C)--(D);
\coordinate (C1) at (-5, 8);
\coordinate (D1) at (-6, 8);
\draw[thick] (C1)--(D1);
\coordinate (C2) at (-5, 14);
\coordinate (D2) at (-6, 14);
\draw[thick] (C2)--(D2);

\node [circle, fill=black, inner sep=1pt](s1) at (-5.5,-2.5) {};
\node [circle, fill=black, inner sep=1pt](s'1) at (25.5,-2.6) {};
\draw[thick, magenta] (s1)--(a0)--(a1)--(s'1);

\node [circle, fill=black, inner sep=1pt](s2) at (-5.5,-0.5) {};
\node [circle, fill=black, inner sep=1pt] (s'2) at (25.5,5.6) {};
\draw[thick, violet] (s2)--(a1)--(a2)--(a9)--(s'2);

\node [circle, fill=black, inner sep=1pt] (s3) at (-5.5,5) {};
\node [circle, fill=black, inner sep=1pt] (s'3) at (25.5,0.6) {};
\draw[thick,blue](s3)--(a5)--(a1)--(s'3);

\node [circle, fill=black, inner sep=1pt] (s4) at (-5.5,1) {};
\node [circle, fill=black, inner sep=1pt](s'4) at (25.5,9.3) {};
\draw[thick, cyan](s4)--(a1)--(a5)--(a6)--(a13)--(a14)--(s'4);

\node [circle, fill=black, inner sep=1pt](s5) at (-5.5,9.4) {};
\node [circle, fill=black, inner sep=1pt](s'5) at (25.5,11.7) {};
\draw[thick, white!10!green](s5)--(a12)--(b2)--(s'5);

\node [circle, fill=black, inner sep=1pt](s6) at (-5.5,11.8) {};
\node [circle, fill=black, inner sep=1pt] (s'6) at (25.5,16.7) {};
\draw[thick, black!15!yellow](s6)--(b2)--(b4)--(b3)--(s'6);

\node [circle, fill=black, inner sep=1pt](s7) at (-5.5,15.6) {};
\node [circle, fill=black, inner sep=1pt] (s'7) at (25.5,13) {};
\draw[thick,orange](s7)--(b1)--(b2)--(s'7);

\node [circle, fill=black, inner sep=1pt](s8) at (-5.5,17.7) {};
\node [circle, fill=black, inner sep=1pt](s'8) at (25.5,18.6) {};
\draw[thick, red](s8)--(b5)--(b7)--(s'8);
\end{scope}
\end{tikzpicture}
    \caption{Possibilistic learning tile.}
    \label{fig:Possibilistic}
\end{figure}

In the next definition of a tile, 
the exact model reduces the learning space graph even further to
the paths actually traversed by the subjects. 
The graph of the tile used there is therefore 
a subgraph of the graph of the possibilistic tile 
(cf.\ Figure \ref{fig:Exact}).

\begin{definition}
    Let $(S, \mathcal{C}, T, \sigma)$ be an OPD instance, 
    $(Q, \cal K)$ a given learning space and let $G(Q, \cal K)$ 
    be a graph of $(Q, \cal K)$. For each $s \in S$ and 
    each pair of consecutive tests $t',t \in T$, let 
    $P_{s,t}$ denote the actually traversed path of subject $s$ 
    between $t'$ and $t$ in $G(Q, \cal K)$. Let 
    $$V'_t:=\bigcup_{s\in S}V(P_{s,t})$$ and 
    $$E'_t:=\bigcup_{s\in S}E(P_{s,t}).$$ We define a graph 
    $G_E(t', t)$ with vertices 
    $$V(G_E(t',t)):=\dset{(t',s)}{s \in S} \cup 
    \dset{(t, s)}{s \in S} \cup \dset{(t,v)}{v \in V'_t}$$ 
    and edges 
    \begin{eqnarray*}    
    E(G_E(t',t))&:=&\dset{(t',s)(t,v)}{s \in S \wedge v 
    \hbox{ is the first vertex of }P_{s,t}}\\
    &\cup&\dset{(t,v)(t,s)}{s \in S \wedge v
    \hbox{ is the last vertex of }P_{s,t}}\\
    &\cup&\dset{(t,u)(t,v)}{uv \in E'_t}. 
    \end{eqnarray*}
    
    Let $\pi_0,\pi_1,\dots,\pi_m\in \Pi(S)$ be a 
    combinatorial layout of an OPD instance $(S,\mathcal{C},T,\sigma)$.
    We define a tile 
    $T_E(\pi_{i-1},\pi_i)=(G_E(t_{i-1},t_i),L,R)$, 
    where $L:=\dset{(t_{i-1},s)}{s\in S}$ is ordered by the
    permutation $\pi_{i-1}$ and $R:=\dset{(t_i,s)}{s\in S}$ is 
    ordered by the permutation $\pi_i$ for $i=1,\ldots,m$. 
    The \emph{exact learning tile} $T_E(\pi_0,\pi_1,\ldots,\pi_m)$ 
    of the OPD instance $(S,\mathcal{C},T,\sigma)$ with the 
    combinatorial layout $\pi_0,\pi_1,\dots,\pi_m$ is obtained as 
    join of the compatible sequence of tiles \newline    $\otimes\left(T_E(\pi_0,\pi_1),T_E(\pi_1,\pi_2),\ldots,T_E(\pi_{m-1},\pi_m)\right)$.    
\end{definition}

\begin{figure}[H]
    \centering
    \usetikzlibrary{matrix}
\usetikzlibrary{shadows}
\usetikzlibrary{shapes.geometric}
\usetikzlibrary{calc}
\begin{tikzpicture}
\begin{scope}[scale=0.28]

\begin{scope}[rotate=14, shift={(1.4,3.5)}]
\coordinate(a0) at (0,-4) {}{};
\coordinate (a1) at (0,0) {} {};
\coordinate (a2) at (4,0) {}{};
\coordinate (a3) at (4,4) {}{};
\coordinate (a4) at (0,4) {}{};
\coordinate (a5) at ($({sqrt(2)}, {sqrt(2)})$) {}{};
\coordinate (a6) at ($({4+sqrt(2)}, {sqrt(2)})$) {}{};
\coordinate (a7) at ($({sqrt(2)+4}, {sqrt(2)+4})$) {}{};
\coordinate (a8) at ($({sqrt(2)}, {sqrt(2)+4})$) {}{};
\coordinate (a9) at (8,0) {} {};
\coordinate (a10) at (12,0) {}{};
\coordinate (a11) at (8,4) {}{};
\coordinate (a12) at (12,4) {}{};
\coordinate (a13) at ($({sqrt(2)+8}, {sqrt(2)})$) {}{};
\coordinate (a14) at ($({12+sqrt(2)}, {sqrt(2)})$) {}{};
\coordinate (a15) at ($({sqrt(2)+8}, {sqrt(2)+4})$) {}{};
\coordinate (a16) at ($({sqrt(2)+12}, {sqrt(2)+4})$) {}{};
\draw(a0)-- (a1);
\draw (a1) -- (a2);
\draw [dashed] (a2)--(a3);
\draw [dashed] (a1)--(a4);
\draw[dashed](a3) -- (a4);
\draw[dashed] (a5) -- (a6);
\draw [dashed](a5)--(a8);
\draw [dashed] (a6)--(a7);
\draw [dashed] (a7) -- (a8);
\draw [dashed](a9) -- (a10);
\draw[dashed](a12) -- (a11);
\draw[dashed](a10)--(a12);
\draw[dashed](a9)--(a11);
\draw [dashed](a13) -- (a14);
\draw[dashed] (a16) -- (a15);
\draw[dashed](a14)--(a16);
\draw [dashed] (a15)--(a13);

\draw [dashed](a1) -- (a5);
\draw [dashed](a2) -- (a6);
\draw [dashed](a3) -- (a7);
\draw[dashed] (a4) -- (a8);

\draw [dashed](a3)--(a11);
\draw [dashed](a2)--(a9);
\draw [dashed](a6)--(a13);
\draw [dashed](a13)--(a9);
\draw [dashed](a14)--(a10);
\draw [dashed](a7)--(a15);
\draw [dashed](a15)--(a11);

\coordinate (b1) at ($({sqrt(2)+12}, {sqrt(2)+8})$) {}{};
\coordinate (b2) at ($({sqrt(2)+12}, {sqrt(2)+4})$) {}{};
\coordinate (b3) at ($({sqrt(2)+16}, {sqrt(2)+8})$) {}{};
\coordinate (b4) at ($({sqrt(2)+16}, {sqrt(2)+4})$) {}{};
\coordinate (b5) at ($({sqrt(2)+20}, {sqrt(2)+8})$) {}{};
\coordinate (b6) at ($({sqrt(2)+20}, {sqrt(2)+4})$) {}{};
\coordinate (b7) at ($({sqrt(2)+24}, {sqrt(2)+8})$) {}{};
\coordinate (b8) at ($({sqrt(2)+24}, {sqrt(2)+4})$) {}{};

\draw [dashed](b1)--(b2);
\draw (b4)--(b3);
\draw [dashed](b2)--(b4);
\draw [dashed](b1)--(b3);
\draw [dashed](b5)--(b6)--(b8)--(b7)--cycle;
\draw [dashed](b3)--(b5);
\draw [dashed](b4)--(b6);
\node [rectangle, fill=black!40!green,  inner sep=2pt] at (a0) {};
\node [rectangle, fill=black!40!green, inner sep=2pt] at (a1) {};
\node [circle,  draw, very thick, color=blue, inner sep=1.5pt] at (a2) {};
\node [circle,  draw, very thick, color=blue, inner sep=1.5pt] at (a5) {};
\node [circle,  draw, very thick, color=blue, inner sep=1.5pt] at (a6) {};
\node [circle,  draw, very thick, color=blue, inner sep=1.5pt] at (a9) {};
\node [diamond, fill=black!20!yellow, inner sep=2pt] at (a12) {};
\node [circle,  draw, very thick, color=blue, inner sep=1.5pt] at (a13) {};
\node [diamond, fill=black!20!yellow, inner sep=2pt] at (a14) {};
\node [diamond, fill=black!20!yellow, inner sep=2pt] at (a16) {};
\node [circle, fill=black!20!red, inner sep=1.5pt] at (b1) {};

\node [circle, fill=black!20!red, inner sep=1.5pt] at (b3) {};
\node [circle, fill=black!20!red, inner sep=1.5pt] at (b4) {};
\node [circle, fill=black!20!red, inner sep=1.5pt] at (b5) {};
\node [circle, fill=black!20!red, inner sep=1.5pt] at (b7) {};
\end{scope}
\draw[thick]  (25,20) rectangle (26,-4);
\coordinate (A) at (25, 2);
\coordinate (B) at (26,2);
\draw[thick] (A)--(B);
\coordinate (A1) at (25, 8);
\coordinate (B1) at (26,8);
\draw[thick] (A1)--(B1);
\coordinate (A2) at (25, 14);
\coordinate (B2) at (26,14);
\draw[thick] (A2)--(B2);
\draw[thick]  (-5,20) rectangle (-6,-4);
\coordinate (C) at (-5, 2);
\coordinate (D) at (-6, 2);
\draw[thick] (C)--(D);
\coordinate (C1) at (-5, 8);
\coordinate (D1) at (-6, 8);
\draw[thick] (C1)--(D1);
\coordinate (C2) at (-5, 14);
\coordinate (D2) at (-6, 14);
\draw[thick] (C2)--(D2);

\node [circle, fill=black, inner sep=1pt](s1) at (-5.5,-2.5) {};
\node [circle, fill=black, inner sep=1pt](s'1) at (25.5,-2.6) {};
\draw[thick, magenta] (s1)--(a0)--(a1)--(s'1);

\node [circle, fill=black, inner sep=1pt](s2) at (-5.5,-0.5) {};
\node [circle, fill=black, inner sep=1pt] (s'2) at (25.5,5.6) {};
\draw[thick, violet] (s2)--(a1)--(a2)--(a9)--(s'2);

\node [circle, fill=black, inner sep=1pt] (s3) at (-5.5,5) {};
\node [circle, fill=black, inner sep=1pt] (s'3) at (25.5,0.6) {};
\draw[thick,blue](s3)--(a5)--(a1)--(s'3);

\node [circle, fill=black, inner sep=1pt] (s4) at (-5.5,1) {};
\node [circle, fill=black, inner sep=1pt](s'4) at (25.5,9.3) {};
\draw[thick, cyan](s4)--(a1)--(a5)--(a6)--(a13)--(a14)--(s'4);

\node [circle, fill=black, inner sep=1pt](s5) at (-5.5,9.4) {};
\node [circle, fill=black, inner sep=1pt](s'5) at (25.5,11.7) {};
\draw[thick, white!10!green](s5)--(a12)--(b2)--(s'5);

\node [circle, fill=black, inner sep=1pt](s6) at (-5.5,11.8) {};
\node [circle, fill=black, inner sep=1pt] (s'6) at (25.5,16.7) {};
\draw[thick, black!15!yellow](s6)--(b2)--(b4)--(b3)--(s'6);

\node [circle, fill=black, inner sep=1pt](s7) at (-5.5,15.6) {};
\node [circle, fill=black, inner sep=1pt] (s'7) at (25.5,13) {};
\draw[thick,orange](s7)--(b1)--(b2)--(s'7);

\node [circle, fill=black, inner sep=1pt](s8) at (-5.5,17.7) {};
\node [circle, fill=black, inner sep=1pt](s'8) at (25.5,18.6) {};
\draw[thick, red](s8)--(b5)--(b7)--(s'8);
    
\end{scope}
\end{tikzpicture}
    \caption{Exact learning tile.}
    \label{fig:Exact}
\end{figure}

Finally, the minimal model ignores the detailed stages 
of the subjects between any two tests, and
reduces the detailed data to the observed ordinal panel data,
implying that each tile is a matching whose edges
represent subjects and connect the vertices representing
the corresponding subject in two consecutive tests. For this model,
it suffices that the tests map subjects into categories; the 
details of subject's knowledge states can be ignored.


\begin{definition}
Let $(S, \mathcal{C}, T, \sigma)$ be an OPD instance. For each pair of 
consecutive tests $t', t \in T$, we define a graph $G_O(t',t)$ with 
vertices $$V(G_O(t',t)):=\dset{(t',s)}{s\in S}\cup 
\dset{(t,s)}{s\in S}$$ 
and edges 
$$E(G_O(t',t)):=\dset{(t',s)(t,s)}{s\in S}.$$ 
Let $\pi_0,\pi_1,\dots,\pi_m\in \Pi(S)$ be a 
combinatorial layout of the OPD 
instance $(S,\mathcal{C},T,\sigma)$.
Note that $\pi_i$ orders the subjects within each category of 
$\mathcal{C}$, whereas $t_i$ does not. We define a tile 
$T_O(\pi_{i-1},\pi_i)=(G_O(t_{i-1},t_i),L,R)$, where $L$ is the 
sequence of vertices $\dset{(t_{i-1},s)}{s\in S}$ ordered by the 
permutation $\pi_{i-1}$ and $R$ is the sequence of vertices 
$\dset{(t_i,s)}{s\in S}$ ordered by the permutation $\pi_i$. 
Note that each vertex of $T_O(\pi_{i-1},\pi_i)$ 
 is either a left or a 
right wall vertex and $|S|=|L|=|R|$. An \DEF{ordinal panel tile} 
$T_O(\pi_0,\pi_1,\ldots,\pi_m)$ of the OPD instance 
$(S,\mathcal{C},T,\sigma)$ with the combinatorial layout 
$\pi_0,\pi_1,\dots,\pi_m$ is obtained by joining the compatible 
sequence of tiles
$\otimes\left(T_O(\pi_0,\pi_1),T_O(\pi_1,\pi_2),\ldots,T_O(\pi_{m-1},\pi_m)\right)$.
\end{definition}

\begin{definition}
An \DEF{ordinal panel drawing} $D(S,\mathcal{C},T,\sigma)$ of an OPD
instance $(S,\mathcal{C},T,\sigma)$ with the combinatorial layout 
$\pi_0,\pi_1,\dots,\pi_m$ is a tile drawing of the ordinal panel tile 
$T_O(\pi_0,\pi_1,\ldots,\pi_m)$, such that for each $i=1,\ldots,m$, 
$D(S,\mathcal{C},T,\sigma)$ restricted to $T_O(\pi_{i-1},\pi_{i})$ is 
a tile drawing of $T_O(\pi_{i-1},\pi_{i})$. A \DEF{panel ranking tile} 
$T_R(t_0,t_1,\ldots,t_m)$ of an OPD instance $(S,\mathcal{C},T,\sigma)$ 
is an ordinal panel tile of an OPD instance $(S,\mathcal{C},T,\sigma)$
that has the smallest crossing number of an ordinal panel drawing
$D(S,\mathcal{C},T,\sigma)$ over all combinatorial layouts of an 
OPD instance $(S,\mathcal{C},T,\sigma)$. The \DEF{panel crossing number} of
$(S,\mathcal{C},T,\sigma)$, $pcr(S,\mathcal{C},T,\sigma)$, is then defined as the 
minimum number of crossings in any ordinal panel drawing of a
panel ranking tile of an ordinal panel instance.  
\end{definition}

\begin{figure}[H]
    \centering
   \usetikzlibrary{matrix}
\usetikzlibrary{shadows}
\usetikzlibrary{calc}
\begin{tikzpicture}
\begin{scope}[scale=0.28]

\draw[thick]  (25,20) rectangle (26,-4);
\coordinate (A) at (25, 2);
\coordinate (B) at (26,2);
\draw[thick] (A)--(B);
\coordinate (A1) at (25, 8);
\coordinate (B1) at (26,8);
\draw[thick] (A1)--(B1);
\coordinate (A2) at (25, 14);
\coordinate (B2) at (26,14);
\draw[thick] (A2)--(B2);
\draw[thick]  (-5,20) rectangle (-6,-4);
\coordinate (C) at (-5, 2);
\coordinate (D) at (-6, 2);
\draw[thick] (C)--(D);
\coordinate (C1) at (-5, 8);
\coordinate (D1) at (-6, 8);
\draw[thick] (C1)--(D1);
\coordinate (C2) at (-5, 14);
\coordinate (D2) at (-6, 14);
\draw[thick] (C2)--(D2);

\node [circle, fill=black, inner sep=1pt](s1) at (-5.5,-2.5) {};
\node [circle, fill=black, inner sep=1pt](s'1) at (25.5,-2.6) {};
\draw[ thick, magenta] (s1)--(s'1);

\node [circle, fill=black, inner sep=1pt](s2) at (-5.5,-0.5) {};
\node [circle, fill=black, inner sep=1pt] (s'2) at (25.5,5.6) {};
\draw[ thick, violet] (s2)--(s'2);

\node [circle, fill=black, inner sep=1pt] (s3) at (-5.5,5) {};
\node [circle, fill=black, inner sep=1pt] (s'3) at (25.5,0.6) {};
\draw[thick,blue](s3)--(s'3);

\node [circle, fill=black, inner sep=1pt] (s4) at (-5.5,1) {};
\node [circle, fill=black, inner sep=1pt](s'4) at (25.5,9.3) {};
\draw[thick, cyan](s4)--(s'4);

\node [circle, fill=black, inner sep=1pt](s5) at (-5.5,9.4) {};
\node [circle, fill=black, inner sep=1pt](s'5) at (25.5,11.7) {};
\draw[thick, white!10!green](s5)--(s'5);

\node [circle, fill=black, inner sep=1pt](s6) at (-5.5,11.8) {};
\node [circle, fill=black, inner sep=1pt] (s'6) at (25.5,16.7) {};
\draw[ thick, black!15!yellow](s6)--(s'6);

\node [circle, fill=black, inner sep=1pt](s7) at (-5.5,15.6) {};
\node [circle, fill=black, inner sep=1pt] (s'7) at (25.5,13) {};
\draw[ thick,orange](s7)--(s'7);

\node [circle, fill=black, inner sep=1pt](s8) at (-5.5,17.7) {};
\node [circle, fill=black, inner sep=1pt](s'8) at (25.5,18.6) {};
\draw[thick, red](s8)--(s'8);
\end{scope}
\end{tikzpicture}
    \caption{Ordinal panel tile}
    \label{fig:pctE}
\end{figure}
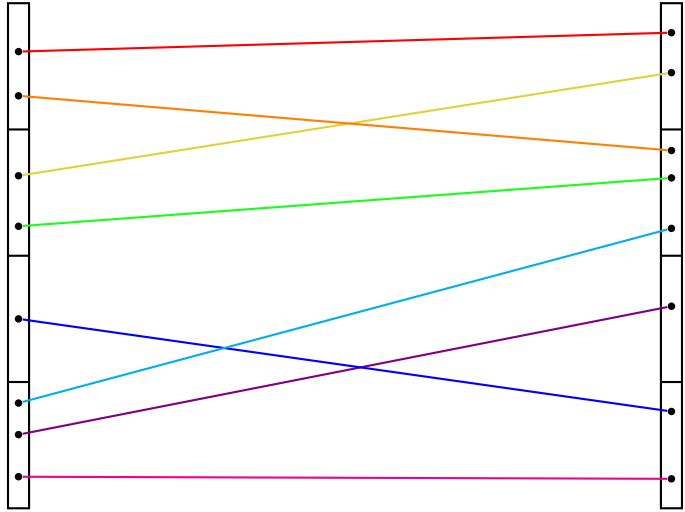

In the above, we have established a formal structure of 
maturity models that allows for various degrees of details.

In the following sections, we investigate their properties,
thus addressing the first challenge of 
improved understanding of maturity models.
Prior to that, we introduce an additional technical definition that will be needed in the following sections. 
\begin{definition}
Let $a,b \in \mathbb{Z}$ be integers and $a \leq b$. We denote $[a,b]=\{a, a+1, \ldots, b-1, b\}$. 

Similarly, let $X=\{x_1, x_2, \ldots, x_k\}$ be an ordered set. We denote
$[x_i,x_j] = \{x_i,x_{i+1},\dots,x_j\}$ for $1\leq i<j \leq k$.
\end{definition}

\section{\NP-completeness of tile crossing number}\label{sc:NPcompleteness}
As a first topic, we address the issue of 
visual representation of the introduced formal structures. 
First, we note that \NP-completeness of general tile crossing number follows from the \NP-completeness of regular crossing number.
This is proved by attaching two vertices of degree one to an arbitrary vertex of a graph G and declaring one to be the right and the other to be the left wall vertex. The tile crossing number of thus obtained tile $T$ is equal to the crossing number of graph $G$, thus if $tcr(T)$ could be obtained in polynomial time, so could $cr(G)$. As crossing number of a graph is \NP-complete~\cite{gareyCrossingNumberNPComplete1983}, even for cubic graphs~\cite{DBLP:journals/jct/Hlineny06a}, we conclude that tile crossing number of a general tile is \NP-complete. A more elaborate gadget subdividing an arbitrary edge of $G$ six times and introducing two new (wall) vertices of degree three would prove $tcr(T)$ is \NP-complete even for cubic tiles, i.e., tiles whose  vertices all have degree equal to three. 

Note that most \NP-completeness results on crossing numbers focus on simple problem adaptations.
The total learning tile, however, introduces two new edges per vertex in a manner that these vertices form (partial) apices over the original graph. A new technique of establishing \NP-completeness of crossing-minimization in the total learning tile is therefore needed. As a first open problem, one may consider asking about \NP-completeness of introducing an apex over the graph: 
\begin{open}
    Let $G'$ be a graph obtained from $G$ as a complete join of $G$ with a new vertex $v$. Is determining the crossing number of $G'$ \NP-complete? Is it still \NP-complete provided that $G$ is a partial hypercube? 
\end{open}

Next, we show that it is computationally infeasible to find visual representations with minimum number of  edge crossings for possibilistic tiles and exact tiles.
\begin{theorem}\label{thm:posstilehardness}
    Given $k\in\mathbb{N}$, it is \NP-complete to decide whether a possibilistic learning tile has an ordinal panel drawing with at most $k$ crossings, even for a single subject, a single category, and a single tile.
\end{theorem}
\begin{figure}
    \centering
    \includegraphics{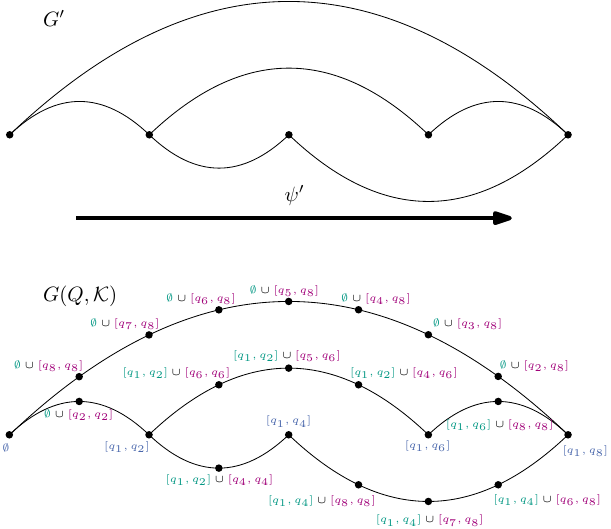}
    \caption{Illustration of the transformation of $G'$ into $G(Q,\mathcal{K})$ in \cref{thm:posstilehardness}. The knowledge states corresponding to the original graph vertices are colored blue. The knowledge states corresponding to the subdivision vertices are given as $A\cup B$ where $A$ is green and $B$ is purple.}
    \label{fig:tilehardness}
\end{figure}
\begin{proof}
    We start with \NP-membership. First, the number of crossings and how they appear along the edges of the drawing of each single tile can be guessed. Then it is easy to decide in polynomial time if the tile drawings can be joined together. Thus, \NP-membership follows.

    For \NP-hardness, we reduce from the classic crossing minimization problem, which is \NP-complete \cite{gareyCrossingNumberNPComplete1983}. In fact, the instances that Garey and Johnson constructed in their proof had specific properties that we will use for the reduction. An instance is given as $(G,k)$ with the \NP-complete question if there exists a drawing of $G$ with at most $k$ crossings. Here, $G$ is a multigraph with the following properties.
    \begin{enumerate}
        \item There exists an ordering $\psi$ of the vertex set $V(G)$ such that each vertex that is not the leftmost or rightmost vertex in the ordering has at least an edge going to the left in the ordering and an edge going to the right in the ordering. Furthermore, the leftmost vertex has an edge going to the right and the rightmost vertex has an edge going to the left. The above property is also known as $st$-numbering \cite{even1976computing}.
        \item If there exists a drawing of $G$ with at most $k$ crossings, then there exists a drawing of $G$ with $k$ crossings, where additionally $\psi(1)$ and $\psi(|V(G)|)$ lie on the outer face \cite[Normalizations 1 and 2]{gareyCrossingNumberNPComplete1983}.
    \end{enumerate}
    Thus, let $(G,k)$ be such an instance. First, we get rid of multiedges by subdividing all edges in $G$ once. Let us call this new graph $G'$. Obviously, $G'$ is still orderable with the new ordering $\psi'$ of $V(G')$, and $G'$ has a drawing with at most $k$ edge crossings if and only if $G$ has a drawing with at most $k$ edge crossings. Let $n=|V(G')|$. We choose $\psi'$ such that $\psi(1)=\psi'(1)$ and $\psi(|V(G)|)=\psi'(n)$. Thus $G$ has a drawing with at most $k$ crossings such that $\psi(1)$ and $\psi(|V(G)|)$ are on the outer face if and only if $G'$ has a drawing with at most crossings such that $\psi'(1)$ and $\psi'(n)$ are on the outer face.
    We will now transform $G'$ into a learning space graph $G(Q,\mathcal{K})$ by subdividing edges in $G'$ and providing knowledge states for the vertices (illustrated in \cref{fig:tilehardness}). Furthermore, $\psi'(1)$ will correspond to the knowledge state of ``knowing nothing'' and $\psi'(n)$ will correspond to the knowledge state of knowing everything.
    Let $Q=\{q_1,q_2,\dots, q_{2(n-1)}\}$, and for the purposes of the following description, define $[q_i,q_j] = \{q_i,q_{i+1},\dots,q_j\}$ for $i<j$.
    Let the vertex $\psi'(i)$ correspond to the knowledge state $[q_1,q_{2(i-1)}]$, thus $\psi'(1)$ corresponds to the knowledge state $\emptyset$.
    Now consider an edge $\{v,w\}\in E(G')$; let $i<j$ be such that $\psi'(i)=v$ and $\psi'(j)=w$.
    We subdivide the edge $\{v,w\}$ $2(j-i)-1$ times. We call these subdivision vertices $u_1,u_2,\dots,u_{2(j-i)-1}$ ordered along the path from $v$ to $w$.
    Let $u_\ell$ correspond to the knowledge state $A\cup B$ where $A=[q_1,q_{2(i-1)}]$ and $B=[q_{2(j-1)-\ell+1},q_{2(j-1)}]$.
    After processing all edges we are left with the learning space graph $G(Q,\mathcal{K})$. 
    Note also that the union of all shortest paths from $\psi'(1)$ to $\psi'(n)$ in $G(Q,\mathcal{K})$ is exactly $G(Q,\mathcal{K})$.
    Now let $(S,\mathcal{C},T,\sigma)$ be an OPD instance such that $S=\{s\}$, $\mathcal{C}=\{c\}$, and $T=\{t_0,t_1\}$ with $t_0(s)=t_1(s)=c$, and $\sigma=(c)$.
    Furthermore, let $\alpha(\mathcal{K})=\{c\}$.
    There is a unique combinatorial layout of this OPD instance. 
    Now consider the possibilistic learning tile $T$ that corresponds to $(S,\mathcal{C},T,\sigma)$, $\alpha$, and $G(Q,\mathcal{K})$. The tile consists of the graph $G(Q,\mathcal{K})$, with $\psi'(1)$ connected to the left wall and $\psi'(n)$ connected to the right wall.
    We claim that $T$ has an ordinal panel tile drawing with at most $k$ crossings if and only if $G'$ has a drawing with at most $k$ crossings, which will complete the proof. We argue both directions.

    ``$\Rightarrow$'': This direction is trivial as the graph corresponding to $T$ contains as induced subgraph a subdivision of the graph $G'$.

    ``$\Leftarrow$'': Consider a drawing of $G'$ with at most $k$ crossings. We can assume that $\psi'(1)$ and $\psi'(n)$ lie on the outer face. Thus, it is easy to connect them to the left and right wall without introducing any new crossings.
\end{proof}

\begin{theorem}\label{thm:exacttilehardness}
    Given $k\in\mathbb{N}$, it is \NP-complete to decide whether an exact learning tile has an ordinal panel drawing with at most $k$ crossings, even for a single subject, single tile and a single category.
\end{theorem}
\begin{proof}
    \NP-membership is argued as in the proof of \cref{thm:posstilehardness}.

    We again reduce from the same problem as in \cref{thm:posstilehardness} with instances $(G,k)$ having the same properties.
    Let $G'$, $G(Q,\mathcal{K})$, $\psi$, $n$, and $\psi'$ be obtained as in \cref{thm:posstilehardness}.
    Let now the edges in $G(Q,\mathcal{K})$ be $\{e_1,e_2,\dots,e_m\}$.
    Let $(S,\mathcal{C}, T, \sigma)$ be an OPD instance with $S=\{s_1,s_2,\dots, s_m\}$, $\mathcal{C}=\{c\}$, $T=\{t_0,t_1\}$ with $t_0(S)=t_1(S)=\{c\}$, and $\sigma=(c)$. Let $\alpha(\mathcal{K})=\{c\}$. For $i\in [m]$, let $P_{s_i,t_1}$ be an arbitrary path in $G(Q,\mathcal{K})$ from $\psi'(1)$ to $\psi'(n)$ visiting the edge $e_i$. It is easy to see that this path always exists as $G'$ is orderable.
    Let $\pi_0$ and $\pi_1$ be two arbitrary permutations of the subjects $S$.
    Let $T$ be the exact learning tile corresponding to $(S,\mathcal{C},T, \sigma)$, $\alpha$, $G(Q,\mathcal{K})$, the paths $P_{s_i,t_1}$, and $(\pi_0,\pi_1)$.
    The tile consists of the graph $G(Q,\mathcal{K})$, with $\psi'(1)$ connected to the left wall with $m$ edges, and $\psi'(n)$ connected to the right wall with $m$ edges. 
    We thus claim that $T$ has an ordinal panel tile drawing with at most $k$ crossings if and only if $G'$ has a drawing with at most $k$ crossings.

    ``$\Rightarrow$'': This direction is trivial as the graph corresponding to $T$ contains as induced subgraph a subdivision of the graph $G'$.

    ``$\Leftarrow$'': This direction is also similar to \cref{thm:posstilehardness}: Consider a drawing of $G'$ with $\psi'(1)$ and $\psi'(n)$ on the outer face. It is now easy to connect $\psi'(1)$ to left wall with $m$ edges and $\psi'(n)$ to the right wall with $m$ edges, without introducing any new crossings.
\end{proof}

\section{Optimal ordinal panel data drawings can be obtained in polynomial time}\label{sc:polynomial}



We consider 
drawings of OPD instances (see \cref{fig:ordinapanelexample}), where subjects are represented by $x$-monotone 
\emph{subject curves}, and each test $t\in T$ is represented by a specific $x$-coordinate $x_t$. 
Further, $x_{t_i}<x_{t_j}$ for $i<j$, and for each $t\in T$: 
\begin{itemize}
    \item[(1)] all $y$-coordinates of the subject curves are distinct at $x_t$, and
    \item[(2)] the $y$-coordinate of subject curve $s$ at $x_t$ is less than the $y$-coordinate of subject curve $s'$ at $x_t$ if $t(s)\prec_{\sigma}t(s')$.
\end{itemize}
\begin{figure*}[htb]
    \centering
    \subfigure[]{ \usetikzlibrary{arrows.meta, arrows, calc}
\begin{tikzpicture}
\begin{scope}[scale=0.27]
\draw[thick, color=violet]  (0,0) rectangle (1,14);
\coordinate(A1) at(0, 2);
\coordinate(B1) at(1, 2);
\draw[thick, color=violet](A1)--(B1);
\coordinate (A) at(0, 4);
\coordinate(B) at (1,4);
\draw[thick, color=violet](A)--(B);
\coordinate (A2) at(0, 6);
\coordinate(B2) at (1,6);
\draw[thick, color=violet](A2)--(B2);
\coordinate (C) at (0, 8);
\coordinate (D) at (1, 8);
\draw[thick, color=violet] (C)--(D);
\coordinate (A3) at(0, 10);
\coordinate(B3) at (1,10);
\draw[thick, color=violet](A3)--(B3);
\coordinate (C1) at (0, 12);
\coordinate (D1) at (1, 12);
\draw[thick, color=violet] (C1)--(D1);
\node(v1) at (-0.5,9.5) {};
\node(v2) at (-0.5,5.5) {};
\draw[thick, color=violet, {latex[scale=1]}-](v1)--(v2);

\node[draw=none, text=violet, left] at (-0.5,7.5){\Large $\sigma$};

\draw[thick, color=violet]  (6,0) rectangle (7,14);
\draw[thick, color=violet] (6,2)--(7,2);
\draw[thick, color=violet] (6,4)--(7,4);
\draw[thick, color=violet] (6,6)--(7,6);
\draw[thick, color=violet] (6,8)--(7,8);
\draw[thick, color=violet](6,10)--(7,10);
\draw[thick, color=violet](6,12)--(7,12);
\node[circle, fill=cyan!20!green,  inner sep=1pt](s1) at (6.5,1) {};
\node[circle, fill=cyan!20!green,  inner sep=1pt](s2) at (6.5,3) {};
\node[circle, fill=cyan!20!green,  inner sep=1pt](s3) at (6.5,5) {};
\node[circle, fill=cyan!20!green,  inner sep=1pt](s4) at (6.5,6.5) {};
\node[circle, fill=cyan!20!green,  inner sep=1pt](s5) at (6.5,7) {};
\node[circle, fill=cyan!20!green,  inner sep=1pt](s6) at (6.5,9) {};
\node[circle, fill=cyan!20!green,  inner sep=1pt](s7) at (6.5,11) {};
\node[circle, fill=cyan!20!green,  inner sep=1pt](s8) at (6.5,12.5) {};
\node[circle, fill=cyan!20!green,  inner sep=1pt](s9) at (6.5,13) {};
\node(v3) at (6.5,0) {};
\node(v4) at (6.5,-2.5) {};
\draw[thick, color=cyan!20!green, -{latex[scale=0.8]}](v3)--(v4);
\node[draw=none, text=cyan!20!green, right] at (6.5, -1){ $\pi_1$};
\node[draw=none, text= black, above] at (6.5, 14){ $x_1$};

\draw[thick, color=violet]  (8,0) rectangle (9,14);
\draw[thick, color=violet] (8,2)--(9,2);
\draw[thick, color=violet] (8,4)--(9,4);
\draw[thick, color=violet] (8,6)--(9,6);
\draw[thick, color=violet] (8,8)--(9,8);
\draw[thick, color=violet](8,10)--(9,10);
\draw[thick, color=violet](8,12)--(9,12);
\node [circle, fill=cyan!20!green,  inner sep=1pt](t1) at (8.5,1) {};
\node [circle, fill=cyan!20!green,  inner sep=1pt](t2) at (8.5,3) {};
\node [circle, fill=cyan!20!green,  inner sep=1pt](t3) at (8.5,5) {};
\node [circle, fill=cyan!20!green,  inner sep=1pt](t4) at (8.5,6.5) {};
\node [circle, fill=cyan!20!green,  inner sep=1pt](t5) at (8.5,7) {};
\node [circle, fill=cyan!20!green,  inner sep=1pt](t6) at (8.5,8.5) {};
\node [circle, fill=cyan!20!green,  inner sep=1pt](t7) at (8.5,9) {};
\node [circle, fill=cyan!20!green,  inner sep=1pt](t8) at (8.5,11) {};
\node [circle, fill=cyan!20!green,  inner sep=1pt] (t9)at (8.5,13) {};
\node(v5) at (8.5,0) {};
\node(v6) at (8.5,-2.5) {};
\draw[thick, color=cyan!20!green, -{latex[scale=0.8]}](v5)--(v6);
\node[draw=none, text=cyan!20!green, right] at (8.5, -1){ $\pi_2$};
\node[draw=none, text= black, above] at (8.5, 14){$x_2$};

\draw[thick, color=violet]  (10,0) rectangle (11,14);
\draw[thick, color=violet] (10,2)--(11,2);
\draw[thick, color=violet] (10,4)--(11,4);
\draw[thick, color=violet] (10,6)--(11,6);
\draw[thick, color=violet] (10,8)--(11,8);
\draw[thick, color=violet](10,10)--(11,10);
\draw[thick, color=violet](10,12)--(11,12);
\node [circle, fill=cyan!20!green,  inner sep=1pt](w1) at (10.5,2.5) {};
\node [circle, fill=cyan!20!green,  inner sep=1pt](w2) at (10.5,3) {};
\node [circle, fill=cyan!20!green,  inner sep=1pt](w3) at (10.5,5) {};
\node [circle, fill=cyan!20!green,  inner sep=1pt](w4) at (10.5,7) {};
\node [circle, fill=cyan!20!green,  inner sep=1pt](w5) at (10.5,8.5) {};
\node [circle, fill=cyan!20!green,  inner sep=1pt](w6) at (10.5,9) {};
\node [circle, fill=cyan!20!green,  inner sep=1pt](w7) at (10.5,10.5) {};
\node [circle, fill=cyan!20!green,  inner sep=1pt](w8) at (10.5,11) {};
\node [circle, fill=cyan!20!green,  inner sep=1pt](w9) at (10.5,13) {};
\node(v7) at (10.5,0) {};
\node(v8) at (10.5,-2.5) {};
\draw[thick, color=cyan!20!green, -{latex[scale=0.8]}](v7)--(v8);
\node[draw=none, text=cyan!20!green, right] at (10.5, -1){ $\pi_3$};
\node[draw=none, text= black, above] at (10.5, 14){$x_3$};

\draw[thick, color=violet]  (12,0) rectangle (13,14);
\draw[thick, color=violet] (12,2)--(13,2);
\draw[thick, color=violet] (12,4)--(13,4);
\draw[thick, color=violet] (12,6)--(13,6);
\draw[thick, color=violet] (12,8)--(13,8);
\draw[thick, color=violet](12,10)--(13,10);
\draw[thick, color=violet](12,12)--(13,12);
\node [circle, fill=cyan!20!green,  inner sep=1pt](u1) at (12.5,1) {};
\node [circle, fill=cyan!20!green,  inner sep=1pt](u2) at (12.5,3) {};
\node [circle, fill=cyan!20!green,  inner sep=1pt](u3) at (12.5,4.5) {};
\node [circle, fill=cyan!20!green,  inner sep=1pt](u4) at (12.5,5) {};
\node [circle, fill=cyan!20!green,  inner sep=1pt](u5) at (12.5,7) {};
\node [circle, fill=cyan!20!green,  inner sep=1pt](u6) at (12.5,9) {};
\node [circle, fill=cyan!20!green,  inner sep=1pt](u7) at (12.5,10.5) {};
\node [circle, fill=cyan!20!green,  inner sep=1pt](u8) at (12.5,11) {};
\node [circle, fill=cyan!20!green,  inner sep=1pt](u9) at (12.5,13) {};
\node(v9) at (12.5,0) {};
\node(v10) at (12.5,-2.5) {};
\draw[thick, color=cyan!20!green, -{latex[scale=0.8]}](v9)--(v10);
\node[draw=none, text=cyan!20!green, right] at (12.5, -1){ $\pi_4$};
\node[draw=none, text= black, above] at (12.5, 14){ $x_4$};
\draw [thick, color=cyan!20!green] (s1) edge (t1);
\draw [thick, color=cyan!20!green] (t1) edge (w2);
\draw [thick, color=cyan!20!green] (w2) edge (u3);
\draw [thick, color=cyan!20!green] (s2) edge (t2);
\draw [thick, color=cyan!20!green] (t2) edge (w1);
\draw [thick, color=cyan!20!green] (w1) edge (u1);
\draw [thick, color=cyan!20!green] (s3) edge (t3);
\draw [thick, color=cyan!20!green] (t3) edge (w4);
\draw [thick, color=cyan!20!green] (w4) edge (u5);
\draw [thick, color=cyan!20!green] (s4) edge (t4);
\draw [thick, color=cyan!20!green] (t4) edge (w3);
\draw [thick, color=cyan!20!green] (w3) edge (u2);
\draw [thick, color=cyan!20!green] (s5) edge (t5);
\draw [thick, color=cyan!20!green] (t5) edge (w5);
\draw [thick, color=cyan!20!green] (w5) edge (u6);
\draw [thick, color=cyan!20!green] (s6) edge (t8);
\draw [thick, color=cyan!20!green] (t8) edge (w9);
\draw [thick, color=cyan!20!green] (w9) edge (u8);
\draw [thick, color=cyan!20!green] (s7) edge (t9);
\draw [thick, color=cyan!20!green] (t9) edge (w7);
\draw [thick, color=cyan!20!green] (w7) edge (u7);
\draw [thick, color=cyan!20!green] (s8) edge (t6);
\draw [thick, color=cyan!20!green] (t6) edge (w8);
\draw [thick, color=cyan!20!green] (w8) edge (u9);
\draw [thick, color=cyan!20!green] (s9) edge (t7);
\draw [thick, color=cyan!20!green] (t7) edge (w6);
\draw [thick, color=cyan!20!green] (w6) edge (u4);
\end{scope}
\end{tikzpicture}}\hspace{0.15\textwidth}%
    \subfigure[]{\usetikzlibrary{arrows.meta, arrows, calc}
\begin{tikzpicture}
\begin{scope}[scale=0.27]
\draw[thick, color=violet]  (6,0) rectangle (7,14);
\draw[thick, color=violet] (6,2)--(7,2);
\draw[thick, color=violet] (6,4)--(7,4);
\draw[thick, color=violet] (6,6)--(7,6);
\draw[thick, color=violet] (6,8)--(7,8);
\draw[thick, color=violet](6,10)--(7,10);
\draw[thick, color=violet](6,12)--(7,12);
\node[circle, fill=cyan!20!green,  inner sep=1pt](s1) at (6.5,1) {};
\node[circle, fill=cyan!20!green,  inner sep=1pt](s2) at (6.5,3) {};
\node[circle, fill=cyan!20!green,  inner sep=1pt](s3) at (6.5,5) {};
\node[circle, fill=cyan!20!green,  inner sep=1pt](s4) at (6.5,6.5) {};
\node[circle, fill=cyan!20!green,  inner sep=1pt](s5) at (6.5,7) {};
\node[circle, fill=cyan!20!green,  inner sep=1pt](s6) at (6.5,9) {};
\node[circle, fill=cyan!20!green,  inner sep=1pt](s7) at (6.5,11) {};
\node[circle, fill=cyan!20!green,  inner sep=1pt](s8) at (6.5,12.5) {};
\node[circle, fill=cyan!20!green,  inner sep=1pt](s9) at (6.5,13) {};
\node(v3) at (6.5,0) {};
\node(v4) at (6.5,-2.5) {};
\draw[thick, color=cyan!20!green, -{latex[scale=0.8]}](v3)--(v4);
\node[draw=none, text=cyan!20!green, right] at (6.5, -1){$\pi_1$};
\node[draw=none, text= black, above] at (6.5, 14){$x_1$};

\draw[thick, color=violet]  (8,0) rectangle (9,14);
\draw[thick, color=violet] (8,2)--(9,2);
\draw[thick, color=violet] (8,4)--(9,4);
\draw[thick, color=violet] (8,6)--(9,6);
\draw[thick, color=violet] (8,8)--(9,8);
\draw[thick, color=violet](8,10)--(9,10);
\draw[thick, color=violet](8,12)--(9,12);
\node [circle, fill=cyan!20!green,  inner sep=1pt](t1) at (8.5,1) {};
\node [circle, fill=cyan!20!green,  inner sep=1pt](t2) at (8.5,3) {};
\node [circle, fill=cyan!20!green,  inner sep=1pt](t3) at (8.5,5) {};
\node [circle, fill=cyan!20!green,  inner sep=1pt](t4) at (8.5,6.5) {};
\node [circle, fill=cyan!20!green,  inner sep=1pt](t5) at (8.5,7) {};
\node [circle, fill=cyan!20!green,  inner sep=1pt](t6) at (8.5,8.5) {};
\node [circle, fill=cyan!20!green,  inner sep=1pt](t7) at (8.5,9) {};
\node [circle, fill=cyan!20!green,  inner sep=1pt](t8) at (8.5,11) {};
\node [circle, fill=cyan!20!green,  inner sep=1pt] (t9)at (8.5,13) {};
\node(v5) at (8.5,0) {};
\node(v6) at (8.5,-2.5) {};
\draw[thick, color=cyan!20!green, -{latex[scale=0.8]}](v5)--(v6);
\node[draw=none, text=cyan!20!green, right] at (8.5, -1){ $\pi_2$};
\node[draw=none, text= black, above] at (8.5, 14){ $x_2$};

\draw[thick, color=violet]  (10,0) rectangle (11,14);
\draw[thick, color=violet] (10,2)--(11,2);
\draw[thick, color=violet] (10,4)--(11,4);
\draw[thick, color=violet] (10,6)--(11,6);
\draw[thick, color=violet] (10,8)--(11,8);
\draw[thick, color=violet](10,10)--(11,10);
\draw[thick, color=violet](10,12)--(11,12);
\node [circle, fill=cyan!20!green,  inner sep=1pt](w1) at (10.5,2.5) {};
\node [circle, fill=cyan!20!green,  inner sep=1pt](w2) at (10.5,3) {};
\node [circle, fill=cyan!20!green,  inner sep=1pt](w3) at (10.5,5) {};
\node [circle, fill=cyan!20!green,  inner sep=1pt](w4) at (10.5,7) {};
\node [circle, fill=cyan!20!green,  inner sep=1pt](w5) at (10.5,8.5) {};
\node [circle, fill=cyan!20!green,  inner sep=1pt](w6) at (10.5,9) {};
\node [circle, fill=cyan!20!green,  inner sep=1pt](w7) at (10.5,10.5) {};
\node [circle, fill=cyan!20!green,  inner sep=1pt](w8) at (10.5,11) {};
\node [circle, fill=cyan!20!green,  inner sep=1pt](w9) at (10.5,13) {};
\node(v7) at (10.5,0) {};
\node(v8) at (10.5,-2.5) {};
\draw[thick, color=cyan!20!green, -{latex[scale=0.8]}](v7)--(v8);
\node[draw=none, text=cyan!20!green, right] at (10.5, -1){$\pi_3$};
\node[draw=none, text= black, above] at (10.5, 14){$x_3$};

\draw[thick, color=violet]  (12,0) rectangle (13,14);
\draw[thick, color=violet] (12,2)--(13,2);
\draw[thick, color=violet] (12,4)--(13,4);
\draw[thick, color=violet] (12,6)--(13,6);
\draw[thick, color=violet] (12,8)--(13,8);
\draw[thick, color=violet](12,10)--(13,10);
\draw[thick, color=violet](12,12)--(13,12);
\node [circle, fill=cyan!20!green,  inner sep=1pt](u1) at (12.5,1) {};
\node [circle, fill=cyan!20!green,  inner sep=1pt](u2) at (12.5,3) {};
\node [circle, fill=cyan!20!green,  inner sep=1pt](u3) at (12.5,4.5) {};
\node [circle, fill=cyan!20!green,  inner sep=1pt](u4) at (12.5,5) {};
\node [circle, fill=cyan!20!green,  inner sep=1pt](u5) at (12.5,7) {};
\node [circle, fill=cyan!20!green,  inner sep=1pt](u6) at (12.5,9) {};
\node [circle, fill=cyan!20!green,  inner sep=1pt](u7) at (12.5,10.5) {};
\node [circle, fill=cyan!20!green,  inner sep=1pt](u8) at (12.5,11) {};
\node [circle, fill=cyan!20!green,  inner sep=1pt](u9) at (12.5,13) {};
\node(v9) at (12.5,0) {};
\node(v10) at (12.5,-2.5) {};
\draw[thick, color=cyan!20!green, -{latex[scale=0.8]}](v9)--(v10);
\node[draw=none, text=cyan!20!green, right] at (12.5, -1){$\pi_4$};
\node[draw=none, text= black, above] at (12.5, 14){$x_4$};
\draw [thick, color=cyan!20!green] (s1) edge (t1);
\draw [thick, color=cyan!20!green] (t1) edge (w2);
\draw [thick, color=cyan!20!green] (w2) edge (u3);
\draw [thick, color=cyan!20!green] (s2) edge (t2);
\draw [thick, color=cyan!20!green] (t2) edge (w1);
\draw [thick, color=cyan!20!green] (w1) edge (u1);
\draw [thick, color=cyan!20!green] (s3) edge (t3);
\draw [thick, color=cyan!20!green] (t3) edge (w4);
\draw [thick, color=cyan!20!green] (w4) edge (u5);
\draw [thick, color=cyan!20!green] (s4) edge (t4);
\draw [thick, color=cyan!20!green] (t4) edge (w3);
\draw [thick, color=cyan!20!green] (w3) edge (u2);
\draw [thick, color=cyan!20!green] (s5) edge (t5);
\draw [thick, color=cyan!20!green] (t5) edge (w5);
\draw [thick, color=cyan!20!green] (w5) edge (u6);
\draw [thick, color=cyan!20!green] (s6) edge (t8);
\draw [thick, color=cyan!20!green] (t8) edge (w9);
\draw [thick, color=cyan!20!green] (w9) edge (u8);
\draw [thick, color=cyan!20!green] (s7) edge (t9);
\draw [thick, color=cyan!20!green] (t9) edge (w7);
\draw [thick, color=cyan!20!green] (w7) edge (u7);

\draw [thick, color=cyan!20!green] (s9) edge (t7);
\draw [thick, color=cyan!20!green] (t7) edge (w8);
\draw [thick, color=cyan!20!green] (w8) edge (u9);
\draw [thick, color=cyan!20!green] (s8) edge (t6);
\draw [thick, color=cyan!20!green] (t6) edge (w6);
\draw [thick, color=cyan!20!green] (u4) edge (w6);
\end{scope}
\end{tikzpicture}}
    
    \caption{The combinatorial and topological layout of
    an ordinal panel data instance. The 7 categories are 
    shown as violet rectangles ordered vertically by 
    $\sigma$. The 9 subjects are assigned categories for 
    each of the 4 tests. The subjects are drawn as green $x$-monotone 
    curves. 
    If subject $s$ is assigned 
    category $C$ for test $t_i$, this is depicted by the 
    subject curve of $s$ passing through the rectangle 
    corresponding to $C$ at $x_i$. The vertical orderings 
    of subjects at $x_i$ are labeled by $\pi_i$ and 
    highlighted by green dots, and form a combinatorial 
    layout. Note that figure (a) shows a layout that is not optimal in terms of a number of crossings, while figure (b) shows an optimal layout that is achieved by swapping the order of the two subjects in the topmost category on the first test.}
    \label{fig:ordinapanelexample}
\end{figure*}
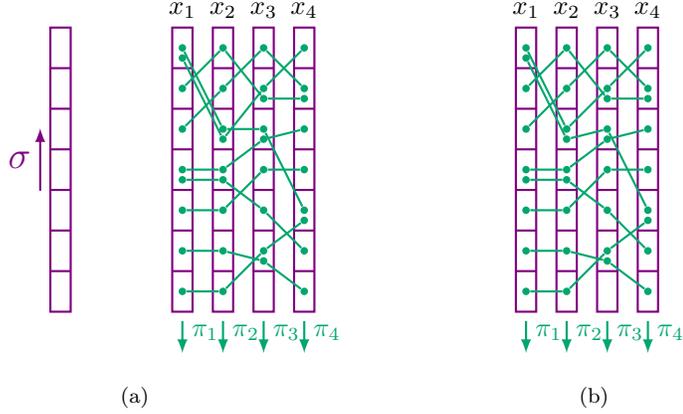

With this definition the subjects assigned to a single category at test $t$ appear consecutively along the vertical line at $x_t$.
We consider the number of crossings in such layouts, that is, the number of crossings between the subject-curves. If we are interested in the minimization of crossings in such layouts, we can determine them combinatorially by considering the vertical orderings of the subjects curves at each $x_t$. Two subject curves cross between $x_{t_i}$ and $x_{t_{i+1}}$ if and only if their vertical order is swapped.
That is why we restate the layout of an OPD instance in a combinatorial way, only representing them by the vertical orders of subject curves along the $x$-coordinates of tests.
The number of crossings of such combinatorial layout $\pi_1,\pi_2,\dots,\pi_m$ can then be computed as the number of $t_i\in T,s,s'\in T$ such that $s\prec_{\pi_i}s'$ but $s\succ_{\pi_{i+1}}s'$. In that case the subject curve for $s$ is below $s'$ for $t_i$ but above $s'$ for $t_{i+1}$, and we say that there is a \emph{crossing involving} $s$ and $s'$ between $t_i$ and $t_{i+1}$.

From a combinatorial layout with the
minimum number of crossings 
we can obtain a topological layout with the 
minimum number of crossings, and vice versa. 
Thus, it is enough to consider combinatorial layouts. 

Jerebic et al.~\cite{JerebicKajzerBokalOpda} observe that
layouts of  ordinal panel data instances can also be seen 
as a graph that belongs to the family of tiled graphs, 
while the drawings induced by a sequence of tests induce 
a tile drawing. This also leads to a polynomial-time 
algorithm to find a combinatorial layout of an OPD instance 
with the minimum number $pcr(S,\mathcal{C},T,\sigma)$ of 
crossings. In this case, we call the layout \emph{optimal}. 
We re-state their results below with detailed proofs 
included, as they were not given in the short conference
paper. The main result of the section is given in 
\cref{theorem2}, which determines an optimal combinatorial layout of an OPD instance in polynomial time.

The basic definitions of tiles have already been given 
in \cref{sec:graphmodels}. Here we add some more 
definitions and results from \cite{pinontoan2003crossing}, 
which will be needed to prove the main theorem. 

\begin{definition}[\hspace{-0.01pt}\cite{pinontoan2003crossing}]
A path $P$ in $G$ is a \DEF{traversing path}
in a tile $(G,L,R)$ if there exist indices 
$j\in\{1,\ldots,|L|\}$ and $k\in\{1,\ldots,|R|\}$ 
such that $P$ is a path from $\lambda_j$ to $\rho_k$ 
and $\lambda_j$ and $\rho_k$ are the only wall vertices 
that lie on $P$. A pair of disjoint traversing paths 
$\{P,P'\}$ is \DEF{aligned} if $j<j'\iff k<k'$, 
and \DEF{twisted} otherwise.
\end{definition}

\begin{proposition}[\hspace{-0.01pt}\cite{pinontoan2003crossing}]
\label{propPR}
Disjointness of the traversing paths of a tile 
$(G,L,R)$ in a twisted pair $\{P,P'\}$  implies 
that some edge of $P$ must cross some edge of 
$P'$ in any tile drawing of $(G,L,R)$.
\end{proposition}

\begin{definition}[\hspace{-0.01pt}\cite{JerebicKajzerBokalOpda}]
Let $(S, \mathcal{C}, T,\sigma)$ be an OPD instance and $i=1,\ldots,m$. We say that a subject $s\in S$:
\begin{wlist}
\item[\textit{(i)}] is below subject $s'$ at test $t_i$, if $t_i(s)\prec_{\sigma}t_i(s')$,
\item[\textit{(ii)}]is above subject $s'$ at test $t_i$, if $t_i(s)\succ_{\sigma}t_i(s')$,
\item[\textit{(iii)}]is level with subject $s'$ at test $t_i$, if $t_i(s)=t_i(s')$,
\item[\textit{(iv)}] overtakes subject $s'$ at test $t_i$, if $s$ is below $s'$ at $t_{i-1}$, but is above $s'$ at $t_i$,
\item[\textit{(v)}] breaks away from subject $s'$ at test $t_i$, if $s$ is level with $s'$ at $t_i$, but is above $s'$ at $t_{i+1}$,
\item[\textit{(vi)}] catches up with subject $s'$ at test $t_i$, if $s$ is below $s'$ at $t_{i-1}$, but is level with $s'$ at $t_i$.
\end{wlist}
\end{definition}

\begin{lemma}[\hspace{-0.01pt}\cite{JerebicKajzerBokalOpda}]
\label{lemma1S} 
 Let $(S, \mathcal{C}, T, \sigma)$ be an OPD instance and $1\leq i\leq m$.
Suppose a subject $s$ overtakes $s'$ at test $t_i$. Then, for every combinatorial layout $\pi_0,\pi_1,\dots,\pi_m$, there is a crossing involving $s$ and $s'$ in any tile drawing of $T_O(\pi_{i-1},\pi_i)$.  
\end{lemma}

\begin{proof}
As already observed, each vertex of $T_O(\pi_{i-1},\pi_i)=(G_O(t_{i-1},t_i),L,R)$ is either a left or right wall vertex and $|S|=|L|=|R|$. Let $L=(\lambda_1,\ldots,\lambda_{|S|})$ stand for the sequence of vertices $\dset{(t_{i-1},s)}{s\in S}$ induced by the linear ordering on $t_{i-1}(S)\subseteq \mathcal{C}$. Vertices of $S(t_{i-1},C)$ are ordered by the permutation $\pi_{i-1}$. Similarly, let $R=(\rho_1,\ldots,\rho_{|S|})$ stand for the sequence of vertices $\dset{(t_i,s)}{s\in S}$ 
induced by the linear ordering on $t_i(S)\subseteq \mathcal{C}$. Vertices of $S(t_i,C)$ are ordered by the permutation $\pi_{i}$. The vertices $(t_{i-1},s)$, $(t_{i-1},s')$, $(t_i,s)$, $(t_i,s')$ can therefore be consecutively denoted as $\lambda_j$, $\lambda_{j'}$, $\rho_k$, $\rho_{k'}$ for some $j,j',k,k'\in\{1,\ldots,|S|\}$. The edge $(t_{i-1},s)(t_i,s)$ is a traversing
path from $\lambda_j$ to $\rho_k$ (path $P$) and the edge $(t_{i-1},s')(t_i,s')$ is a traversing path from $\lambda_{j'}$ to $\rho_{k'}$ (path $P'$). Since subject $s$ overtakes $s'$ at test $t_i$, we have $\lambda_j<_{\sigma}\lambda_{j'}$ (implying $j<j'$) and $\rho_k>_{\sigma}\rho_{k'}$ (implying $k>k'$) for every combinatorial layout of an OPD instance $(S,\mathcal{C},T,\sigma)$ (regardless of the arrangements within $S(t_{i-1},C)$ and $S(t_i,C)$). Hence, $P$ and $P'$ are a twisted pair of disjoint traversing paths (edges). By \cref{propPR}, they cross in any tile drawing of $T_O(\pi_{i-1},\pi_i)$.   
\end{proof}

The crossings characterized by \cref{lemma1S} are called \textit{strongly forced crossings}, in contrast to the \textit{weakly forced crossings} characterized as follows:

\begin{lemma}[\hspace{-0.01pt}\cite{JerebicKajzerBokalOpda}]
\label{lemma2}
 Let $(S, \mathcal{C}, T, \sigma)$ be an OPD instance.
Suppose the following three conditions hold:
\begin{wlist}
\item[(i)] a subject $s$ catches up with $s'$ at test $t_i$,
\item[(ii)] $s$ breaks away from $s'$ at test $t_j$, and
\item[(iii)] for $l=i,\dots,j$, $s$ is level with $s'$ at test $t_l$.
\end{wlist}
Then, for each combinatorial layout $\pi_0,\pi_1,\dots,\pi_m$, there is some $k \in \{i,\ldots,j+1\}$, such that there is a crossing involving $s$ and $s'$ between $t_{k-1}$ and $t_{k}$ in any tile drawing of $T_O(\pi_{i-1},\ldots,\pi_{j+1})$. 
\end{lemma}

\begin{proof}
According to the assumptions, the following applies to $s$ and $s'$:
\begin{wlist}
\item[\textit{(i)}] $t_{i-1}(s)<_{\sigma}t_{i-1}(s')$ and $t_i(s)=t_i(s')$,
\item[\textit{(ii)}] $t_j(s)=t_j(s')$ and $t_{j+1}(s)>_{\sigma}t_{j+1}(s')$,
\item[\textit{(iii)}] for $l=i,\ldots,j$, $t_l(s)=t_l(s')$.
\end{wlist}
Let $T_S(\pi_{i-1},\ldots,\pi_{j+1})$ with $L=(\lambda_1,\ldots,\lambda_{|S|})$ and $R=(\rho_1,\ldots,\rho_{|S|})$ be the join of compatible tiles $T_S(\pi_{i-1},\pi_i),\ldots,T_S(\pi_j,\pi_{j+1})$ and let the vertices $(t_{i-1},s)$, $(t_{i-1},s')$, $(t_{j+1},s)$, $(t_{j+1},s')$ be consecutively denoted as $\lambda_n$, $\lambda_{n'}$, $\rho_m$, $\rho_{m'}$ for some $n,n',m,m'\in\{1,\ldots,|S|\}$. Then, the paths 
$$P:\lambda_n=(t_{i-1},s)(t_i,s)\ldots (t_{j+1},s)=\rho_m\;\; {\rm and}$$
$$P':\lambda_{n'}=(t_{i-1},s')(t_i,s')\ldots (t_{j+1},s')=\rho_{m'}$$ are disjoint traversing paths. Moreover, $\lambda_n<_{\sigma}\lambda_{n'}$ (implying $n<n'$) and $\rho_m>_{\sigma}\rho_{m'}$ (implying $m>m'$). Hence, $P$ and $P'$ are a twisted pair of disjoint traversing paths. By \cref{propPR}, some edge of $P$ must cross some edge of $P'$ in any tile drawing of $T_O(\pi_{i-1},\ldots,\pi_{j+1})$. 
\end{proof}

Strongly and weakly forced crossings cannot be avoided. However, there can be more crossings in a combinatorial layout, but we show that they can always be avoided. We characterize them as follows. 
\begin{definition}
    Let $(S,\mathcal{C}, T, \sigma)$ be an OPD instance where $S=\lset{s_1}{s_n}$, $\mathcal{C}=\lset{C_1}{C_k}$, with a linear ordering $\sigma$ and $T=\lset{t_0}{t_m}$. Let $\pi_0, \ldots , \pi_m$, respectively, determine sequences of subjects $s_1,\ldots s_n$ on tests $t_0, \ldots , t_m$ respectively. 
    We say that a crossing between edges $(s,t_i)(s, t_{i+1})$ and $(s',t_i)(s',t_{i+1})$, where $s, s' \in S$ and $i \in \lset{0}{m-1}$ is \textbf{forward redundant}, if either $t_i(s)=t_i(s')$, $t_{i+1}(s)=t_{i+1}(s')$ and $s$ is before $s'$ in the sequence $\pi_i$ but $s$ is after $s'$ in the sequence $\pi_{i+1}$ (or vice versa), or if $t_{i}(s)\succ_{\sigma}t_{i}(s')$, $t_{i+1}(s)=t_{i+1}(s')$ and $s$ is before $s'$ in sequence $\pi_{i+1}$. 

    We say that a crossing between edges $(s,t_i)(s, t_{i+1})$ and $(s',t_i)(s',t_{i+1})$, where $s, s' \in S$ and $i \in \lset{0}{m-1}$ is \textbf{backward redundant}, if $t_{i+1}(s)\succ_{\sigma}t_{i+1}(s')$, $t_j(s)=t_j(s')$ for every $j\leq i$ and $s$ is before $s'$ in sequence $\pi_{i}$ or vice versa.

    If a crossing between edges $(s,t_i)(s, t_{i+1})$ and $(s',t_i)(s',t_{i+1})$ can be resolved by eliminating a forward (backward) redundant crossing between edges $(s,t_j)(s, t_{j+1})$ and $(s',t_j)(s',t_{j+1})$, where $j < i$, then such crossing is \textbf{forward (backward) induced}.  
\end{definition}

The following lemma establishes that once we have removed every forward and backward redundant crossing, the remaining crossings are always strongly or weakly forced.
\begin{lemma}\label{lm:redundant}
     Let $(S,\mathcal{C}, T, \sigma)$ be an OPD instance where $S=\lset{s_1}{s_n}$, $\mathcal{C}=\lset{C_1}{C_k}$, with a linear ordering $\sigma$ and $T=\lset{t_0}{t_m}$. Let $\pi_0, \ldots , \pi_m$ determine sequences of subjects $s_1,\ldots, s_n$ on tests $t_0, \ldots , t_m$.  If a crossing between edges $(s,t_i)(s, t_{i+1})$ and $(s',t_i)(s',t_{i+1})$, where $s, s' \in S$ and $i \in \lset{0}{m-1}$ is not forward or backward redundant nor forward (backward) induced, it is strongly or weakly forced. 
\end{lemma}

\begin{proof}
    Let $(S,\mathcal{C}, T, \sigma)$ be an OPD instance with $S=\lset{s_1}{s_n}$, $T=\lset{t_0}{t_m}$ and $\mathcal{C}=\lset{C_1}{C_k}$, with a linear ordering $\sigma$.  Let $\pi_0, \ldots , \pi_m$ determine sequences of subjects $s_1,\ldots s_n$ on tests $t_0, \ldots , t_m$. 
    Observe the edges $(s,t_i)(s, t_{i+1})$ and $(s',t_i)(s',t_{i+1})$, where $s, s' \in S$ and $i \in \lset{0}{m-1}$. Let there exist a crossing between them and let the crossing not be forward or backward redundant or forward (backward) induced.
    
    First, let $t_i(s)=t_i(s')$.
    If $t_{i+1}(s)=t_{i+1}(s')$, the crossing can only exist if the order of $s$ and $s'$ is different in $\pi_{i}$ and $\pi_{i+1}$ (either $s$ is before $s'$ in $\pi_{i}$ and after $s'$ in $\pi_{i+1}$, or the other way around), making it a forward redundant crossing, which it isn't by assumption. 
    
    Thus,  $t_{i+1}(s)\neq t_{i+1}(s')$. 
    W.l.o.g., let $t_{i+1}(s)\succ_{\sigma}t_{i+1}(s')$. A crossing occurs, when $s$ is before $s'$ in sequence $\pi_{i}$. By assumption, the crossing isn't backward redundant, thus there exists a $t_\ell$, $0 \leq \ell < i$, such that $t_{\ell}(s)\neq t_{\ell}(s')$. Let $\ell$ be the last such $\ell$, meaning that for every $j \in \lset{\ell +1}{i}$, $t_j(s)=t_j(s')$.  
    
    If $t_{\ell}(s)\succ_{\sigma} t_{\ell}(s')$, there either exists a crossing between edges $(t_\ell, s)(t_{\ell+1},s)$ and $(t_\ell, s')(t_{\ell+1}, s')$ or between edges $(t_j, s)(t_{j+1},s)$, $(t_j, s')(t_{j+1},s')$ for some $j$ such that $\ell < j < i$. In the first case, the crossing between $(t_\ell, s)(t_{\ell+1},s)$ and $(t_\ell, s')(t_{\ell+1}, s')$ is forward redundant, as $t_{\ell}(s)\succ_{\sigma} t_{\ell}(s')$, $t_{\ell+1}(s)= t_{\ell+1}(s')$ and $s$ is before $s'$ in $\pi_{\ell+1}$. In the second case, a crossing between edges $(t_j, s)(t_{j+1},s)$, $(t_j, s')(t_{j+1},s')$ is forward redundant as $t_j(s)=t_j(s')$ and $t_{j+1}(s)=t_{j+1}(s')$. In both cases the crossing between $(s,t_i)(s, t_{i+1})$ and $(s',t_i)(s',t_{i+1})$ is forward induced, contradicting the assumption. 
    
    Thus, $t_{\ell}(s)\prec_{\sigma} t_{\ell}(s')$, where $\ell$ is the highest such that $t_{\ell}(s)\neq t_{\ell}(s')$, meaning that for every $j \in \lset{\ell +1}{i}$, $t_j(s)=t_j(s')$. This remaining instance fulfills all assumptions of Lemma~\ref{lemma2}, making the crossing weakly forced.
    
    It remains to consider $t_i(s)\neq t_i(s')$. W.l.o.g. let $t_i(s)\succ_{\sigma} t_i(s')$.
    If $t_{i+1}(s)\neq t_{i+1}(s')$, we can only get a crossing between the edges $(s,t_i)(s, t_{i+1})$ and\\ $(s',t_i)(s',t_{i+1})$ if $t_{i+1}(s)\prec_{\sigma} t_{i+1}(s')$, which means that subject $s'$ overtook subject $s$ on test $t_{i+1}$. By Lemma~\ref{lemma1S}, the crossing is strongly forced. 
    Let us now assume, that $t_{i+1}(s)= t_{i+1}(s')$. We can only get a crossing if $s$ is before $s'$ in sequence $\pi_{i+1}$. By definition, such crossing is forward redundant. 
\end{proof}
The above lemma is used to prove the following theorem. Essentially, we propose an algorithm that first removes all forward redundant crossings, and then all backward redudundant crossings. 
\begin{theorem}[\hspace{-0.01pt}\cite{JerebicKajzerBokalOpda}]
\label{theorem2}
Let  $(S, \mathcal{C}, T, \sigma)$ be an OPD instance. There exists an algorithm which computes in time $\mathcal{O}(|T|\cdot (|S|+|\mathcal{C}|))$ a combinatorial layout of $(S, \mathcal{C}, T, \sigma)$ for which every crossing is either strongly or weakly forced, i.e., it achieves the minimum number of crossings.
\end{theorem}

\begin{proof}
Let  $(S, \mathcal{C}, T, \sigma)$ be an OPD instance with $S=\lset{s_1}{s_n}$, $T=\lset{t_0}{t_m}$ and $\mathcal{C}=\lset{C_1}{C_k}$, with a linear ordering $\sigma$. We want to find a combinatorial layout $\pi_0,\pi_1,\dots,\pi_m$ for which every crossing is either strongly or weakly forced. 

For every $i \in \lset{0}{m}$ we denote $\pi_{i,C}$ as $\pi_i$ restricted to the category $C \in \mathcal{C}$. To get $\pi_i$ from $\pi_{i,C_1}, \ldots, \pi_{i,C_k}$ for $C_1, \ldots, C_k \in \mathcal{C}$, we define an operation $\star$ on $\pi_{i,C}, \pi_{i, C'}$; $\pi_{i,C} \star \pi_{i, C'}$, which works as a concatenation of permutations. 
Note that concatenation is associative, thus operation $\star$ is well defined on more than two permutations. Then $\pi_i=\pi_{i,C_1} \star \ldots \star \pi_{i, C_k}$.

We say that an ordering of $s \in S$ at test $t_i \in T$ is induced by $\pi_i$, if for every $s, s' \in S$, $s$ is before $s'$ if and only if $\pi_i(s) < \pi_i(s')$.

\cref{algo:OPD_Drawing} depicts pseudocode which gives us a combinatorial layout of subjects of an OPD instance.
\begin{algorithm}[ht]
\caption{Optimal ordinal panel data drawing}
\label{algo:OPD_Drawing}
\KwIn{$S, \mathcal{C}, T, \sigma$}
\KwOut{Optimal ordering of subjects of an OPD instance.}
\Begin{
\For{each $C$ in $\mathcal{C}$ \label{alg_row:first forloop}}{$\pi_{0,C}\gets$ random permutation of $S(t_0,C)$\;
    }
$\pi_0\gets \pi_{0,\sigma(1)}\star \pi_{0,\sigma(2)}\star\dots\star \pi_{0,\sigma(|\mathcal{C}|)}$\;
\For{$i=0, \ldots, m-1$ \label{alg_row:second forloop}}{
    \For{each $C$ in $\mathcal{C}$}{
        $\pi_{i+1,C}\gets \pi_i[S(t_{i+1},C)]$\;
    }
    $\pi_{i+1}\gets \pi_{i+1,\sigma(1)}\star \pi_{i+1,\sigma(2)}\star\dots\star \pi_{i+1,\sigma(|\mathcal{C}|)}$\;
}
\For{$i=m-1,\dots,0$ \label{alg_row:third forloop}}{
    \For{each $C$ in $\mathcal{C}$}{
        $\pi_{i,C}\gets \pi_{i+1}[S(t_{i},C)]$\;
    }
    $\pi_{i}\gets \pi_{i,\sigma(1)}\star \pi_{i,\sigma(2)}\star\dots\star \pi_{i,\sigma(|\mathcal{C}|)}$\;
}
\Return ($\pi_0, \ldots, \pi_m$)}
\end{algorithm}
The algorithm takes an OPD instance. Utilising the for loop in \cref{alg_row:first forloop}, first sequence $\pi_0$ is obtained by joining random permutations of $S(t_0, C)$ in order $\sigma$ on $\mathcal{C}$. Next, we obtain $\pi_1, \ldots, \pi_m$ with the for loop in  \cref{alg_row:second forloop} as follows: For each $C \in \mathcal{C}$, we define a sequence $\pi_{i+1,C}$ as the subsequence $S(t_{i+1},C)$ of $\pi_{i}$. By joining the sequences $\pi_{i+1,\sigma(1)}, \ldots \pi_{i+1,\sigma(|C|)}$, we obtain a sequence $\pi_{i+1}$ for each $i \in \{0, \ldots, m-1\}$.

We show that with this process, we eliminate all forward redundant crossings. 
After executing the for loop in  \cref{alg_row:second forloop}, let there be a forward redundant crossing between edges $(s,t_i)(s, t_{i+1})$ and $(s',t_i)(s',t_{i+1})$. First, let $t_i(s)=t_i(s')$, $s$ be before $s'$ according to $\pi_i$ and let $t_{i+1}(s)=t_{i+1}(s')=C$, where $C$ is some element from $\mathcal{C}$. Then $s, s' \in S(t_{i+1},C)$ and $s'$ is before $s$ by $\pi_{i+1}$. 
According to the algorithm, $\pi_{i+1,C}$ is obtained by taking a subsequence of $\pi_{i}$, containing the elements from $S(t_{i+1},C)$ ordered with the ordering induced by $\pi_{i}$. Thus, as $s$ is before $s'$ by $\pi_i$, it has to be before $s'$ also by $\pi_{i+1,C}$. As the operation $\star$
, which gives us $\pi_{i+1}$, respects the orderings on $\pi_{i+1,C}$ for each $C \in \mathcal{C}$, $s$ is before $s'$ by $\pi_{i+1}$, which is a contradiction. 

It remains to consider $t_{i}(s)\succ_{\sigma}t_{i}(s')$, $t_{i+1}(s)=t_{i+1}(s')=C$, where $C$ is some element from $\mathcal{C}$ and $s$ is before $s'$ by $\pi_{i+1}$. By the definition of the sequence $\pi_i$, $s'$ is before $s$ by $\pi_i$. 
By the algorithm, $\pi_{i+1,C}$ contains the elements from $S(t_{i+1},C)$ with the ordering induced by $\pi_i$. As $s, s' \in S(t_{i+1},C)$, the ordering of $s$ and $s'$ in $\pi_{i+1,C}$ is induced by the ordering in $\pi_i$, meaning $s'$ is before $s$ by $\pi_{i+1,C}$. Following the same argument as before, $s'$ is before $s$ by $\pi_{i+1}$, giving us a contradiction. 

We have thus proven that the process described above eliminates all forward redundant crossings. Note that, by eliminating forward redundant crossings, we also eliminate all forward induced crossings. 

The for loop in \cref{alg_row:third forloop} describes the same process as above, but backwards. We prove that executing this loop eliminates all backward redundant crossings. 
After executing this for loop, let there exist a backward redundant crossing between edges $(s,t_i)(s, t_{i+1})$ and $(s',t_i)(s',t_{i+1})$. 

W.l.o.g., let $t_{i+1}(s)\succ_{\sigma}t_{i+1}(s')$, $t_j(s)=t_j(s')$ for every $j\leq i$ and let $s$ be before $s'$ by $\pi_{i}$.
First we note that $t_{i+1}(s)\succ_{\sigma}t_{i+1}(s')$ implies that $s'$ is before $s$ by $\pi_{i+1}$. We also know, that $t_i(s)=t_i(s')=C$, where $C$ is some element from $\mathcal{C}$, thus $s,s' \in S(t_i, C)$.
By algorithm, $\pi_{i,C}$ contains $s$ and $s'$ in the same order as they were in the ordering induced by $\pi_{i+1}$, meaning that  $s'$ is before $s$ in $\pi_{i, C}$. As the operation $\star$, which gives us $\pi_{i}$, respects the orderings on $\pi_{i,C}$ for all $C \in \mathcal{C}$, $s'$ also has to remain before $s$ in $\pi_{i+1}$, which gives us a contradiction. 

Note that, by eliminating backward redundant crossings, we also eliminate all backward induced crossings.

It remains to consider, if any forward redundant crossing is created during the execution of the for loop in \cref{alg_row:third forloop}. 
Let $(s,t_i)(s, t_{i+1})$ and $(s',t_i)(s',t_{i+1})$ be such edges, that the new ordering on $\pi_i$ caused a forward redundant crossing between them. 
If $t_{i}(s)\succ_{\sigma}t_{i}(s')$, $t_{i+1}(s)=t_{i+1}(s')$, $s$ is after $s'$ by $\pi_{i+1}$ the second loop and third loop won't have any affect on it.
If $t_i(s)=t_i(s')$, $t_{i+1}(s)=t_{i+1}(s')$, by algorithm order of the sequence $\pi_{i+1}$ induces the new order on $\pi_i$. Thus, forward redundant crossing cannot appear. 

By Lemma ~\ref{lm:redundant}, we know that if a crossing is not forward redundant or backward redundant, nor is it forward or backward induced, it has to be either strongly or weakly forced.

We have thus shown that the algorithm described above gives us a combinatorial layout with only strongly or weakly forced crossings. By \cref{lemma1S} and \cref{lemma2}, such combinatorial layout is optimal. It remains to show the stated runtime. For this, it is enough to show that the loop in \cref{alg_row:second forloop} and the loop in \cref{alg_row:third forloop} can be implemented in $\mathcal{O}(|\mathcal{C}|+|S|)$ time. For the loop in \cref{alg_row:second forloop} this can be done as follows: Iterate over the subjects $s$ in the order of $\pi_i$. If $s$ is in category $C_j$, then append $s$ to the initially empty $\pi_{i+1,C_j}$. The for loop in \cref{alg_row:third forloop} can be implemented similarly.
\end{proof}

\section{Extremal instances}\label{sec:extremal}
In this section, we study extremal examples of ordinal panel data instances. 
Let $S(t,C) = \{s \in S \mid t(s)=C\}$ be the set of subjects, which are assigned the same category $C\in \mathcal{C}$ by test $t\in T$. 
We investigate the maximum crossing number of  OPD instances $(S,\mathcal{C},T,\sigma)$ over relevant sets of  timestamps.
In a given set $\mathcal{T}$ that contains sequences of tests; we seek those with highest panel crossing number. The study of extremal instances is motivated by understanding the comparison of real world instances to worst-case scenarios of most turbulent processes with least consistent behavior of tested subjects. 

\begin{definition}\label{def:extremal}
    Let $S=\lset{s_1}{s_n}$ be a set of subjects 
    and $\mathcal{C}=\lset{C_1}{C_k}$ a set of categories with ordering $\sigma$. 
    Let $\mathcal{T}=\dset{(t_0,\ldots,t_m)}{\dfnc{t_0,\ldots,t_m}{S}{\mathcal{C}}}$ be 
    a set of sequences of functions assigning categories to
    subjects at timestamps $t_0,\ldots,t_m$.
    Then, we define the $\mathcal{T}$-extremal crossing number of $S,\mathcal{C}$ as 
    \[ecr(S,\mathcal{C},\mathcal{T},\sigma)=\max\dset{pcr(S,\mathcal{C},T,\sigma)}{T\in\mathcal{T}}.\]
\end{definition}
In the following subsections, we first consider $\mathcal{T}$-extremal crossing number for the set $\mathcal{T}$ of all possible tests without restrictions. Next in Subsection \cref{sec:extremalconsistent}, we restrict ourselves to consistent instances, i.e.\ such instances that do not experience any regressions in maturity levels. The generality of definition  \cref{def:extremal} allows for further refinement of the extremality concept, should interesting instances be observed.
\subsection{Extremal general instances}\label{sec:extremal}
In the following Lemma, we compute the extremal crossing number for two tests. Let us simplify $a_i(t)=|S(t,C_i)|$ for $t\in T,C_i\in \mathcal{C}$.
\begin{lemma}\label{spec_max} \sloppy
Let $S=\lset{s_1}{s_n}$, $\mathcal{C}=\lset{C_1}{C_k}$ ordered by $\sigma$. Without loss of generality, we assume that $\sigma(i)=C_i$. 
Let $a_1,\ldots,a_k$  be non-negative integers such that 
$\sum_{i=1}^{k}a_i=n$. For $\dfnc{t}{S}{\mathcal{C}}$, 
define $a_i(t)=|S(t,C_i)|$. 
Let $\mathcal{T}=\dset{(t_1,t_2)}{\dfnc{t_1,t_2}{S}{\mathcal{C}};a_i(t_1)=a_i}$. 
Then, the $\mathcal{T}$-extremal crossing number is achieved when $a_i(t_1)=a_{k-i+1}(t_2)$ and it equals 
\begin{equation}
ecr(S,C,\mathcal{T},\sigma)=\tfrac{1}{2}\sum_{i=1}^{k}a_i(n-a_i).\label{eq:extremaltwotests}
\end{equation}
\end{lemma}
\begin{proof}
    We prove both directions of the equality in \eqref{eq:extremaltwotests}. 

For ``$\le$'', assume an arbitrary $T\in\mathcal{T}$ and consider a combinatorial layout $\pi_1,\pi_2$ of $(S,\mathcal{C},T,\rho)$ with $pcr(S,\mathcal{C},T,\rho)$ crossings. If two subjects $s$ and $s'$ are in the same category for test $t_1$ or $t_2$, then there cannot be a crossing involving $s$ and $s'$ because we could reorder them to reduce the number of crossings, contradicting that we have $pcr(S,\mathcal{C},T,\rho)$ crossings. Thus, each subject $s \in S$ can at most cross with $n-a_i$ subjects in categories different from $t_1(s)$. Factoring in double-counting, this results in $ecr(S,C,\mathcal{T},\rho)\le\tfrac{1}{2}\sum_{i=1}^{k}a_i(n-a_i)$.

For ``$\ge$'', we give $(t_1,t_2)=T\in \mathcal{T}$ with $pcr(S,\mathcal{C},T,\rho)=\tfrac{1}{2}\sum_{i=1}^{k}a_i(n-a_i)$ crossings. Consider $s_j\in S$. Let $i$ be the maximum integer such that $\sum_{\ell=1}^ia_\ell\le j$. Then let $t_1,t_2$ be such that $t_1(s_j)=C_i$ and $t_2(s_j)=C_{k-i+1}$. Note that $(S,\mathcal{C},\{t_1,t_2\})$ has exactly $\tfrac{1}{2}\sum_{i=1}^{k}a_i(n-a_i)$ strongly forced crossings. Furthermore $a_i(t_1)=a_{k-i+1}(t_2)$ for $i=1,\dots,m$. This completes the proof.
\end{proof}

The following two statements are used to prove the main theorem of the section, \cref{nondevisable}.

\begin{lemma} \label{convex}
    $f : \mathbb{R} \rightarrow \mathbb{R}$, $f(x)=\frac{x(x-1)}{2}$ is a strictly convex function.  
\end{lemma}
\begin{proof}
    We use a folklore characterization of convexity, which claims: If $f(x)$ is twice differentiable on an interval $I$, then $f(x)$ is strictly convex on $I$ if and only if $f''(x) > 0$. Hence, $f(x)=\frac{x(x-1)}{2}$ is strictly convex, as $f''(x)=1$. 
\end{proof}
\begin{proposition}[\hspace{-0.01pt}\cite{BOKAL2012460}]\label{partition}
 Let $Z \subset \mathbb{Z} ^n$ be a set of positive n-element partitions of an integer r. Furthermore, let $f : \mathbb{R} \rightarrow \mathbb{R}$ be any strictly convex function, and $cost_f (z) = \sum\limits_{i=1}^{n} f (z_i)$ for $z \in Z$. Then $z' \in Z$ is a minimum of $cost_f$ if and only if $(\max_{i=1}^n \{z_i\}-\min_{i=1}^n \{z_i\})$ is minimum over all $z \in Z$.
\end{proposition}
This allows us to compute the extremal crossing number for any number of subjects, categories, and tests.
\begin{theorem}\label{nondevisable}\sloppy
Let  $S=\{s_1,\ldots,s_n\}$ be a set of subjects, which are assigned 
categories from a set $\mathcal{C}=\{C_1,\ldots,C_k\}$ ordered by $\sigma$, and let 
$\mathcal{T}=\dset{(t_0,\ldots,t_m)}{\dfnc{t_0,\dots,t_m}{S}{\mathcal{C}}}$
be the set of all possible outcomes of tests at timestamps 
$t_0,\ldots,t_m$. For $i=1,\ldots,k$, $j=0,\ldots,m$, 
define $a_i(t_j)=|S(t_j,C_i)|$. Finally, let  $n=xk+y$
with $x,y\in\mathbb{N}_0$ and $0\leq y<k$. Then, 
\[ ecr(S,\mathcal{C},\mathcal{T},\sigma)=\tfrac{m}{2}(kx(n-x)+y(n-2x-1)),\]
which is achieved when for $i=1,\ldots,k$, $j=0,\ldots,m$,
we have $a_i(t_j)=a_{k-i+1}(t_{j+1})$, with $k-y$ sets $S(t_j,C_i)$ containing 
$x$ subjects and $y$ sets $S(t_j,C_i)$ containing 
$x+1$ subjects. 
\end{theorem}
\begin{proof}
    Let us consider the extremal number of crossings on a pair of consecutive timestamps $t, t' \in T$.

We start with ``$\le$''. For any optimal layout of an OPD instance, we can assume w.l.o.g.\ that weakly forced crossings appear ``as late as possible''. That is, assume there is a weakly forced crossing between subject $s$ and subject $s'$ resulting from $s$ catching up with $s'$ at test $t_i$, breaking away from $s'$ at $t_j$, and being level with $s'$ at test $t_l$ for $l=i,\dots,j$. Then we can assume that the crossing appears between test $t_{j}$ and test $t_{j+1}$ where $s$ and $s'$ are in different categories in $t_{j}$. We can assume this because the algorithm for computing a combinatorial layout with $pcr(S,\mathcal{C},T,\rho)$ crossings discussed in \cref{sc:polynomial} can be implemented such that it computes layouts with this property: The algorithm first computes an initial ordering $\pi_1$ that allows for the least amount of crossings. Then, given ordering $\pi_i$, the algorithm computes $\pi_{i+1}$ greedily such that the fewest number of crossings are created. 
 Hence, we can assume that the maximum number of crossings between $n$ subjects occurs when each subject crosses all subjects that are in different categories for test $t'$. With any category containing $a > 1$ subjects, $\frac{a(a-1)}{2}$ crossings can be avoided. The maximum number of crossings between subjects from categories $\{C_1,\ldots,C_k\}$ is achieved when this loss is minimal. 
    Let $Z \subset \mathbb{Z} ^k$ be a set of size vectors of all possible partitions of $n$ elements into $k$ categories. 
     We observe that the grouping of subjects into $k$ categories, interpreted as a partition $(z_{1}, \ldots z_{k})$, results in a loss of $\sum\limits_{i=1}^{k}\frac{z_{i}(z_{i-1})}{2}$ crossings. \cref{convex} implies that $f(z_{i})=\frac{z_{i}(z_{i-1})}{2}$ is a strictly convex function. By \cref{partition}, the sum $\sum_{i=1}^{k}\frac{z_{i}(z_{i-1})}{2}$ achieves a minimum when 
     \[z_i=\begin{cases}
        x+1 & \text{; } \forall i \in \{1, \ldots y\}\\
        x & \text{; } \forall i \in \{y+1, \ldots k\}
    \end{cases},\text{ for }n=xk + y, \; 0 \leq y < k.\]
     
We note that the above is true for each $t_j, t_{j+1} \in T$, where $j \in \{0, \ldots m-1\}$. As the number of crossings for timestamps $t_0,\ldots t_m$ is a sum of crossings on consecutive pairs of timestamps, it follows that the proposed distribution is the one yielding an extremal number of crossings on $t_0,\ldots t_m$. We have proven that the $\mathcal{T}$-extremal number of crossings is achieved when, for $i=1,\ldots,k$, $j=0,\ldots,m$,
we have $a_i(t_j)=a_{k-i-1}(t_{j+1})$, with $k-y$ sets $S(t_j,C_i)$ containing 
$x$ subjects and $y$ sets $S(t_j,C_i)$ containing 
$x+1$ subjects. We use \cref{spec_max} to calculate $ecr(S, \mathcal{C}, \mathcal{T},\rho)$:
\begin{equation*}
\begin{multlined}
    \tfrac{m-1}{2}\sum_{i=1}^{k}a_i(n-a_i)= \tfrac{m-1}{2}(\sum_{i=1}^{y}(x+1)(n-x-1) + \sum_{i=y+1}^{k}x(n-x))= \\ \frac{m-1}{2}(y(x+1)(n-x-1)+(k-y)x(n-x))= \\ \frac{m-1}{2}(yxn-yx^2-2yx+yn-y+kxn-yxn-kx^2+yx^2)= \\ \frac{m-1}{2}(kx(n-x)+y(n-2x-1)).     
\end{multlined}
\end{equation*}

For ``$\ge$'', there is a straight-forward construction that achieves this amount of crossings. For test $t_0$, put $x+1$ subjects into categories $1,\dots,y$ and $x$ subjects into categories $y+1,\dots,k$. Then set $t_{i+1}$ to simply be the test assignment that reverses the order of subjects of $t_i$ with respect to the category ordering. For example if for some subject $s$, $t_i(s)=c=\rho(x)$, then $t_{i+1}(s)=c'=\rho(k-x+1)$. Note that the number of strongly forced crossing equals $\frac{m-1}{2}(kx(n-x)+y(n-2x-1))$.
\end{proof}

\subsection{Extremal consistent instances}\label{sec:extremalconsistent}
Next, we consider consistent instances, characterized as follows.
\begin{definition} 
      Let $S=\{s_1,\ldots,s_n\}$ be a set of subjects,
    $\mathcal{C}=\{C_1, \ldots , C_k\}$ a set of categories ordered by $\sigma$, and
    $T=\{t_0, \ldots , t_m\}$ a set of timestamps. 
    We say that $(S,\mathcal{C},T,\sigma)$ is a \emph{consistent ordinal panel data instance} if
    for each timestamp $t_i$, $i \in \{1, \ldots m\}$ and 
    each subject $s \in S$, we have $t_{i-1}(s)\preceq_{\sigma} t_i(s)$. Informally, a subject never 
    gets assigned a smaller category at a later timestamp.
\end{definition}
\begin{definition}\sloppy
    Let $S=\{s_1,\ldots,s_n\}$ be a set of subjects and $\mathcal{C}=\{C_1, \ldots , C_k\}$ a set of categories ordered by $\sigma$. Let $\mathcal{T}_c=\dset{(t_0, \ldots t_m)}{\dfnc{t_0,\ldots, t_m}{S}{\mathcal{C}} \wedge \ \forall \ t_i, i \in \{1, \ldots m\}, \ \forall s \in S: t_{i-1}(s)\preceq_{\sigma} t_i(s) }$ be a set of consistent sequences of tests assigning categories in $\mathcal{C}$ to subjects in $S$. Then, $$ecr(S,\mathcal{C},\mathcal{T}_c,\sigma)=\max\dset{pcr(S,C,T,\sigma)}{T \in \mathcal{T}_c}.$$
\end{definition}
We first prove two upper bounds on the extremal crossing number of consistent instances. 
\begin{lemma}\label{lemma:UpperBound}\sloppy
     Let $S=\{s_1,\ldots,s_n\}$ be a set of subjects, $\mathcal{C}=\{C_1, \ldots , C_k\}$ a set of categories ordered by $\sigma$, and let $\mathcal{T}_c=\dset{(t_0, \ldots t_m)}{\dfnc{t_0,\ldots, t_m}{S}{\mathcal{C}} \wedge \ \forall \ t_i, i \in \{1, \ldots m\}, \ \forall s \in S: t_{i-1}(s)\preceq_{\sigma} t_i(s) }$ be a set of consistent sequences of tests assigning categories in $\mathcal{C}$ to subjects in $S$. Then,
    \begin{enumerate}
        \item [i.)] $ecr(S,\mathcal{C},\mathcal{T}_c, \sigma) \leq m\binom{n}{2}$ and
        \item [ii.)] $ecr(S,\mathcal{C},\mathcal{T}_c, \sigma) \leq (k-2)\binom{n}{2}$.
    \end{enumerate}
\end{lemma}
\begin{proof} \sloppy
 As stated in the Lemma, let $S=\{s_1,\ldots,s_n\}$ be a set of subjects and $\mathcal{C}=\{C_1, \ldots , C_k\}$ a set of categories ordered by $\sigma$. Let $\mathcal{T}_c=\dset{(t_0, \ldots t_m)}{\dfnc{t_0,\ldots, t_m}{S}{\mathcal{C}} \wedge \ \forall \ t_i, i \in \{1, \ldots m\}, \ \forall s \in S: t_{i-1}(s)\preceq_{\sigma} t_i(s) }$.
    \begin{enumerate}
        \item [i.)] Any two subjects can cross at most once per test. There are $\binom{n}{2}$ different pairs of subjects and at most they each cross between each two subsequent tests. As there are $m+1$ tests, we get $ecr(S,\mathcal{C},\mathcal{T}_c, \sigma) \leq m\binom{n}{2}$.
        \item [ii.)] Let $s, s' \in S$ be two subjects. If they cross once, then $k \geq 3$. For each subsequent crossing, at least one more category is needed, hence a pair of subjects crosses at most $k-2$ times. As there are $\binom{n}{2}$ different pairs, $ecr(S, \mathcal{C}, \mathcal{T}_c, \sigma) \leq (k-2)\binom{n}{2}$.
    \end{enumerate}
\end{proof}
In the following lemma, we give a lower bound for the extremal crossing number of a consistent OPD instance. The intuition for this lower bound comes from \cref{nondevisable}, where we argue that the distribution that yields a maximum number of crossings in an optimal drawing is such that each pair of subjects that have assigned different categories cross at each test. In a consistent instance, we replicate that and construct an example of a consistent ordinal panel instance $(S, \mathcal{C},T, \sigma)$ with aforementioned properties. We state the panel crossing number of such instances as a lower bound of extremal crossing number of a consistent OPD instance.
\begin{lemma}\label{lemma:LowerBound}\sloppy
    Let $S=\{s_1,\ldots,s_n\}$ be a set of subjects, $\mathcal{C}=\{C_1, \ldots , C_k\}$ a set of categories ordered by $\sigma$ and let $\mathcal{T}_c=\dset{(t_0, \ldots t_m)}{\dfnc{t_0,\ldots, t_m}{S}{\mathcal{C}} \wedge \ \forall \ t_i, i \in \{1, \ldots m\}, \ \forall s \in S: t_{i-1}(s)\preceq_{\sigma} t_i(s) }$ be a set of consistent sequences of tests assigning categories in $\mathcal{C}$ to subjects in $S$. Let $k'= \max\{\lceil \frac{k}{m+1} \rceil, 2\}$ and $n=k'x+y$, where $0 \leq y < k'$. Then, $$ecr(S, \mathcal{C}, \mathcal{T}_c, \sigma) \geq \frac{m}{2}y(n-x^2+x(n-2)-1)+(k'-y)x(n-x).$$
\end{lemma}
\begin{proof}
We will construct a consistent ordinal panel data instance $(S, \mathcal{C},T, \sigma)$ with the stated number of crossings.
For the test $t_0$, we place the $n$ subjects from $S$ into the first $k'$ categories $C_1$ to $C_{k'}$, such that $k'-y$ of them contain $x$ subjects and the remaining $y$ categories contain $x+1$ subjects. 

For the further construction, we will identify bundles of $k'-1$ categories, respecting the ordering $\sigma$.
For $\ell \in \lset{0}{m}$, we set $g_\ell:=\ell (k'-1)+1$ and $h_\ell:=\ell (k'-1)+k'$.
These numbers $g_\ell$ and $h_\ell$ represent the first and last index of the categories from the $\ell$-th bundle, namely of categories $C_{g_\ell}, C_{g_{\ell+1}}, \ldots, C_{h_\ell}$.
In addition, we have $g_{\ell+1}=h_\ell$. 

Using the bundles, we can construct $t_1, \ldots, t_m$ such that for $\ell \in \lset{1}{m}$ and $j \in [g_\ell,h_\ell]$: $\forall s \in S(t_{\ell-1}, C_j)$, we have $s \in S(t_\ell, C_{h_{\ell-1}+(g_\ell-j)})$.

\begin{figure} [htb]
    \centering
    \input{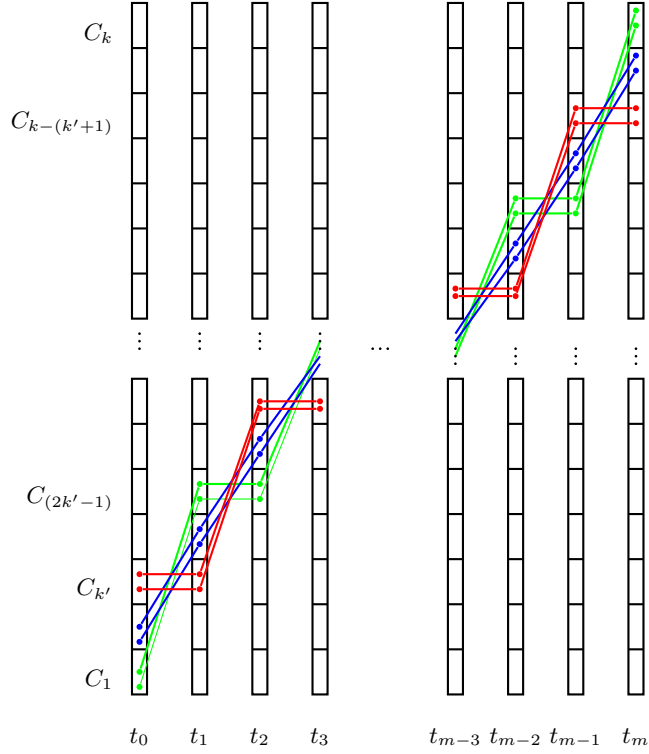}
    \caption{Sketch of the construction of an OPD instance $(S, \mathcal{C}, T, \sigma)$, where $|S|=n$, $|\mathcal{C}|=k$, $|T|=m$ and $k'= max\{\lceil \frac{k}{m+1} \rceil, 2\}$.}
    \label{fig:consistent_construction}
\end{figure}

Note that the construction (\cref{fig:consistent_construction}) is a consistent ordinal panel data instance. On each set of its consecutive timestamps, we create crossings between every such pair of subjects, where the subjects don't share a category on those timestamps. On each set of timestamps, such construction yields 
\begin{eqnarray*}
    \frac{1}{2}\left( \sum\limits_{i=1}^{y} (x+1)(n-x-1)+\sum\limits_{i=y+1}^{k'}x(n-x)\right)&=& \\
    =\frac{y(x+1)(n-x-1)+(k'-y)x(n-x)}{2}& &
\end{eqnarray*}
crossings. For $m+1$ timestamps, we get $\frac{m}{2}(y(x+1)(n-x-1)+(k'-y)x(n-x))$ crossings, which is equal to $\frac{m}{2}(y(n-x^2+x(n-2)-1)+(k'-y)x(n-x))$. We found a consistent ordinal panel data instance $(S, \mathcal{C}, T, \sigma)$, for which $pcr(S, \mathcal{C}, T, \sigma)=\frac{m}{2}(y(n-x^2+x(n-2)-1)+(k'-y)x(n-x))$, thus $ecr(S, \mathcal{C}, \mathcal{T}_c, \sigma) \geq \frac{m}{2}(y(n-x^2+x(n-2)-1)+(k'-y)x(n-x))$.
\end{proof}

The above lower bound gives a slightly complicated term. In the following we show that a weaker lower bound is implied which is exactly one half of the upper bound of \cref{lemma:UpperBound}. 
\begin{corollary}\sloppy
    Let $S=\{s_1,\ldots,s_n\}$ be a set of subjects, $\mathcal{C}=\{C_1, \ldots , C_k\}$ a set of categories ordered by $\sigma$, and let $\mathcal{T}_c=\dset{(t_0, \ldots t_m)}{\dfnc{t_0,\ldots, t_m}{S}{\mathcal{C}} \wedge \ \forall \ t_i, i \in \{1, \ldots m\}, \ \forall s \in S: t_{i-1}(s)\preceq_{\sigma} t_i(s) }$ be a set of consistent sequences of tests assigning categories in $\mathcal{C}$ to subjects in $S$. Then, $$\frac{1}{2}\binom{n}{2}\min(k-2,m) \leq ecr(S, \mathcal{C}, \mathcal{T}_c, \sigma) \leq \binom{n}{2}\min(k-2, m).$$
\end{corollary}
\begin{proof}\sloppy
 Let $S=\{s_1,\ldots,s_n\}$ be a set of subjects, $\mathcal{C}=\{C_1, \ldots , C_k\}$ a set of categories ordered by $\sigma$, and let $\mathcal{T}_c=\dset{(t_0, \ldots t_m)}{\dfnc{t_0,\ldots, t_m}{S}{\mathcal{C}} \wedge \ \forall \ t_i, i \in \{1, \ldots m\}, \ \forall s \in S: t_{i-1}(s)\preceq_{\sigma} t_i(s) }$ be a set of consistent sequences of tests assigning categories in $\mathcal{C}$ to subjects in $S$. Let $k'= max\{\lceil \frac{k}{m+1} \rceil, 2\}$ and $n=k'x+y$, where $0 \leq y < k'$.
Note that $ecr(S, \mathcal{C}, \mathcal{T}_c, \sigma) \leq min(k-2, m)\binom{n}{2}$ by \cref{lemma:UpperBound}. 

By \cref{lemma:LowerBound}, $$\frac{m}{2}(y(n-x^2+x(n-2)-1)+(k'-y)x(n-x)) \leq ecr(S, \mathcal{C}, \mathcal{T}_c, \sigma), $$ thus it is enough to prove $$\frac{1}{2}\binom{n}{2}\min(k-2,m) \leq \frac{m}{2}(y(n-x^2+x(n-2)-1)+(k'-y)x(n-x)).$$
We divide the proof in two cases.
First, let $k-2 \leq m$. This implies $k \leq m+1$ and $\lceil \frac{k}{m+1}\rceil=1$. Thus, $k'=max\{\lceil \frac{k}{m+1} \rceil, 2\}=2$ and $n=2x+y$, where $y \in \{0,1\}$. 

Let $n=2x+1$. By inserting the values in the inequation, we get $$\frac{2x^2+x}{2}(k-2) \leq \frac{2x^2+2x}{2}m,$$ which holds, as $k-2 \leq m$ and $x \geq 0$. 

Now, let $n=2x$. By inserting the values, we get
 $$\frac{2x^2-x}{2}(k-2) \leq x^2m,$$ which, again, holds by $k-2 \leq m$ and $x \geq 0$.
We have thus proven the claim for $k-2 \leq m$.
In the second case, we assume $m < k-2$. In such cases, $k'=\lceil \frac{k}{m+1}\rceil$ and $k' \geq 2$. Let $n=k'x+y$ for $0 \leq y < k'$. By inserting the values, we get $$\frac{m}{2}(2k'xy+y^2+(k'x)^2-k'x^2-2xy-y) \geq \frac{m}{4}((k'x)^2+2k'xy+y^2-k'x-y).$$ 

Assume that the opposite holds. The inequation simplifies to $$2xy(k'-2)+y(y-1)+k'x^2(k'-2)+k'x<0.$$

As $k' \geq 2$ and $x,y \geq 0$, $2xy(k'-2) \geq 0, k'x^2(k'-2) \geq 0$ and $k'x \geq 0$. Thus, the above inequality can hold if and only if $y(y-1)<0$ and $y(y-1) > 2xy(k'-2)+k'x^2(k'-2)+k'x$. Note that $y(y-1) < 0$ implies that $y-1<0$ which can only occur when $y=0$, thus $y(y-1)=0$ and $0 > 2xy(k'-2)+k'x^2(k'-2)+k'x$, which is a contradiction. Thus, $$\frac{m}{2}(2k'xy+y^2+(k'x)^2-k'x^2-2xy-y) \geq \frac{m}{4}((k'x)^2+2k'xy+y^2-k'x-y).$$ and the claim holds also for $k-2 > m$, concluding the proof.
\end{proof}

\section{Random instances}\label{sec:expected}

In this section, we find the expected number of crossings for a certain type of random instances.
We consider the probability model of choosing one category uniformly at random for each subject at each timestamp. More precisely, we assume 
\begin{equation*}
\forall \, 1 \leq i \leq m, 1 \leq j \leq n, 1 \leq \ell \leq k : \mathbb{P}[t_i(s_j) = c_\ell] = \frac{1}{k}.
\end{equation*}

\begin{theorem}\label{thm:random}
Let $n$ be the number of subjects, $k > 1$ the number of categories, and $m+1$ the number of timestamps. If we choose independently for each subject at each timestamp one of the categories uniformly at random, then the expected number of crossings equals
\begin{equation*}
\mathbb{E}(pcr(S,\mathcal{C},T,\sigma)) = \binom{n}{2} \frac{\left( \frac{1}{k} \right)^{m} + m(k-1)-1}{2k}.
\end{equation*}
\end{theorem}

\begin{proof}
We consider first only $2$ subjects and $2$ timestamps. Then we get a strongly forced crossing if and only if, for each of the $2$ timestamps, the subjects get placed into different categories and the order of the respective categories changes from the first to the second timestamp. In other words, the probability of a crossing is
\begin{equation*}
\frac{\left( 1 - \frac{1}{k} \right)^2}{2}
\end{equation*}
because the second subject should not be placed into the same category as the first subjects on both timestamps and in half of the respective cases the order of categories gets inverted.

Next we consider more than $2$ timestamps and the probability that a weakly forced crossing happens between $2$ subjects. Again we need the subjects to get placed into different and inverted categories on the first and last timestamp. In addition, they need to get placed into the same category on every intermediate timestamp. If there are $i$ intermediate timestamps, then the probability of a weakly forced crossing is
\begin{equation*}
\frac{\left( 1 - \frac{1}{k} \right)^2}{2} \left( \frac{1}{k} \right)^{i}.
\end{equation*}

Further note that, with $m+1$ timestamps, there are $m$ different events for a strongly forced crossing and $m-i$ different events for a weakly forced crossing with $i$ intermediate timestamps. The random variable of the total number of crossings is the sum over all those individual crossing events summed over all pairs of subjects. By linearity of expectation we get

\begin{equation*}
\begin{split}
&\mathbb{E}(pcr(S,\mathcal{C},T,\sigma)) = \binom{n}{2} \sum_{i=0}^{m-1} \frac{\left( 1 - \frac{1}{k} \right)^2}{2} \left( \frac{1}{k} \right)^{i} \left( m-i \right)\\
\intertext{which after some transformations}
&=   \binom{n}{2} \frac{\left( 1 - \frac{1}{k} \right)^2}{2} \left[ m\sum_{i=0}^{m-1} \left(\frac{1}{k} \right)^{i}- \sum_{i=0}^{m-1} i  \left(\frac{1}{k} \right)^{i}\right]\\
\end{split}
\end{equation*}
\begin{equation*}
\begin{split}
\intertext{and applying formulas for geometric sums}
&= \binom{n}{2} \frac{\left( 1 - \frac{1}{k} \right)^2}{2} \left[ m \frac{1 - \left( \frac{1}{k} \right)^{m}}{1 - \frac{1}{k}} - \frac{k^{1-m}\left(k^m-km+m-1\right)}{\left(k-1\right)^2} \right]\\
\intertext{simplifies to the desired term}
&= \binom{n}{2} \frac{\left( \frac{1}{k} \right)^{m} + m(k-1)-1}{2k}.\qedhere
\end{split}
\end{equation*}
\end{proof}

\section{Optimal Maturity Models}\label{sec:optimize}

While the previous sections addressed the first challenge
of improved understanding of the structure of 
maturity models, we conclude in this final section
with a discussion of the second challenge, 
developing optimal maturity models from 
the collected data \cite{Bach}. 
While there may be several optimality
criteria, we discuss those that can be observed in 
weakest data collection process and require only
ordinal panel data about the maturing process to 
be finalized.

These optimality criteria ask to 
minimize the number of conflicts of 
the observed data with the ideal model. There are two 
types of conflicts that can happen: either 
the maturity of the observed subjects is inconsistent,
resulting in subjects overtaking and regressing over 
each other during the maturing process, 
or the subjects regress in maturity over time. 
This latter conflict has been studied extensively
as feedback arc set problem \cite{younger1963minimum},
essentially identifying the ordering of categories 
that minimizes the number of regressions 
of subjects. This model is relevant in idealised circumstances with no
natural deterioration process, 
hence the progress of subjects
depends on their internal traits allowing them to at least
maintain if not improve their maturity. We do not
focus on this model in the paper. 

The conflicts of the first model, however, 
result in crossings 
in the drawing of the underlying traces of subjects in 
the tests over time (see \cref{fig:ordinapanelexample} 
for an example of such a drawing). 
The ordering of categories
that minimizes this conflict measure can be interpreted
as considering external environment of the subjects 
to be inducing or preventing progress in maturity, hence
the subjects all progress or regress at the same time.

For both of these criteria, rather than treating the 
order of categories as part of the data
of an ordinal panel data instance $(S,\cC,T,\sigma)$, 
only the panel data instance $(S,\cC,T)$ is given, and 
the ordering $\sigma$ is computed by minimizing the 
number of discrepancies $\sigma$ produces in 
the observed data.
%

We want to find an ordering $\sigma$ of the categories that allows for the least number of crossings. 
If we are given $\sigma$, we can find in polynomial time  the combinatorial 
layout with the minimum number of crossings respecting $\sigma$. The decision problem of finding $\sigma$ that respects a given upper bound of crossings is given below.
\probdef{\panelcrossmin}
{A panel data instance $(S,\mathcal{C},T)$ and an integer $k$.}
{Does there exist $\sigma\in\Pi(\mathcal{C})$ such that $pcr(S,\mathcal{C},T,\sigma)\le k$?}
We prove \NP-completeness of this problem and state an ILP-formulation for the corresponding optimization problem below.


\subsection{\NP-completeness.}
We start by giving \NP-completeness of the decision problem.
\begin{restatable}{theorem}{thmnpcomp}\label{thm:NPcomp}
    \panelcrossmin\ is \NP-complete, even if the number of tests is bounded by 2.
\end{restatable}
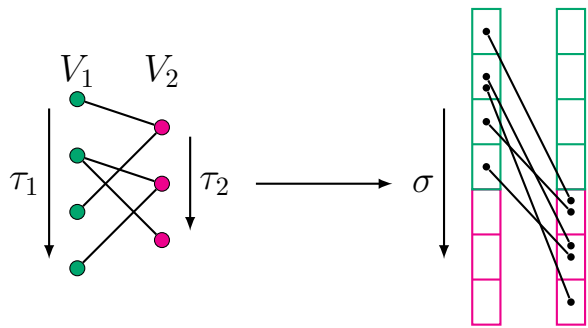
\begin{figure}[t]
    \centering
    \usetikzlibrary{arrows.meta, arrows, calc}
\begin{tikzpicture}
\begin{scope}[scale=0.75]
\node [circle, fill=cyan!20!green, draw=black,  inner sep=2pt] (v1) at (-4.5,4) {};
\node [circle, fill=cyan!20!green, draw=black,  inner sep=2pt] (v4) at (-4.5,3) {};
\node [circle, fill=cyan!20!green,draw=black,    inner sep=2pt] (v7) at (-4.5,1) {};
\node [circle, fill=cyan!20!green, draw=black,   inner sep=2pt] (v3) at (-4.5,2) {};
\node [circle, fill=magenta, draw=black,  inner sep=2pt] (v2) at (-3,3.5) {};
\node [circle, fill=magenta, draw=black,  inner sep=2pt] (v5) at (-3,2.5) {};
\node [circle, fill=magenta, draw=black,  inner sep=2pt] (v6) at (-3,1.5) {};
\draw [thick, color=black] (v1) edge (v2);
\draw [thick, color=black] (v2) edge (v3);
\draw [thick, color=black] (v4) edge (v5);
\draw [thick, color=black] (v4) edge (v6);
\draw [thick, color=black] (v5) edge (v7);
\node[draw=none, text=black] at (-4.5,4.5){\Large $V_1$};
\node[draw=none, text=black] at (-3,4.5){\Large $V_2$};
\node (v1) at (-5,4) {};
\node (v2) at (-5,1) {};
\draw[thick, color=black, -{latex[scale=1.5]}](v1)--(v2);
\node[draw=none, text=black, left] at (-5,2.5){\Large $\tau_1$};
\node (v3) at (-2.5, 3.5){};
\node (v4) at (-2.5, 1.5){};
\draw[thick, color=black, -{latex[scale=1.5]}](v3)--(v4);
\node[draw=none, text=black, right] at (-2.5,2.5){\Large $\tau_2$};
\node (v8) at (-1.5,2.5) {};
\node (v9) at (1.25,2.5) {};
\draw [thick, color=black, -{latex[scale=1.5]}] (v8) edge (v9);
\node (v10) at (2,4) {};
\node  (v11) at (2,1) {};
\draw [thick, color=black, -{latex[scale=1.5]}] (v10) edge (v11);
\node[draw=none, text=black, left] at (2,2.5){\Large $\sigma$};
\draw [thick, magenta] (2.5,0) rectangle (3,2.4);
\draw [thick, magenta] (2.5,0.8)--(3,0.8);
\draw [thick, magenta] (2.5, 1.6)--(3, 1.6);
\draw [thick, cyan!20!green] (2.5,2.4) rectangle (3,5.6);
\draw [thick, cyan!20!green] (2.5,3.2)--(3,3.2);
\draw [thick, cyan!20!green] (2.5, 4)--(3, 4);
\draw [thick, cyan!20!green] (2.5, 4.8)--(3, 4.8);

\draw [thick, magenta] (4,0) rectangle (4.5,2.4);
\draw [thick, magenta] (4,0.8)--(4.5,0.8);
\draw [thick, magenta] (4, 1.6)--(4.5, 1.6);
\draw [thick, cyan!20!green] (4,2.4) rectangle (4.5,5.6);
\draw [thick, cyan!20!green] (4,3.2)--(4.5,3.2);
\draw [thick, cyan!20!green] (4, 4)--(4.5, 4);
\draw [thick, cyan!20!green] (4, 4.8)--(4.5, 4.8);
\node [circle, fill=black,  inner sep=1pt] (v12) at (2.75,5.2) {};
\node [circle, fill=black,  inner sep=1pt] (v14) at (2.75,4.4) {};
\node [circle, fill=black,  inner sep=1pt] (v16) at (2.75,4.2) {};
\node [circle, fill=black,  inner sep=1pt] (v18) at (2.75,3.6) {};
\node [circle, fill=black,  inner sep=1pt] (v20) at (2.75,2.8) {};
\node [circle, fill=black,  inner sep=1pt] (v13) at (4.25,2.2) {};
\node [circle, fill=black,  inner sep=1pt] (v19) at (4.25,2) {};
\node [circle, fill=black,  inner sep=1pt] (v15) at (4.25,1.4) {};
\node [circle, fill=black,  inner sep=1pt] (v21) at (4.25,1.2) {};
\node [circle, fill=black,  inner sep=1pt] (v17) at (4.25,0.4) {};
\draw [thick, color=black] (v12) edge (v13);
\draw [thick, color=black] (v14) edge (v15);
\draw [thick, color=black] (v16) edge (v17);
\draw [thick, color=black] (v18) edge (v19);
\draw [thick, color=black] (v20) edge (v21);
\end{scope}
\end{tikzpicture}
    \caption{To the left side is a bipartite graph corresponding to an instance of \textsc{Bipartite Crossing Number}. The two orderings $\tau_1$ and $\tau_2$ lead to 3 crossings. To the right side is the corresponding instance of \panelcrossmin\ and the ordering $\sigma$ obtained from $\tau_1$ and $\tau_2$ with 3 forced crossings.}
    \label{fig:panelcrossinghardness}
\end{figure}
\begin{proof}\sloppy
    As discussed already in \cref{sc:polynomial}, computing the value $pcr(S,\mathcal{C},T,\sigma)$ for a given permutation $\sigma$ can be done in polynomial time. Hence, \panelcrossmin is in \NP.

    For \NP-hardness, we give a reduction from the \textsc{Bipartite Crossing Number} problem, which is \NP-complete \cite{gareyCrossingNumberNPComplete1983}. The problem takes as input an integer $k$ and a bipartite graph $G=(V_1\cup V_2,E)$ where $V_1$ and $V_2$ are the two partition sets of graph's vertex set, and asks for two permutations $\tau_1\in\Pi(V_1)$ and $\tau_2\in \Pi(V_2)$ such that there are at most $k$ (unordered) pairs of edges that \emph{cross} w.r.t.\ $\tau_1$ and $\tau_2$. Two edges $\{v_1,v_2\},\{w_1,w_2\}\in E$ (here we assume $v_1,w_1\in V_1$) cross w.r.t.\ $\tau_1$ and $\tau_2$ iff.\ 
    \begin{itemize}
        \item $v_1\prec_{\tau_1} w_1$ and $v_2\succ_{\tau_2} w_2$, or
        \item $v_1\succ_{\tau_1} w_1$ and $v_2\prec_{\tau_2} w_2$.
    \end{itemize}
    We define the \panelcrossmin\ instance $(S,\mathcal{C},T)$ such that $S=E$, $\mathcal{C}=V_1\cup V_2$, and $T=\{t_1,t_2\}$ with
    \[t_1(e)=e\cap V_1, t_2(e)=e\cap V_2\]
    for $e\in S$.
    An illustration of this reduction is given in \cref{fig:panelcrossinghardness}. Next, we show that there exist permutations $\tau_1\in \Pi(V_1)$ and $\tau_2\in \Pi(V_2)$ with at most $k$ pairs of edges that cross if and only if there exists a permutation $\sigma\in \Pi(C)$ with $pcr(S,\mathcal{C},T,\sigma)\le k$. We argue both directions.

    ``$\Rightarrow$'': Let $\tau_1\in \Pi(V_1)$ and $\tau_2\in \Pi(V_2)$ with at most $k$ pairs of edges that cross in $G$. We set $\sigma$ as the concatenation of $\tau_1$ and $\tau_2$. That is, $\sigma$ is the ordering such that $C\prec_{\sigma} C'$ iff. 
    \begin{itemize}
        \item $C'\in V_1$ and $C\in V_2$, or
        \item $C,C'\in V_1$ and $C\prec_{\tau_1}C'$, or
        \item $C,C'\in V_2$ and $C\prec_{\tau_2}C'$.
    \end{itemize}
    It is easy to see that two edges $e_1,e_2$ cross w.r.t.\ $\tau_1$ and $\tau_2$ iff. the subjects corresponding to $e_1$ and $e_2$ strongly force a crossing between $t_1$ and $t_2$ (see \cref{fig:panelcrossinghardness}).

    ``$\Leftarrow$'': Let $\sigma\in \Pi(C)$ with $pcr(S,\mathcal{C},T,\sigma)\le k$. We obtain $\tau_1$ and $\tau_2$ by restricting $\sigma$ to $V_1$ and $V_2$ respectively. That is, for $v,w\in V_1$, $v\prec_{\tau_1} w$ iff.\ $v\prec_{\sigma}w$.
    Similarly for $v,w\in V_2$, $v\prec_{\tau_2} w$ iff.\ $v\prec_{\sigma}w$. 
    Again, two edges $e_1,e_2$ cross w.r.t.\ $\tau_1$ and $\tau_2$ iff. the subjects corresponding to $e_1$ and $e_2$ strongly force a crossing between $t_1$ and $t_2$.
\end{proof}\sloppy
We have now established \NP-hardness of deciding \panelcrossmin\ for arbitrary $k$. But it is still open to determine the complexity of deciding the problem for fixed $k$, in particular recognizing planar instances with $k=0$. The problem might be similar to two combinatorial graph drawing problems called  $\mathcal{T}$-level planarity testing and level planarity testing~\cite{DBLP:journals/tcs/AngeliniLBFR15,DBLP:conf/gd/JungerLM98}.
\begin{open}
    What is the computational complexity of deciding planar instances of \panelcrossmin, i.e.\ for $k=0$?
\end{open}

\subsection{Integer program formulation.}
The \NP-hardness given above motivates the following integer linear programming (ILP) formulation for the optimization problem. Note that the following formulation resembles formulations for classic crossing minimization problems in layered graph drawing \cite{DBLP:journals/jea/ChimaniHJM11,gjlm-cmsv-16,DBLP:conf/cocoa/ZhengB07}.

Let $\mathcal{C}^2_{\ne}=\{(C,C')\mid C,C'\in \mathcal{C},C\ne C'\}$.
Observe that we can attribute weakly and strongly forced crossings to specific ordering relations in the category ordering $\sigma$. Namely, consider a strongly forced crossing between subject $s_i$ and $s_j$, between the two tests $t_{\ell}$ and $t_{\ell+1}$. This implies that we have one of two possibilities:
    (i) $t_{\ell}(s_i)\prec_{\sigma}t_{\ell}(s_j)$ and $t_{\ell+1}(s_i)\succ_{\sigma}t_{\ell+1}(s_j)$, or
    (ii) $t_{\ell}(s_i)\succ_{\sigma}t_{\ell}(s_j)$ and $t_{\ell+1}(s_i)\prec_{\sigma}t_{\ell+1}(s_j)$.
Note that we might have that $t_{\ell}(s_i)=t_{\ell+1}(s_j)$ and/or $t_{\ell+1}(s_i)=t_{\ell}(s_j)$.
We say that the pair of pairs $((t_\ell(s_i), t_\ell(s_j)), (t_{\ell+1}(s_i), t_{\ell+1}(s_j)))$ is \emph{responsible} for the crossing between $s_i$ and $s_j$.
This motivates introducing binary variables $x_{C,C'}$ in the ILP for each $(C,C')\in \mathcal{C}^2_{\ne}$. Semantically, $x_{C,C'}$ should be 1 iff.\ $C\prec_\sigma C'$. Further the above condition of having the forced crossing between $s_i$ and $s_j$ at time $\ell$ can be simplified as $x_{t_{\ell}(s_i), t_{\ell}(s_j)}\ne x_{t_{\ell+1}(s_i), t_{\ell+1}(s_j)}$ (note that this might be a tautology if $t_{\ell}(s_i)=t_{\ell+1}(s_j)$ and $t_{\ell+1}(s_i)=t_{\ell}(s_j)$), and we have a crossing iff.\ the `exclusive or' of the two variables is 1.

Let now $\mathcal{C}^4_{<}=\{(C_{\alpha},C_{\beta}), (C_{\gamma}, C_{\delta})\in \mathcal{C}_{\ne}^2\times \mathcal{C}_{\ne}^2\mid (\alpha,\beta)< (\gamma,\delta), \alpha<\beta\}$ and let $((C_{\alpha},C_{\beta}), (C_{\gamma}, C_{\delta}))\in \mathcal{C}^4_{<}$ (the inequalities break symmetries and prevent double counting). We define $\mathrm{sc}((C_{\alpha},C_{\beta}),(C_{\gamma},C_{\delta}))$ as the number of strongly forced crossings for which $((C_{\alpha}, C_{\beta}), (C_{\gamma},C_{\delta}))$ is responsible for when $C_{\alpha}\prec_{\sigma}C_{\beta}$ ($C_{\alpha}\succ_{\sigma}C_{\beta}$) and $C_{\gamma}\succ_{\sigma}C_{\delta}$ ($C_{\gamma}\prec_{\sigma}C_{\delta}$). All of these values can be computed in amortized time $\mathcal{O}(|S|^2\cdot |T|)$ by iterating over all pairs of subjects and following their ``path'' through the categories over the set of tests in increasing order. As we prevent double counting, each strongly forced crossing computed in this way is only added to one value in $\textrm{sc}$.
In a similar way we define $\mathrm{wc}((C_\alpha,C_\beta),(C_\gamma,C_\delta$)) as the number of weakly forced crossings $((C_{\alpha}, C_{\beta}), (C_{\gamma},C_{\delta}))$ is responsible for when $C_{\alpha}\prec_{\sigma}C_{\beta}$ ($C_{\alpha}\succ_{\sigma}C_{\beta}$) and $C_{\gamma}\succ_{\sigma}C_{\delta}$ ($C_{\gamma}\prec_{\sigma}C_{\delta}$).
These values can be computed in a similar way by iterating over all pairs of subjects and following their path through the categories over the tests in increasing order, in this case ignoring tests for which the two subjects belong to the same category.

We introduce another set of variables $y_{p,p'}$ for all $(p,p')\in \mathcal{C}_{<}^4$ with $p\ne p'$. The variable $y_{p,p'}$ shall be 1 if the categories in $p$ are not ordered the same as the categories in $p'$. This leads to the following formulation.

\begin{align}
   \text{minimize}\quad&\sum_{(p,p')\in \mathcal{C}_<^4}(\mathrm{wc}(p,p')+\mathrm{sc}(p,p'))y_{p,p'}\label{eq:objective}\\
   \text{s.t.}\quad&x_{C_i,C_j}=1-x_{C_j,C_i}&1\le i,j\le |\mathcal{C}|,i\ne j\label{eq:antisymmetry}\\
   &0\le x_{C_i,C_j}+x_{C_j,C_k}-x_{C_i,C_k}\le 1&1\le i,j,k\le |\mathcal{C}|,\label{eq:transitivity}\\&&i\ne j,j\ne k,i\ne k\nonumber\\
   & y_{(C_{\alpha},C_{\beta}), (C_{\gamma},C_{\delta})}\ge x_{C_{\alpha},C_{\beta}}-x_{C_{\gamma},C_{\delta}}&((C_{\alpha},C_{\beta}),(C_{\gamma},C_{\delta}))\in \mathcal{C}_<^4\label{eq:xor1}\\
   & y_{(C_{\alpha},C_{\beta}), (C_{\gamma},C_{\delta})}\ge x_{C_{\gamma},C_{\delta}}-x_{C_{\alpha},C_{\beta}}&((C_{\alpha},C_{\beta}),(C_{\gamma},C_{\delta}))\in \mathcal{C}_<^4\label{eq:xor2}
\end{align}
The objective \eqref{eq:objective} equals the total of strongly and weakly forced crossings of the output ordering induced by the $x$-variables.
Constraint \eqref{eq:antisymmetry} ensures antisymmetry of the ordering, while \eqref{eq:transitivity} ensures transitivity.
Constraints \eqref{eq:xor1} and \eqref{eq:xor2} ensure that $z_{(C_{\alpha},C_{\beta}), (C_{\gamma},C_{\delta})}$ is larger than the `exclusive or' of $x_{C_{\alpha},C_{\beta}}$ and $x_{C_{\gamma},C_{\delta}}$. 
Equality is guaranteed by the objective function.
An optimal solution of this formulation immediately gives an ordering $\sigma$ obtained by the $x$-variables that allows for the least amount of forced crossings. As stated before, from this ordering a drawing with this amount of crossings can be found in polynomial time.

Note that the formulation has $\mathcal{O}(|\mathcal{C}|^4)$ variables and $\mathcal{O}(|\mathcal{C}|^4)$ constraints. By only creating variables $y_{p,p'}$ for which $wc(p,p')+sc(p,p')>0$, this upper bound can also be stated as $\mathcal{O}(|\mathcal{C}|^2+|S|^2\cdot |T|)$ for the number of variables and $\mathcal{O}(|\mathcal{C}|^3+|S|^2\cdot |T|)$ for the number of constraints.

\section{Conclusion}
We considered the recently introduced panel crossing number problem~\cite{JerebicKajzerBokalOpda} motivated by maturity model visualization from a graph drawing and crossing minimization perspective.
We studied extremal and expected crossing numbers,
and showed that real-world instances 
are far from  random or extremal cases. 
Further, we proposed two ILP models that solve the \NP-hard crossing and regress minimization problems optimally in less than a second for our real-world datasets.

\section*{Acknowledgements}
The research of Š.~Kajzer, J.~Jerebic, and D. Bokal was supported in part through ARIS grants J1-2452, P1-0297, and P5-0433. J. Orthaber was supported by the Austrian Science Fund (FWF) grant W1230. A.~Dobler and M.~Nöllenburg were supported by the Vienna Science and Technology Fund (WWTF) [10.47379/ICT19035].

The authors would like to acknowledge the Crossing Number Workshop in Strobl (2022) and in Rogla (2023), where significant steps in this research were drafted.

\bibliographystyle{abbrv}
\bibliography{literature.bib}

\begin{thebibliography}{10}

\bibitem{DBLP:journals/tcs/AngeliniLBFR15}
P.~Angelini, G.~D. Lozzo, G.~D. Battista, F.~Frati, and V.~Roselli.
\newblock The importance of being proper: (in clustered-level planarity and {T}-level planarity).
\newblock {\em Theor. Comput. Sci.}, 571:1--9, 2015.

\bibitem{H2020}
G.~Annexes.
\newblock G. technology readiness levels (trl).
\newblock \url{https://ec.europa.eu/research/participants/data/ref/h2020/wp/2014_2015/annexes/h2020-wp1415-annex-g-trl_en.pdf}, 2020.
\newblock Accessed on 4 February 2024.

\bibitem{ArgyriouBKS10}
E.~N. Argyriou, M.~A. Bekos, M.~Kaufmann, and A.~Symvonis.
\newblock On metro-line crossing minimization.
\newblock {\em J. Graph Algorithms Appl.}, 14(1):75--96, 2010.

\bibitem{Bach}
J.~Bach.
\newblock The immaturity of the cmm.
\newblock {\em American Programmer}, 7:13--13, 1994.

\bibitem{bnuw-miecw-07}
M.~Benkert, M.~Nöllenburg, T.~Uno, and A.~Wolff.
\newblock Minimizing intra-edge crossings in wiring diagrams and public transportation maps.
\newblock In M.~Kaufmann and D.~Wagner, editors, {\em Graph Drawing (GD'06)}, volume 4372 of {\em LNCS}, pages 270--281. Springer-Verlag, 2007.

\bibitem{bokal2015degree}
D.~Bokal, M.~Bra{\v{c}}i{\v{c}}, M.~Der{\v{n}}{\'a}r, and P.~Hlin{\v{e}}n{\`y}.
\newblock On degree properties of crossing-critical families of graphs.
\newblock In {\em International Symposium on Graph Drawing}, pages 75--86. Springer, 2015.

\bibitem{BOKAL2012460}
D.~Bokal, B.~Bresar, and J.~Jerebic.
\newblock A generalization of hungarian method and hall's theorem with applications in wireless sensor networks.
\newblock {\em Discret. Appl. Math.}, 160(4-5):460--470, 2012.

\bibitem{bokal2021properties}
D.~Bokal, M.~Chimani, A.~Nover, J.~Schierbaum, T.~Stolzmann, M.~H. Wagner, and T.~Wiedera.
\newblock Properties of large 2-crossing-critical graphs.
\newblock {\em arXiv preprint arXiv:2112.04854}, 2021.

\bibitem{BokalJerebicIntro}
D.~Bokal and J.~Jerebic.
\newblock {\em Modelling states of knowledge to aid navigation in learning spaces}, pages 50--75.
\newblock Cambridge Scholars Publishing, 2024.

\bibitem{DBLP:journals/jea/ChimaniHJM11}
M.~Chimani, P.~Hungerl{\"{a}}nder, M.~J{\"{u}}nger, and P.~Mutzel.
\newblock An {SDP} approach to multi-level crossing minimization.
\newblock {\em {ACM} J. Exp. Algorithmics}, 17(1), 2011.

\bibitem{crawford2021project}
J.~K. Crawford.
\newblock {\em Project management maturity model}.
\newblock CRC Press, 2021.

\bibitem{doignon2014learning}
J.-P. Doignon.
\newblock Learning spaces, and how to build them.
\newblock In {\em Formal Concept Analysis: 12th International Conference, ICFCA 2014, Cluj-Napoca, Romania, June 10-13, 2014. Proceedings 12}, pages 1--14. Springer, 2014.

\bibitem{eppstein2007media}
D.~Eppstein, J.-C. Falmagne, and S.~Ovchinnikov.
\newblock {\em Media theory: interdisciplinary applied mathematics}.
\newblock Springer Science \& Business Media, 2007.

\bibitem{even1976computing}
S.~Even and R.~E. Tarjan.
\newblock Computing an st-numbering.
\newblock {\em Theoretical Computer Science}, 2(3):339--344, 1976.

\bibitem{falmagne2010learning}
J.-C. Falmagne and J.-P. Doignon.
\newblock {\em Learning spaces: Interdisciplinary applied mathematics}.
\newblock Springer Science \& Business Media, 2010.

\bibitem{falmagne2011learning}
J.-C. Falmagne and J.-P. Doignon.
\newblock {\em Learning spaces: Interdisciplinary applied mathematics}.
\newblock Springer-Verlag, 2011.

\bibitem{fp-mcmhatc-13}
M.~Fink and S.~Pupyrev.
\newblock Metro-line crossing minimization: Hardness, approximations, and tractable cases.
\newblock In S.~Wismath and A.~Wolff, editors, {\em Graph Drawing (GD'13)}, volume 8242 of {\em LNCS}, pages 328--339. Springer-Verlag, 2013.

\bibitem{gareyCrossingNumberNPComplete1983}
M.~R. Garey and D.~S. Johnson.
\newblock Crossing {Number} is {NP}-{Complete}.
\newblock {\em SIAM. J. Alg. Discr. Meth.}, 4(3):312--316, 1983.

\bibitem{Gottschalk}
P.~Gottschalk.
\newblock Maturity levels for interoperability in digital government.
\newblock {\em Gov. Inf. Q.}, 26(1):75--81, 2009.

\bibitem{GronemannJLM16}
M.~Gronemann, M.~J{\"{u}}nger, F.~Liers, and F.~Mambelli.
\newblock Crossing minimization in storyline visualization.
\newblock In Y.~Hu and M.~N{\"{o}}llenburg, editors, {\em Graph Drawing (GD'16)}, volume 9801 of {\em LNCS}, pages 367--381. Springer, 2016.

\bibitem{gjlm-cmsv-16}
M.~Gronemann, M.~J{\"{u}}nger, F.~Liers, and F.~Mambelli.
\newblock Crossing minimization in storyline visualization.
\newblock In Y.~Hu and M.~N{\"{o}}llenburg, editors, {\em Graph Drawing and Network Visualization - 24th International Symposium, {GD} 2016, Athens, Greece, September 19-21, 2016, Revised Selected Papers}, volume 9801 of {\em Lecture Notes in Computer Science}, pages 367--381. Springer, 2016.

\bibitem{hn-hda-13}
P.~Healy and N.~S. Nikolov.
\newblock Hierarchical drawing algorithms.
\newblock In R.~Tamassia, editor, {\em Handbook of Graph Drawing and Visualization}, chapter~13, pages 409--454. CRC Press, 2014.

\bibitem{DBLP:journals/jct/Hlineny06a}
P.~Hlinen{\'{y}}.
\newblock Crossing number is hard for cubic graphs.
\newblock {\em J. Comb. Theory, Ser. {B}}, 96(4):455--471, 2006.

\bibitem{hsiao2014analysis}
C.~Hsiao.
\newblock {\em Analysis of panel data}.
\newblock NY: Cambridge university press, 2014.

\bibitem{JerebicKajzerBokalOpda}
J.~Jerebic, {\v S}.~Kajzer, M.~Vogrinec, and D.~Bokal.
\newblock Longitudinal dynamics between linearly ordered classes.
\newblock In S.~Drobne, L.~Zadnik~Stirn, M.~Kljajić~Borštnar, J.~Povh, and J.~Žerovnik, editors, {\em International Symposium on Operational Research in Slovenia (SOR'21)}, page 221–226. Slovenian Society Informatika, Section for Operational Research, 2021.

\bibitem{DBLP:conf/gd/JungerLM98}
M.~J{\"{u}}nger, S.~Leipert, and P.~Mutzel.
\newblock Level planarity testing in linear time.
\newblock In S.~Whitesides, editor, {\em Proc. Graph Drawing and Network Visualization (GD'98)}, volume 1547 of {\em Lecture Notes in Computer Science}, pages 224--237. Springer, 1998.

\bibitem{korte2012greedoids}
B.~Korte, L.~Lov{\'a}sz, and R.~Schrader.
\newblock {\em Greedoids}, volume~4.
\newblock Springer Science \& Business Media, 2012.

\bibitem{KostitsynaNP0S15}
I.~Kostitsyna, M.~N{\"{o}}llenburg, V.~Polishchuk, A.~Schulz, and D.~Strash.
\newblock On minimizing crossings in storyline visualizations.
\newblock In E.~D. Giacomo and A.~Lubiw, editors, {\em Graph Drawing (GD'15)}, volume 9411 of {\em LNCS}, pages 192--198. Springer, 2015.

\bibitem{NASA}
C.~G. Manning.
\newblock Technology readiness levels.
\newblock \url{https://www.nasa.gov/directorates/heo/scan/engineering/technology/technology\_readiness\_level}, 2023.
\newblock Accessed on 12 February 2024.

\bibitem{Mettler}
T.~Mettler and P.~Rohner.
\newblock Situational maturity models as instrumental artifacts for organizational design.
\newblock In V.~K. Vaishnavi and S.~Purao, editors, {\em Conference on Design Science Research in Information Systems and Technology (DESRIST'09)}. {ACM}, 2009.

\bibitem{nonaka1994dynamic}
I.~Nonaka.
\newblock A dynamic theory of organizational knowledge creation.
\newblock {\em Organization science}, 5(1):14--37, 1994.

\bibitem{n-iamlc-10}
M.~Nöllenburg.
\newblock An improved algorithm for the metro-line crossing minimization problem.
\newblock In D.~Eppstein and E.~R. Gansner, editors, {\em Graph Drawing (GD'09)}, volume 5849 of {\em LNCS}, pages 381--392. Springer Berlin Heidelberg, 2010.

\bibitem{CMM}
M.~C. Paulk, B.~Curtis, M.~B. Chrissis, and C.~V. Weber.
\newblock Capability maturity model, version 1.1.
\newblock {\em IEEE software}, 10(4):18--27, 1993.

\bibitem{pinontoan2003crossing}
B.~Pinontoan and R.~B. Richter.
\newblock Crossing numbers of sequences of graphs ii: planar tiles.
\newblock {\em Journal of Graph Theory}, 42(4):332--341, 2003.

\bibitem{poppelbuss2011makes}
J.~P{\"{o}}ppelbu{\ss} and M.~R{\"{o}}glinger.
\newblock What makes a useful maturity model? a framework of general design principles for maturity models and its demonstration in business process management.
\newblock In V.~K. Tuunainen, M.~Rossi, and J.~Nandhakumar, editors, {\em European Conference on Information Systems (ECIS'11)}, page~28, 2011.

\bibitem{Prananto}
A.~Prananto, J.~McKay, and P.~Marshall.
\newblock A study of the progression of e-business maturity in australian smes: Some evidence of the applicability of the stages of growth for e-business model.
\newblock In {\em Pacific Asia Conference on Information Systems (PACIS'03)}, page~5. AISeL, 2003.

\bibitem{SPICE}
M.~Sarshar, R.~Haigh, M.~Finnemore, G.~Aouad, P.~Barrett, D.~Baldry, and M.~Sexton.
\newblock Spice: a business process diagnostics tool for construction projects.
\newblock {\em Engineering, construction and Architectural management}, 7(3):241--250, 2000.

\bibitem{SINNWELL2019557}
C.~Sinnwell, C.~Siedler, and J.~C. Aurich.
\newblock Maturity model for product development information.
\newblock {\em Procedia CIRP}, 79:557--562, 2019.
\newblock 12th CIRP Conference on Intelligent Computation in Manufacturing Engineering, 18-20 July 2018, Gulf of Naples, Italy.

\bibitem{stt-mvuhss-81}
K.~Sugiyama, S.~Tagawa, and M.~Toda.
\newblock Methods for visual understanding of hierarchical system structures.
\newblock {\em {IEEE} Trans. Syst. Man Cybern.}, 11(2):109--125, 1981.

\bibitem{tm-dcosv-12}
Y.~Tanahashi and K.~Ma.
\newblock Design considerations for optimizing storyline visualizations.
\newblock {\em {IEEE} Trans. Vis. Comput. Graph.}, 18(12):2679--2688, 2012.

\bibitem{vegi2023counting}
A.~Vegi~Kalamar.
\newblock Counting traversing hamiltonian cycles in tiled graphs.
\newblock {\em Mathematics}, 11(12):2650, 2023.

\bibitem{vegi2021counting}
A.~Vegi~Kalamar, T.~{\v{Z}}erak, and D.~Bokal.
\newblock Counting hamiltonian cycles in 2-tiled graphs.
\newblock {\em Mathematics}, 9(6):693, 2021.

\bibitem{younger1963minimum}
D.~Younger.
\newblock Minimum feedback arc sets for a directed graph.
\newblock {\em IEEE Transactions on Circuit Theory}, 10(2):238--245, 1963.

\bibitem{DBLP:conf/cocoa/ZhengB07}
L.~Zheng and C.~Buchheim.
\newblock A new exact algorithm for the two-sided crossing minimization problem.
\newblock In A.~W.~M. Dress, Y.~Xu, and B.~Zhu, editors, {\em Combinatorial Optimization and Applications, First International Conference, {COCOA} 2007, Xi'an, China, August 14-16, 2007, Proceedings}, volume 4616 of {\em Lecture Notes in Computer Science}, pages 301--310. Springer, 2007.

\end{thebibliography}

\clearpage
\appendix
\end{document}